\documentclass[submission,copyright,creativecommons]{eptcs}
\providecommand{\event}{EXPRESS~2022} 

\usepackage{style/appendix}
\usepackage{iftex}

\ifpdf
  \usepackage{underscore}         
  \usepackage[T1]{fontenc}        
\else
  \usepackage{breakurl}           
\fi

\usepackage{amsmath}
\usepackage{amsthm}
\usepackage{amssymb}
\usepackage{mathtools}
\usepackage{thm-restate}
\usepackage{stmaryrd}
\SetSymbolFont{stmry}{bold}{U}{stmry}{m}{n}
\usepackage{bm}
\usepackage{tikz}
\usepackage{relsize}
\usetikzlibrary{positioning, shapes, calc}
\usepackage{wrapfig}
\usepackage{calc}
\usepackage{varwidth}
\usepackage{bussproofs}
\usepackage{mathpartir}
\usepackage{xspace}
\usepackage[framemethod=tikz]{mdframed}
\usepackage{centernot}
\usepackage{tikzborder}
\usepackage{suffix}
\usepackage{xcolor}
\usepackage{hyperref}
\usepackage{cleveref}
\usepackage{microtype}

\definecolor{cblRed}{RGB}{154,0,3}
\definecolor{cblBlue}{RGB}{41,59,107}
\definecolor{cblBlueLt}{RGB}{0,144,156}
\definecolor{cblYellow}{RGB}{255,176,0}
\definecolor{cblPink}{RGB}{216,0,106}
\definecolor{cblPurple}{RGB}{25,0,138}
\definecolor{cblPurpleLt}{RGB}{193,99,228}
\definecolor{cblOrange}{RGB}{189,74,0}
\definecolor{cblGreen}{RGB}{47,80,18}

\colorlet{cChanClr}{cblRed}
\colorlet{dChanClr}{cblBlue}
\colorlet{aChanClr}{cblGreen}
\colorlet{bChanClr}{cblPurpleLt}

\hypersetup{
    hypertexnames=false,
    colorlinks=true,
    final,
    citecolor=cblGreen,
    linkcolor=cblRed,
    bookmarksnumbered=true,
    bookmarksopen=true,
    bookmarksopenlevel=1,
    pdfpagemode=UseOutlines
}

\newcommand{\etal}{\emph{et al.}\xspace}

\newcommand{\lGV}{\ensuremath{\lambda^{\mkern-5mu\text{sess}}}\xspace}
\newcommand\ourGV{CGV\xspace}
\newcommand\ourGVlong{Concurrent GV\xspace}

\newcommand{\sff}[1]{\relax\ifmmode\mathsf{#1}\else\textsf{#1}\fi}
\newcommand{\mathsc}[1]{\text{\normalfont\scshape#1}}
\newcommand{\scc}[1]{\relax\ifmmode\mathsc{#1}\else\textsc{#1}\fi}
\newcommand{\mbb}[1]{\mathbb{#1}}
\newcommand{\<}{\langle}
\renewcommand{\>}{\rangle}
\newcommand{\sepr}{\;\mbox{\large{$\mid$}}\;}
\newcommand{\redd}{\mathbin{\longrightarrow}}
\newcommand{\nredd}{{\centernot\longrightarrow}}
\newcommand{\rred}[1]{\shortrightarrow_{#1}}

\newcommand{\fn}{\mathrm{fn}}

\newcommand{\an}{\mathrm{an}}
\newcommand{\live}{\mathrm{live}}
\newcommand{\ol}[1]{\overline{#1}}
\newcommand{\rott}[1]{\mathpalette\rot{#1}}
\newcommand{\rot}[2]{\rotatebox[origin=c]{180}{$#1{#2}$}}
\newcommand{\pr}{\mathrm{pr}}

\newcommand{\figref}[1]{Fig.~\labelcref{#1}}

\let\oprod\prod
\renewcommand{\prod}{\mathchoice{\textstyle}{}{}{}{\oprod}}
\let\omax\max
\renewcommand{\max}{\mathchoice{\textstyle}{}{}{}{\omax}}

\newcommand{\puts}{\mathbin{\triangleleft}}
\renewcommand{\gets}{\mathbin{\triangleright}}
\newcommand{\call}[1]{\<#1\>}
\let\onu\nu
\renewcommand{\nu}[1]{(\mkern-1mu\bm{\onu} #1)}
\let\oprl\|
\renewcommand{\|}{\mathbin{|}}
\newcommand{\prl}{\mathbin{\oprl}}
\newcommand{\0}{\bm{0}}
\newcommand{\fwd}{\mathbin{\leftrightarrow}}
\newcommand{\subst}[1]{\{#1\}}

\newcommand{\tensor}{\ensuremath{\mathbin{\otimes}}}
\newcommand{\parr}{\mathbin{\rott{\&}}}
\newcommand{\pri}{\mathsf{o}}
\newcommand{\opri}{\kappa}

\newcommand{\1}{\bm{1}}

\newtheorem{theorem}{Theorem}[section]
\newtheorem{lemma}[theorem]{Lemma}
\newtheorem{proposition}[theorem]{Proposition}
\newtheorem{corollary}[theorem]{Corollary}
\newtheorem{definition}[theorem]{Definition}
\newtheorem{example}[theorem]{Example}

\newcommand{\sdot}{\mathbin{.}}
\newcommand{\dom}{\mathrm{dom}}

\newcommand{\term}[1]{\textcolor{cblPurple}{#1}}
\newcommand{\type}[1]{\textcolor{cblOrange}{#1}}

\newcommand{\mcl}[1]{\mathcal{#1}}
\newcommand{\lolli}{\mathbin{\multimap}}
\newcommand{\enc}[2]{{\color{cblRed}\left\llbracket\normalcolor #2 \color{cblRed}\right\rrbracket\normalcolor}{#1}}

\newcommand{\encc}[2]{{\color{cblRed}\llbracket\normalcolor \term{#2} \color{cblRed}\rrbracket\normalcolor}{#1}}
\newcommand{\enct}[1]{{\color{cblRed}\llbracket\normalcolor \type{#1} \color{cblRed}\rrbracket\normalcolor}}
\newcommand{\benc}[3]{{\color{cblRed}\left\llbracket\normalcolor #3 \color{cblRed}\right\rrbracket\normalcolor}_{#1}^{#2}}

\newcommand{\bencb}[3]{{\color{cblRed}\llbracket\normalcolor \term{#3} \color{cblRed}\rrbracket\normalcolor}_{#1}^{#2}}

\newcommand{\cmdtt}[2]{#1_\mathtt{#2}}
\newcommand{\equivtt}[1]{\cmdtt{\equiv}{#1}}
\newcommand{\reddtt}[1]{\cmdtt{\redd}{#1}}
\newcommand{\vdashtt}[1]{\cmdtt{\vdash}{#1}}
\newcommand{\equivM}{\equivtt{M}}
\newcommand{\equivC}{\equivtt{C}}
\newcommand{\reddM}{\reddtt{M}}
\newcommand{\reddC}{\reddtt{C}}

\newcommand{\reddL}{\reddtt{L}}
\newcommand{\nreddL}{\cmdtt{\nredd}{L}}
\newcommand{\vdashM}{\vdashtt{M}}
\newcommand{\vdashB}{\vdashtt{B}}
\newcommand{\vdashC}[1]{\vdashtt{C}^{\term{#1}}}
\newcommand{\nuf}[1]{(\overset{\leftrightarrow}{\bm{\onu}} #1)}
\newcommand{\ih}[1]{IH\textsubscript{#1}\xspace}
\newcommand{\dashes}{\vspace{-.5em}\hbox to \textwidth{\leaders\hbox to 3pt{\hss . \hss}\hfil}\smallskip}
\makeatletter
\newcommand*\dashline{\rotatebox[origin=c]{90}{$\dabar@\dabar@\dabar@$}}
\makeatother
\newcommand{\nunil}{\nu{\_\,\_}}
\newcommand{\refwpage}[1]{\Cref{#1} on \Cpageref{#1}\xspace}
\newenvironment{tarr}[2][t]{\begin{array}[#1]{@{}#2@{}}}{\end{array}}
\newcommand{\encpropref}[1]{\Cref{l:transProps} (\labelcref{#1})}

\makeatletter
\@ifpackageloaded{stix}{%
}{%
  \DeclareFontEncoding{LS2}{}{\noaccents@}
  \DeclareFontSubstitution{LS2}{stix}{m}{n}
  \DeclareSymbolFont{stix@largesymbols}{LS2}{stixex}{m}{n}
  \SetSymbolFont{stix@largesymbols}{bold}{LS2}{stixex}{b}{n}
  \DeclareMathDelimiter{\lBrace}{\mathopen} {stix@largesymbols}{"E8}%
                                            {stix@largesymbols}{"0E}
  \DeclareMathDelimiter{\rBrace}{\mathclose}{stix@largesymbols}{"E9}%
                                            {stix@largesymbols}{"0F}
}
\makeatother

\newcommand{\bcont}[1]{\mathrm{bcont}_{#1}}
\newcommand{\vdashAst}{\vdash^{\mkern-3mu\ast}}
\newcommand{\xsub}[1]{{\lBrace #1 \rBrace}}
\newcommand{\bfr}[1]{[#1\>}

\newcommand{\main}{{\mathchoice{\scriptstyle}{\scriptstyle}{\scriptscriptstyle}{}\blacklozenge}}
\newcommand{\child}{{\mathchoice{\scriptstyle}{\scriptstyle}{\scriptscriptstyle}{}\lozenge}}
\newcommand{\rLab}[1]{{#1}}

\newcommand{\sccsm}[1]{{\scriptstyle\scc{#1}}}

\newcommand{\fwded}[1]{{\scriptstyle\langle}#1{\scriptstyle\rangle}}
\renewcommand{\implies}{\mathbin{\Rightarrow}}

\title{Asynchronous Functional Sessions: Cyclic and Concurrent}
\def\titlerunning{Asynchronous Functional Sessions: Cyclic and Concurrent}
\author{
    Bas van den Heuvel \qquad\qquad Jorge A.\ P\'erez
    \institute{University of Groningen, The Netherlands}
    \email{\{b.van.den.heuvel, j.a.perez\} @ rug.nl}
}
\def\authorrunning{B.\ van den Heuvel \& J.A.\ P\'erez}

\begin{document}
\maketitle

\begin{abstract}
We present \ourGVlong (\ourGV), a functional calculus with message-passing concurrency governed by session types.
With respect to prior calculi, \ourGV has increased support for concurrent evaluation and for cyclic network topologies.
The design of \ourGV draws on APCP, a session-typed asynchronous $\pi$-calculus developed in prior work.
Technical contributions are (i)~the syntax, semantics, and type system of \ourGV; (ii)~a correct translation of \ourGV into APCP; (iii) a technique for establishing deadlock-free \ourGV programs, by resorting to APCP's priority-based type system.
\end{abstract}

\section{Introduction}
\label{s:intro}

The goal of this paper is to introduce a new functional calculus with message-passing concurrency governed by linearity and session types.
Our work contributes to a research line initiated by Gay and Vasconcelos~\cite{journal/jfp/GayV10}, who proposed a functional calculus with sessions here referred to as \lGV; this line of work has received much recent attention thanks to Wadler's GV calculus~\cite{conf/icfp/Wadler12}, which is  a variation of \lGV.

Our new calculus is dubbed \ourGVlong (\ourGV);
with respect to previous work, it
presents three intertwined novelties: asynchronous (buffered) communication; a highly concurrent reduction strategy; and thread configurations with cyclic topologies.
The design of \ourGV rests upon a solid basis: an operationally correct translation into APCP (Asynchronous Priority-based Classical Processes), a session-typed $\pi$-calculus in which asynchronous processes communicate by forming cyclic  networks~\cite{report/vdHeuvelP21B}.


We discuss the salient features of \ourGV by example, using a simplified syntax.
As in \lGV, communication in \ourGV is asynchronous: send operations place their messages in buffers, and receive operations read the messages from these buffers.
Let us write $\term{\sff{send}~(u,x)}$ to denote the output of message $\term{u}$ along channel $\term{x}$, and $\term{\sff{recv}~y}$ to denote an input on $y$.
The following program expresses the parallel composition ($\term{\parallel}$) of two threads:
\[
    \term{
        \left.
            \begin{array}{@{}l@{}}
                \sff{let} \, x = \sff{send}~(u,x) \, \sff{in}
                \\
                \quad \sff{let} \, (v,y) = \sff{recv}~y \, \sff{in} \, ()
            \end{array} \qquad \middle|\!\middle| \qquad \sff{let} \, y = \underline{\sff{send}~(w,y)} \, \sff{in} \, ()
        \right.
    }
\]
In variants of \lGV with \emph{synchronous} communication,
such as GV and Kokke and Dardha's PGV~\cite{conf/forte/KokkeD21,report/KokkeD21},
this program is stuck: the send on $\term{y}$ (underlined, on the right) cannot synchronize with the receive on $\term{y}$ (on the left): it is blocked by the send on $\term{x}$ (on the left), and there is no receive on $\term{x}$.
In contrast, in \ourGV the send on $\term{x}$ can be buffered after which the communication on $\term{y}$ can take place.

In \ourGV, reduction  is  ``more concurrent'' than usual call-by-value or call-by-name strategies.
Consider the following  program:
\[
    \term{
        \left(
            \begin{array}{@{}l@{}l@{}}
                \lambda x \sdot {}
                & \sff{let} \, (u,y) = \sff{recv}~y \, \sff{in}
                \\
                & \quad \sff{let} \, x = \sff{send}~(u,x) \, \sff{in} ()
            \end{array}
        \right) \quad \big( \sff{send}~(v,z) \big)
    }
\]
In  \lGV, reduction is call-by-value and so the function on $\term{x}$ can only be applied on a value.
However, the function's parameter (send on $\term{z}$) is not a value, so it needs to be evaluated before the function on $\term{x}$ can be applied.
Hence, this program can only be evaluated in one order: first the send on $\term{z}$, then the receive on $\term{y}$.
In contrast, the semantics of \ourGV evaluates a function and its parameters \emph{concurrently}: the send on $\term{z}$ and the receive on $\term{y}$ can be evaluated in any order.
Note that asynchrony plays no role here: both buffering a message and synchronous communication entail a reduction in the function's parameter.

The third novelty is cyclic thread configurations: threads can be connected by channels to form cyclic networks.
Consider the following program:
\[
    \term{
        \left.
            \begin{array}{@{}l@{}}
                \sff{let} \, (u,x) = \sff{recv}~x \, \sff{in}
                \\
                \quad \sff{let} \, y = \sff{send}~(u,y) \, \sff{in} \, ()
            \end{array} \qquad \middle|\!\middle| \qquad \begin{array}{@{}l@{}}
                \sff{let} \, x = \sff{send}~(v,x) \, \sff{in}
                \\
                \quad \sff{let} \, (w,y) = \sff{recv}~y \, \sff{in} \, ()
            \end{array}
        \right.
    }
\]
Here we have two threads connected on channels $\term{x}$ and $\term{y}$, thus forming a cyclic thread configuration.
Clearly, this program is deadlock-free.
In \lGV, the program is well-typed, but there is no deadlock-freedom guarantee: the type system of \lGV admits deadlocked cyclic thread configurations.
In GV and Fowler~\etal's EGV~\mbox{\cite{conf/popl/FowlerLMD19}} (an extension of Fowler's AGV~\mbox{\cite{thesis/Fowler19}}) there is a deadlock-freedom guarantee for well-typed programs; however, their type systems only support \emph{tree-shaped} thread configurations---this limitation is studied in~\cite{conf/express/DardhaP15,journal/jlamp/DardhaP22}. Hence, the program above is not well-typed in GV and EGV.

These  novelties  are \emph{intertwined}, in the following sense.
Asynchronous communication reduces the synchronization points in programs (as output-like operations are non-blocking), therefore increasing concurrent  evaluation.
In turn, reduced synchronization points can  streamline verification techniques for deadlock-freedom based on \emph{priorities}~\cite{conf/concur/Kobayashi06,conf/csl/Padovani14,conf/forte/PadovaniN15,conf/fossacs/DardhaG18}, which unlock the analysis of process networks with cyclic topologies.
Indeed, in an asynchronous setting only input-like operations require priorities.

We endow \ourGV with a type system with functional types and session types; we opted for a design in which well-typed terms enjoy subject reduction / type preservation but not deadlock-freedom.
To validate our semantic design and attain the three  novelties motivated above, we resort to APCP.
In our developments, APCP operates as a ``low-level'' reference programming model.
We give a typed translation of \ourGV into APCP, which satisfies strong correctness properties, in the sense of  Gorla~\cite{journal/ic/Gorla10}. In particular, it enjoys operational correspondence, which provides a significant sanity check to justify our key design decisions in \ourGV's operational semantics.
Interestingly, using our correct translation and the deadlock-freedom guarantees for well-typed processes in APCP, we obtain a   technique for transferring the deadlock-freedom property to \ourGV programs.
That is,
given a \ourGV program $\term{C}$,
we prove that if the APCP translation of $\term{C}$  is typable (and hence, deadlock-free), then $\term{C}$ itself is deadlock-free.
This result thus delineates a class of deadlock-free \ourGV programs that includes cyclic thread configurations.

In summary, this paper presents the following technical contributions:
(1)~\ourGV, a new functional calculus with session-based asynchronous concurrency;
(2)~A typed translation of \ourGV into APCP, which is proven to satisfy well-studied encodability criteria;
(3)~A transference result for the deadlock-freedom property from APCP to \ourGV programs.
\ifappendix Appendices contain omitted technical details.
\else An extended version contains omitted technical details~\cite{report/vdHeuvelP22}.
\fi


\section{\ourGVlong}
\label{s:ourGV}


\subsection{Syntax and Semantics}
\label{ss:ourGVSyntaxSemantics}

\begin{figure}[t!]
    \begin{mdframed}
        Terms ($\term{L},\term{M},\term{N}$):
        \begin{align*}
            \term{x} &\quad \text{(variable)}
            &
            \term{\sff{new}} &\quad \text{(create new channel)}
            \\
            \term{()} &\quad \text{(unit value)}
            &
            \term{\sff{spawn}~M} &\quad \text{(execute pair $\term{M}$ in parallel)}
            \\
            \term{\lambda x \sdot M} &\quad \text{(abstraction)}
            &
            \term{\sff{send}~(M,N)} &\quad \text{(send $\term{M}$ along $\term{N}$)}
            \\
            \term{M~N} &\quad \text{(application)}
            &
            \term{\sff{recv}~M} &\quad \text{(receive along $\term{M}$)}
            \\
            \term{(M,N)} &\quad \text{(pair construction)}
            &
            \term{\sff{select}\,\ell\,M} &\quad \text{(select label $\ell$ along $\term{M}$)}
            \\
            \term{\sff{let}\,(x,y)=M\,\sff{in}\,N} &\quad \text{(pair deconstruction)}
            &
            \term{\sff{case}\,M\,\sff{of}\,\{i:M\}_{i \in I}} &\quad \text{(offer labels in $I$ along $\term{M}$)}
        \end{align*}

        \dashes

        Runtime terms ($\term{\mbb{L}},\term{\mbb{M}},\term{\mbb{N}}$) and reduction contexts ($\term{\mcl{R}}$):
        \begin{align*}
            \term{\mbb{L}},\term{\mbb{M}},\term{\mbb{N}} ::= {}
            &~ \ldots \sepr \term{\mbb{M}\xsub{ \mbb{N}/x }} \sepr \term{\sff{send}'(\mbb{M},\mbb{N})}
            \\
            \term{\mcl{R}} ::= {}
            &~ \term{[]} \sepr \term{\mcl{R}~\mbb{M}} \sepr \term{\sff{spawn}~\mcl{R} \sepr \term{\sff{send}~\mcl{R}}} \sepr \term{\sff{recv}~\mcl{R}} \sepr \term{\sff{let}\, (x,y) = \mcl{R}\, \sff{in}\, \mbb{M}}
            \\
            \sepr
            &~ \term{\sff{select}\, \ell\, \mcl{R}} \sepr \term{\sff{case}\, \mcl{R}\, \sff{of}\, \{i: \mbb{M}\}_{i \in I}} \sepr \term{\mcl{R}\xsub{ \mbb{M}/x }} \sepr \term{\mbb{M}\xsub{ \mcl{R}/x }} \sepr \term{\sff{send}'(\mbb{M},\mcl{R})}
        \end{align*}

        \dashes

        Structural congruence for terms ($\equivM$) and  term reduction ($\reddM$):
        \begin{align*}
            & \scc{SC-SubExt}
            &
            \term{x} \notin \fn(\term{\mcl{R}}) \implies
            \term{(\mcl{R}[\mbb{M}])\xsub{ \mbb{N}/x }}
            &\equivM \term{\mcl{R}[\mbb{M}\xsub{ \mbb{N}/x }]}
%
%
\\
            & \scc{E-Lam}
            &
            \term{(\lambda x \sdot \mbb{M})~\mbb{N}}
            &\reddM \term{\mbb{M}\xsub{ \mbb{N}/x }}
            \\
            & \scc{E-Pair}
            &
            \term{\sff{let}\, (x,y) = (\mbb{M}_1,\mbb{M}_2)\, \sff{in}\, \mbb{N}}
            &\reddM \term{\mbb{N}\xsub{ \mbb{M}_1/x,\mbb{M}_2/y }}
            \\
            & \scc{E-SubstName}
            &
            \term{\mbb{M}\xsub{ y/x }}
            &\reddM \term{\mbb{M}\{y/x\}}
            \\
            & \scc{E-NameSubst}
            &
            \term{x\xsub{ \mbb{M}/x }}
            &\reddM \term{\mbb{M}}
            \\
            & \scc{E-Send}
            &
            \term{\sff{send}~(\mbb{M},\mbb{N})}
            &\reddM \term{\sff{send}'(\mbb{M},\mbb{N})}
            \\
            & \scc{E-Lift}
            &
            \term{\mbb{M}} \reddM \term{\mbb{N}} \implies \term{\mcl{R}[\mbb{M}]}
            &\reddM \term{\mcl{R}[\mbb{N}]}
            \\
            & \scc{E-LiftSC}
            &
            \term{\mbb{M}} \equivM \term{\mbb{M}'} \wedge \term{\mbb{M}'} \reddM \term{\mbb{N}'} \wedge \term{\mbb{N}'} \equivM \term{\mbb{N}} \implies
            \term{\mbb{M}}
            &\reddM \term{\mbb{N}}
        \end{align*}
    \end{mdframed}
    \caption{The \ourGV term language.}\label{f:gvTerms}
\end{figure}

\begin{figure}[t!]
    \begin{mdframed}
        Markers ($\term{\phi}$), messages ($\term{m},\term{n}$), configurations ($\term{C},\term{D},\term{E}$), thread ($\term{\mcl{F}}$) and configuration ($\term{\mcl{G}}$) contexts:
        \begin{align*}
            \term{\phi} ::= {}
            &~ \term{\main} \sepr \term{\child}
            &
            \term{m},\term{n} ::= {}
            &~ \term{\mbb{M}} \sepr \term{\ell}
            \\
            \term{C},\term{D},\term{E} ::= {}
            &~ \term{\phi\, \mbb{M}} \sepr \term{C \prl D} \sepr \term{\nu{x\bfr{\vec{m}}y}C} \sepr \term{C\xsub{ \mbb{M}/x }}
            \\
            \term{\mcl{F}} ::= {}
            &~ \term{\phi\, \mcl{R}} \sepr \term{C\xsub{ \mcl{R}/x }}
            &
            \term{\mcl{G}} ::= {}
            &~ \term{[]} \sepr \term{\mcl{G} \prl C} \sepr \term{\nu{x\bfr{\vec{m}}y}\mcl{G}} \sepr \term{\mcl{G}\xsub{ \mbb{M}/x }}
        \end{align*}

        \dashes

        Structural congruence for configurations ($\equivC$):
        \begin{align*}
            & \scc{SC-TermSC}
            &
            \term{\mbb{M}} \equivM \term{\mbb{M}'} \implies
            \term{\phi\, \mbb{M}}
            &\equivC \term{\phi\, \mbb{M}'}
            \\
            & \scc{SC-ResSwap}
            &
            \term{\nu{x\bfr{\epsilon}y}C}
            &\equivC \term{\nu{y\bfr{\epsilon}x}C}
            \\
            & \scc{SC-ResComm}
            &
            \term{\nu{x\bfr{\vec{m}}y}\nu{z\bfr{\vec{n}}w}C}
            &\equivC \term{\nu{z\bfr{\vec{n}}w}\nu{x\bfr{\vec{m}}y}C}
            \\
            & \scc{SC-ResExt}
            &
            \term{x},\term{y} \notin \fn(\term{C}) \implies
            \term{\nu{x\bfr{\vec{m}}y}(C \prl D)}
            &\equivC \term{C \prl \nu{x\bfr{\vec{m}}y}D}
            \\
            & \scc{SC-ResNil}
            &
            \term{x},\term{y} \notin \fn(\term{C}) \implies
            \term{\nu{x\bfr{\epsilon}y}C}
            &\equivC \term{C}
            \\
            & \scc{SC-Send'}
            &
            \term{\nu{x\bfr{\vec{m}}y}(\hat{\mcl{F}}[\sff{send}'(\mbb{M},x)] \prl C)}
            &\equivC \term{\nu{x\bfr{\mbb{M},\vec{m}}y}(\hat{\mcl{F}}[x] \prl C)}
            \\
            & \scc{SC-Select}
            &
            \term{\nu{x\bfr{\vec{m}}y}(\mcl{F}[\sff{select}\, \ell\, x] \prl C)}
            &\equivC \term{\nu{x\bfr{\ell,\vec{m}}y}(\mcl{F}[x] \prl C)}
            \\
            & \scc{SC-ParNil}
            &
            \term{C \prl \child\, ()}
            &\equivC \term{C}
            \\
            & \scc{SC-ParComm}
            &
            \term{C \prl D}
            &\equivC \term{D \prl C}
            \\
            & \scc{SC-ParAssoc}
            &
            \term{C \prl (D \prl E)}
            &\equivC \term{(C \prl D) \prl E}
            \\
            & \scc{SC-ConfSubst}
            &
            \term{\phi\,(\mbb{M}\xsub{ \mbb{N}/x })}
            &\equivC \term{(\phi\,\mbb{M})\xsub{ \mbb{N}/x }}
            \\
            & \scc{SC-ConfSubstExt}
            &
            \term{x} \notin \fn(\term{\mcl{G}}) \implies
            \term{(\mcl{G}[C])\xsub{ \mbb{M}/x }}
            &\equivC \term{\mcl{G}[C\xsub{ \mbb{M}/x }]}
        \end{align*}

        \dashes

        Configuration reduction ($\reddC$):
        \begin{align*}
            & \scc{E-New}
            &
            \term{\mcl{F}[\sff{new}]}
            &\reddC \term{\nu{x\bfr{\epsilon}y}(\mcl{F}[(x,y)])}
            \\
            & \scc{E-Spawn}
            &
            \term{\hat{\mcl{F}}[\sff{spawn}~(\mbb{M},\mbb{N})]}
            &\reddC \term{\hat{\mcl{F}}[\mbb{N}] \prl \child\,\mbb{M}}
            \\
            & \scc{E-Recv}
            &
            \term{\nu{x\bfr{\vec{m},\mbb{M}}y}(\hat{\mcl{F}}[\sff{recv}~y] \prl C)}
            &\reddC \term{\nu{x\bfr{\vec{m}}y}(\hat{\mcl{F}}[(\mbb{M},y)] \prl C)}
            \\
            & \scc{E-Case}
            &
            j \in I \implies
            \term{\nu{x\bfr{\vec{m},j}y}(\mcl{F}[\sff{case}\, y\, \sff{of}\, \{i:\mbb{M}_i\}_{i \in I}] \prl C)}
            &\reddC \term{\nu{x\bfr{\vec{m}}y}(\mcl{F}[\mbb{M}_j~y] \prl C)}
            \\
            & \scc{E-LiftC}
            &
            \term{C} \reddC \term{C'} \implies
            \term{\mcl{G}[C]}
            &\reddC \term{\mcl{G}[C']}
            \\
            & \scc{E-LiftM}
            &
            \term{\mbb{M}} \reddM \term{\mbb{M}'} \implies
            \term{\mcl{F}[\mbb{M}]}
            &\reddC \term{\mcl{F}[\mbb{M'}]}
            \\
            & \scc{E-ConfLiftSC}
            &
            \term{C} \equivC \term{C'} \wedge \term{C'} \reddC \term{D'} \wedge \term{D'} \equivC \term{D} \implies
            \term{C}
            &\reddC \term{D}
        \end{align*}
    \end{mdframed}
    \caption{The \ourGV configuration language.}\label{f:gvConfs}
\end{figure}

The main syntactic entities in \ourGV are \emph{terms}, \emph{runtime terms}, and \emph{configurations}.
Intuitively, terms reduce to runtime terms; configurations correspond to the parallel composition of a main thread and several child threads, each executing a runtime term. Buffered messages are part of configurations. We define two reduction relations: one is on terms, which is then subsumed by reduction on configurations.

The syntax of terms ($\term{L},\term{M},\term{N}$) is given and described in \Cref{f:gvTerms}. 
We use $x,y,\ldots$ for \emph{variables}; we write  \emph{endpoint} to refer to a variable used for session operations (send, receive, select, offer).
Let $\fn(\term{M})$ denote the free variables of a term.
All variables are free unless bound: $\term{\lambda x \sdot M}$ binds $\term{x}$ in $\term{M}$, and $\term{\sff{let}\,(x,y)=M\,\sff{in}\,N}$ binds $\term{x}$ and $\term{y}$ in $\term{N}$.
We introduce syntactic sugar for applications of abstractions: $\term{\sff{let}\, x = M\, \sff{in}\, N}$ denotes $\term{(\lambda x \sdot N)~M}$.
For $\term{(\lambda x \sdot M)~N}$, we assume $\term{x} \notin \fn(\term{N})$, and for $\term{\sff{let}\, (x,y) = M\, \sff{in}\, N}$, we assume $x \neq y$ and $\term{x},\term{y} \notin \fn(\term{M})$.

\Cref{f:gvTerms} also gives the reduction semantics of \ourGV terms ($\reddM$), which relies on runtime terms ($\term{\mbb{L}},\term{\mbb{M}},\term{\mbb{N}}$), reduction contexts ($\term{\mcl{R}}$), and structural congruence ($\equivM$).
Note that this semantics comprises the functional fragment of \ourGV; we define the concurrent semantics of \ourGV hereafter.

Runtime terms, whose syntax extends that of terms, guide the evaluation strategy of \ourGV; we discuss an example evaluation of a term using runtime terms after introducing the reduction rules (\mbox{\Cref{x:runtime}}).
Explicit substitution $\term{\mbb{M}\xsub{ \mbb{N}/x }}$ enables the concurrent execution of a function and its parameters.
The intermediate primitive $\term{\sff{send}'(\mbb{M},\mbb{N})}$ enables $\term{\mbb{N}}$ to reduce to an endpoint; the $\term{\sff{send}}$ primitive takes a pair of terms as an argument, inside which reduction is not permitted (\mbox{cf.~\cite{journal/jfp/GayV10}}).
Reduction contexts define the \emph{non-blocking} parts of terms, where subterms may reduce.
We write $\term{\mcl{R}[\mbb{M}]}$ to denote the runtime term obtained by replacing the hole $\term{[]}$ in $\term{\mcl{R}}$ by $\term{\mbb{M}}$, and $\fn(\term{\mcl{R}})$ to denote $\fn(\term{\mcl{R}[()]})$; we will use similar notation for other kinds of contexts later.

We discuss the reduction rules.
The Structural congruence rule~\scc{SC-SubExt} allows the scope extrusion of explicit substitutions along reduction contexts.
Rule~\scc{E-Lam} enforces application, resulting in an explicit substitution.
Rule~\scc{E-Pair} unpacks the elements of a pair into two explicit substitutions (arbitrarily ordered, due to the syntactical assumptions introduced above).
Rules~\scc{E-SubstName} and \mbox{\scc{E-NameSubst}} convert explicit substitutions of or on variables into standard substitutions.
Rule~\scc{E-Send} reduces a $\term{\sff{send}}$ into a $\term{\sff{send}'}$ primitive.
Rules~\scc{E-Lift} and \scc{E-LiftSC} close term reduction under contexts and structural congruence, respectively.
We write $\term{\mbb{M}} \reddM^k \term{\mbb{N}}$ to denote that $\term{\mbb{M}}$ reduces to $\term{\mbb{N}}$ in $k$~steps.

\begin{example}
\label{x:runtime}
    We illustrate the evaluation of terms using runtime terms through the following example, which contains a $\term{\sff{send}}$ primitive and nested abstractions and applications.
    In each reduction step, we underline the subterm that reduces and give the applied rule:
    \begin{align*}
        & \underline{ \term{ \big( \lambda x \sdot \sff{send}~((), x) \big)~\big( (\lambda y \sdot y)~z \big) } }
        & \scc{(E-Lam)}
        \\
        & {} \reddM \term{ \big( {\color{black} \underline{ \term{ \sff{send}~((), x) } } } \big) \xsub{ \big( (\lambda y \sdot y)~z \big) / x } }
        & \scc{(E-Send)}
        \\
        & {} \reddM \term{ \big( \sff{send'}((), x) \big) \xsub{ {\color{black} \underline{ \term{ \big( (\lambda y \sdot y)~z \big) } } } / x } }
        & \scc{(E-Lam)}
        \\
        & {} \reddM \term{ \big( \sff{send'}((), x) \big) \xsub{ ( y \xsub{ z / y } ) / x } }
        & \scc{(SC-SubExt)}
        \\
        & {} \equivM \term{ \sff{send'}\big((), {\color{black} \underline{ \term{ x \xsub{ ( y \xsub{ z / y } ) / x } } } } \big) }
        & \scc{(E-NameSubst)}
        \\
        & {} \reddM \term{ \sff{send'}((), {\color{black} \underline{ \term{ y \xsub{ z / y } } } } ) }
        & \scc{(E-SubstName)}
        \\
        & {} \reddM \term{ \sff{send'}((), z) }
    \end{align*}
    Notice how the $\term{\sff{send}}$ primitive needs to reduce to a $\term{\sff{send'}}$ runtime primitive such that the explicit substitution of $\term{x}$ can be applied.
    Also, note that the concurrency of \ourGV allows many more paths of reduction.
\end{example}

Note that the concurrent evaluation strategy of \ourGV may also be defined without explicit substitutions.
In principle, this would require additional reduction contexts specific to applications on abstractions and pair deconstruction, as well as variants of Rules~\scc{E-SubstName} and \scc{E-NameSubst} specific to these contexts.
However, it is not clear how to define scope extrusion (Rule~\scc{SC-SubExt}) for such a semantics.
Hence, we find that using explicit substitutions drastically simplifies the semantics of \ourGV.

Concurrency in \ourGV allows the parallel execution of terms that communicate through buffers.
The syntax of configurations ($\term{C},\term{D},\term{E}$) is given in \Cref{f:gvConfs}.
The configuration $\term{\phi\,\mbb{M}}$ denotes a \emph{thread}: a concurrently executed term.
The thread marker helps to distinguish the \emph{main} thread ($\term{\phi} = \term{\main}$) from \emph{child} threads ($\term{\phi} = \term{\child}$).
The configuration $\term{C \prl D}$ denotes parallel composition.
The configuration $\term{\nu{x\bfr{\vec{m}}y}C}$ denotes a \emph{buffered restriction}: it connects the endpoints $\term{x}$ and $\term{y}$ through a buffer $\term{\bfr{\vec{m}}}$, binding $\term{x}$ and $\term{y}$ in $\term{C}$.
The buffer's content, $\term{\vec{m}}$, is a sequence of messages (terms and labels).
Buffers are directed: in $\term{x[\vec{m}\>y}$, messages can be added to the front of the buffer on $x$, and they can be taken from the back of the buffer on $y$.
We write $\term{\bfr{\epsilon}}$ for the empty buffer.
The configuration $\term{C\xsub{ \mbb{M}/x }}$ lifts explicit substitution to the level of configurations: this allows spawning and sending terms under explicit substitution, such that the substitution can be moved to the context of the spawned or sent term.

The reduction semantics for configurations ($\reddC$, also in \Cref{f:gvConfs}) relies on thread and configuration contexts ($\term{\mcl{F}}$ and $\term{\mcl{G}}$, respectively) and structural congruence ($\equivC$).
We write $\term{\hat{\mcl{F}}}$ to denote a thread context in which the hole does not occur under explicit substitution, i.e.\ the context is not constructed using the clause $\term{\mcl{R}\xsub{ \mbb{M}/x }}$; this is used in rules for $\term{\sff{send}'}$, $\term{\sff{spawn}}$, and $\term{\sff{recv}}$, effectively forcing the scope extrusion of explicit substitutions when terms are moved between contexts (\mbox{cf.\ \Cref{x:restrictedThread}}).

\begin{sloppypar}
We comment on some of the congruences and reduction rules.
\mbox{Rule~\scc{SC-ResSwap}} allows to swap the direction of an empty buffer; this way, the endpoint that could read from the buffer before the swap can now write to it.
\mbox{Rule~\scc{SC-ResComm}} allows to interchange buffers, and \mbox{Rule~\scc{SC-ResExt}} allows to extrude their scope.
Rule~\scc{SC-ResNil} garbage collects buffers of closed sessions.
\mbox{Rule~\scc{SC-ConfSubst}} lifts explicit substitution at the level of terms to the level of threads, and Rule~\scc{SC-ConfSubstExt} allows the scope extrusion of explicit substitution along configuration contexts.
Notably, putting messages in buffers is not a reduction: Rules~\scc{SC-Send'} and \scc{SC-Select} \emph{equate} sends and selects on an endpoint $\term{x}$ with terms and labels in the buffer for $\term{x}$, as asynchronous outputs are computationally equivalent to messages in buffers.
\end{sloppypar}

Reduction rule~\scc{E-New} creates a new buffer, leaving a reference to the newly created endpoints in the thread.
Rule~\scc{E-Spawn} spawns a child thread (the parameter pair's first element) and continues (as the pair's second element) inside the calling thread.
Rule~\scc{E-Recv} retrieves a term from a buffer, resulting in a pair containing the term and a reference to the receiving endpoint.
Rule~\scc{E-Case} retrieves a label from a buffer, resulting in a function application of the label's corresponding branch to a reference to the receiving endpoint.
There are no reduction rules for closing sessions, as they are closed silently.
We write $\term{C} \reddC^k \term{D}$ to denote that $\term{C}$ reduces to $\term{D}$ in $k$ steps.
Also, we write $\reddC^+$ to denote the transitive closure of $\reddC$ (i.e., reduction in at least one step).

We illustrate \ourGV's semantics by giving some examples.
The following discusses a cyclic thread configuration which does not deadlock due to asynchrony:

\begin{example}\label{x:gvRed}
    Consider configuration $\term{C_1}$ below, in which two threads are spawned and cyclically connected through two channels.
    One thread first sends on the first channel and then receives on the second, while the other thread first sends on the second channel and then receives on the first.
    Under synchronous communication, this would determine a configuration that deadlocks; however,  under asynchronous communication, this is not the case (cf.\ the third example in Sec.\ \labelcref{s:intro}).
    We detail some interesting reductions:
    \begin{align}
        \term{C_1} =~ & \term{\main\, (\sff{let}\, (f,g) = \sff{new}\, \sff{in}\,
        \sff{let}\, (h,k) = \sff{new}\, \sff{in}\,
        \sff{spawn}~\left(\begin{tarr}[c]{l}
                \sff{let}\, f' = (\sff{send}~(u,f))\, \sff{in} \\
                \quad \sff{let}\, (v',h') = (\sff{recv}~h)\, \sff{in}\, (),
                \\[4pt]
                \sff{let}\, k' = (\sff{send}~(v,k))\, \sff{in} \\
                \quad \sff{let}\, (u',g') = (\sff{recv}~g)\, \sff{in}\, ()
        \end{tarr}\right))}
        \nonumber
        \allowdisplaybreaks[1]
        \\
        {}\reddC^8{}
        & \term{\nu{x\bfr{\epsilon}y}\nu{w\bfr{\epsilon}z}(\main\, (\sff{spawn}~\left(\begin{tarr}[c]{l}
                \sff{let}\, f' = (\sff{send}~(u,x))\, \sff{in} \\
                \quad \sff{let}\, (v',h') = (\sff{recv}~w)\, \sff{in}\, (),
            \end{tarr}~\begin{tarr}[c]{l}
                \sff{let}\, k' = (\sff{send}~(v,z))\, \sff{in} \\
                \quad \sff{let}\, (u',g') = (\sff{recv}~y)\, \sff{in}\, ()
        \end{tarr}\right))}
        \label{eq:exGvRedNew}
        \allowdisplaybreaks[1]
        \\
        {}\reddC{}
        & \term{\nu{x\bfr{\epsilon}y}\nu{w\bfr{\epsilon}z}(\main\, \left(\begin{tarr}[c]{l}
                \sff{let}\, k' = (\sff{send}~(v,z))\, \sff{in} \\
                \quad \sff{let}\, (u',g') = (\sff{recv}~y)\, \sff{in}\, ()
        \end{tarr}\right) \prl \child\, \left(\begin{tarr}[c]{l}
                \sff{let}\, f' = (\sff{send}~(u,x))\, \sff{in} \\
                \quad \sff{let}\, (v',h') = (\sff{recv}~w)\, \sff{in}\, ()
        \end{tarr}\right))}
        \label{eq:exGvRedSpawn}
        \allowdisplaybreaks[1]
        \\
        {}\reddC^2{}
        & \term{\nu{x\bfr{\epsilon}y}\nu{w\bfr{\epsilon}z}(\main\, \left(
            \begin{array}{@{}l@{}}
                (\sff{let}\, (u',g') = (\sff{recv}~y)\, \sff{in}\, ())
                \\
                \quad \xsub{ \sff{send}~(v,z)/k' }
            \end{array}
        \right) \prl \child\, \left(\begin{tarr}[c]{l}
                (\sff{let}\, (v',h') = (\sff{recv}~w)\, \sff{in}\, ())
                \\
                \quad \xsub{ \sff{send}~(u,x)/f' }
        \end{tarr}\right))}
        \label{eq:exGvRedLet}
        \allowdisplaybreaks[1]
        \\
        {}\reddC^2{}
        & \term{\nu{x\bfr{\epsilon}y}\nu{w\bfr{\epsilon}z}(\main\, \left(
            \begin{array}{@{}l@{}}
                (\sff{let}\, (u',g') = (\sff{recv}~y)\, \sff{in}\, ())
                \\
                \quad \xsub{ \sff{send}'(v,z)/k' }
            \end{array}
        \right) \prl \child\, \left(\begin{tarr}[c]{l}
                (\sff{let}\, (v',h') = (\sff{recv}~w)\, \sff{in}\, ())
                \\
                \quad \xsub{ \sff{send'}(u,x)/f' }
        \end{tarr}\right))}
        \label{eq:exGvRedSend}
        \allowdisplaybreaks[1]
        \\
        {}\equiv{}
        & \term{\nu{x\bfr{u}y}\nu{z\bfr{v}w}(\main\, ((\sff{let}\, (u',g') = (\sff{recv}~y)\, \sff{in}\, ())\xsub{ z/k' }) \prl \child\, ((\sff{let}\, (v',h') = (\sff{recv}~w)\, \sff{in}\, ()) \xsub{ x/f' }))}
        \label{eq:exGvRedBuf}
        \allowdisplaybreaks[1]
        \\
        {}\reddC^2{}
        & \term{\nu{x\bfr{u}y}\nu{z\bfr{v}w}(\main\, (\sff{let}\, (u',g') = (\sff{recv}~y)\, \sff{in}\, ()) \prl \child\, (\sff{let}\, (v',h') = (\sff{recv}~w)\, \sff{in}\, ()))}
        \nonumber
        \allowdisplaybreaks[1]
        \\
        {}\reddC^2{}
        & \term{\nu{x\bfr{\epsilon}y}\nu{z\bfr{\epsilon}w}(\main\, (\sff{let}\, (u',g') = (v,y)\, \sff{in}\, ()) \prl \child\, (\sff{let}\, (v',h') = (v,w)\, \sff{in}\, ()))}
        \reddC^4 \term{\main\,()}
        \label{eq:exGvRedRecv}
    \end{align}

    \noindent
    Intuitively, reduction \labelcref{eq:exGvRedNew} instantiates two buffers and assigns the endpoints through explicit substitutions.
    Reduction \labelcref{eq:exGvRedSpawn} spawns the left term as a child thread.
    Reduction \labelcref{eq:exGvRedLet} turns $\term{\sff{let}}$s into explicit substitutions.
    Reduction \labelcref{eq:exGvRedSend} turns the $\term{\sff{send}}$s into $\term{\sff{send}'}$s.
    Structural congruence \labelcref{eq:exGvRedBuf} equates the $\term{\sff{send}'}$s with messages in the buffers.
    Reduction \labelcref{eq:exGvRedRecv} retrieves the messages from the buffers.
    Note that many of these steps represent several reductions that may happen in any order.
\end{example}

The following example illustrates \ourGV's flexibility for communicating functions over channels:
\begin{example}
    \label{x:nontrivial}
    In the following configuration, a buffer and two threads have already been set up (cf.\ \Cref{x:gvRed} for an illustration of such an initialization).
    The main thread sends an interesting term to the child thread: it contains the $\term{\sff{send}}$ primitive from which the main thread will subsequently receive from the child thread.
    We give the configuration's major reductions, with the reducing parts underlined:
    \begin{align*}
        & \term{
            \nu{x\bfr{\epsilon}y} ( \main\, \left(
                \begin{array}{@{}l@{}}
                    {\color{black} \underline{ \term{ \sff{let}\, x' = \sff{send}~\big(\lambda z \sdot \sff{send}~((),z),x\big)\, \sff{in} } } }
                    \\
                    \quad \sff{let}\, (v,x'') = \sff{recv}~x'\, \sff{in}\, v
                \end{array}
            \right) \prl \child\, \left(
                \begin{array}{@{}l@{}}
                    \sff{let}\, (w,y') = \sff{recv}~y\, \sff{in}
                    \\
                    \quad \sff{let}\, y'' = (w~y')\, \sff{in}\, ()
                \end{array}
            \right) )
        }
        \\
        {}\reddC^3{} & \term{
            \nu{x\bfr{\lambda z \sdot \sff{send}~((),z)}y} \Big( \main\, \big( \sff{let}\, (v,x'') = \sff{recv}~x\, \sff{in}\, v \big) \prl \child\, \left(
                \begin{array}{@{}l@{}}
                    \sff{let}\, (w,y') = {\color{black} \underline{ \term{ \sff{recv}~y } } }\, \sff{in}
                    \\
                    \quad \sff{let}\, y'' = (w~y')\, \sff{in}\, ()
                \end{array}
            \right) \Big)
        }
        \\
        {}\reddC{} & \term{
            \nu{y\bfr{\epsilon}x} \Big( \main\, \big( \sff{let}\, (v,x'') = \sff{recv}~x\, \sff{in}\, v \big) \prl \child\, \left(
                \begin{array}{@{}l@{}}
                    {\color{black} \underline{ \term{ \sff{let}\, (w,y') = (\lambda z \sdot \sff{send}~((),z),y)\, \sff{in} } } }
                    \\
                    \quad \sff{let}\, y'' = (w~y')\, \sff{in}\, ()
                \end{array}
            \right) \Big)
        }
        \\
        {}\reddC{} & \term{
            \nu{y\bfr{\epsilon}x} \Big( \main\, \big( \sff{let}\, (v,x'') = \sff{recv}~x\, \sff{in}\, v \big) \prl \child\, \Big( {\color{black} \underline{ \term{ \big( \sff{let}\, y'' = (w~y')\, \sff{in}\, () \big) \xsub{\big(\lambda z \sdot \sff{send}~((),z)\big) / w, y/y'} } } } \Big) \Big)
        }
        \\
        {}\reddC^2{} & \term{
            \nu{y\bfr{\epsilon}x} \Big( \main\, \big( \sff{let}\, (v,x'') = \sff{recv}~x\, \sff{in}\, v \big) \prl \child\, \Big( \sff{let}\, y'' = \Big( {\color{black} \underline{ \term{ \big(\lambda z \sdot \sff{send}~((),z)\big)~y } } } \Big)\, \sff{in}\, () \Big) \Big)
        }
        \\
        {}\reddC^2{} & \term{
            \nu{y\bfr{\epsilon}x} \Big( \main\, \big( \sff{let}\, (v,x'') = \sff{recv}~x\, \sff{in}\, v \big) \prl \child\, \big( {\color{black} \underline{ \term{ \sff{let}\, y'' = \sff{send}~((),y)\, \sff{in} } } }\, () \big) \Big)
        }
        \\
        {}\reddC^3{} & \term{
            \nu{y\bfr{()}x} \Big( \main\, \big( {\color{black} \underline{ \term{ \sff{let}\, (v,x'') = \sff{recv}~x\, \sff{in} } } }\, v \big) \Big)
        }
        \reddC^4 \term{ \main\, () }
    \end{align*}
\end{example}

The following example illustrates why the restricted thread context $\term{\hat{\mcl{F}}}$ is used:

\begin{example}
\label{x:restrictedThread}
    Consider the configuration $\term{C} = \term{ \nu{x \bfr{\epsilon} y} \bigg( \main\, \Big( \big( \sff{send'}(z,x) \big) \xsub{ v / z } \Big) \prl D \bigg) }$.
    Suppose Structural congruence rule~$\scc{SC-Send'}$ were defined on unrestricted thread contexts; then the rule applies under the explicit substitution of $\term{z}$:
    $
        \term{C} \equivC \term{ \nu{x \bfr{z} y} \Big( \main\, \big( x \xsub{ v / z } \big) \prl D \Big) }
    $.
    Here, $\term{C}$ and the right-hand-side are inconsistent with each other: in $\term{C}$, the variable $\term{z}$ is bound by the explicit substitution, whereas $\term{z}$ is free on the right-hand-side.
    With the restricted thread contexts we are forced to first extrude the scope of the explicit substitution before applying Rule~$\scc{SC-Send'}$, making sure that $\term{z}$ remains bound:
    \[
        \term{C} \equivC \term{ \bigg( \nu{x \bfr{\epsilon} y} \Big( \main\, \big( \sff{send'}(z,x) \big) \prl D \Big) \bigg) \xsub{v / z} }
        \equivC \term{ \big( \nu{x \bfr{z} y} ( \main\, x \prl D ) \big) \xsub{v / z} }.
    \]
\end{example}

\subsection{Type System}
\label{ss:ourGVTypeSys}

We define a type system for \ourGV, with functional types for functions and pairs and session types for communication.
The syntax and meaning of functional types ($\type{T},\type{U}$) and session types ($\type{S}$) are as follows:
\begin{align*}
    \type{T},\type{U} &::=
    \type{T \times U} & \text{(pair)}
    & ~\sepr
    \type{T \lolli U} & \text{(function)}
    & ~\sepr
    \type{\1} & \text{(unit)}
    & ~\sepr
    \type{S} & \text{(session)}
    \\
    \type{S} &::=
    \type{{!}T \sdot S} & \text{(output)}
    & ~\sepr
    \type{{?}T \sdot S} & \text{(input)}
    & ~\sepr
    \type{\oplus\{i:T\}_{i \in I}} & \text{(select)}
    & ~\sepr
    \type{\&\{i:T\}_{i \in I}} & \text{(case)}
    & ~\sepr
    \type{\sff{end}} 
\end{align*}

\noindent
Session type duality ($\type{\ol{S}}$) is defined as usual; note that only the continuations, and not the messages, of output and input types are dualized.
\begin{align*}
    \type{\ol{{!}T \sdot S}}
    &= \type{{?}T \sdot \ol{S}}
    &
    \type{\ol{{?}T \sdot S}}
    &= \type{{!}T \sdot \ol{S}}
    &
    \type{\ol{\oplus\{i:S_i\}_{i \in I}}}
    &= \type{{\&}\{i:\ol{S_i}\}_{i \in I}}
    &
    \type{\ol{{\&}\{i:S_i\}_{i \in I}}}
    &= \type{\oplus\{i:\ol{S_i}\}_{i \in I}}
    &
    \type{\ol{\sff{end}}}
    &= \type{\sff{end}}
\end{align*}

\begin{figure}[t!]
    \begin{mdframed}\small
        \begin{mathpar}
            \inferrule[T-Var]{ }{
                \term{x}: \type{T} \vdashM \term{x}: \type{T}
            }
            \and
            \inferrule[T-Abs]{
                \type{\Gamma}, \term{x}: \type{T} \vdashM \term{\mbb{M}}: \type{U}
            }{
                \type{\Gamma} \vdashM \term{\lambda x \sdot \mbb{M}}: \type{T \lolli U}
            }
            \and
            \inferrule[T-App]{
                \type{\Gamma} \vdashM \term{\mbb{M}}: \type{T \lolli U}
                \\
                \type{\Delta} \vdashM \term{\mbb{N}}: \type{T}
            }{
                \type{\Gamma}, \type{\Delta} \vdashM \term{\mbb{M}~\mbb{N}}: \type{U}
            }
            \and
            \inferrule[T-Unit]{ }{
                \type{\emptyset} \vdashM \term{()}: \type{\1}
            }
            \and
            \inferrule[T-Pair]{
                \type{\Gamma} \vdashM \term{\mbb{M}}: \type{T}
                \\
                \type{\Delta} \vdashM \term{\mbb{N}}: \type{U}
            }{
                \type{\Gamma}, \type{\Delta} \vdashM \term{(\mbb{M},\mbb{N})}: \type{T \times U}
            }
            \and
            \inferrule[T-Split]{
                \type{\Gamma} \vdashM \term{\mbb{M}}: \type{T \times T'}
                \\
                \type{\Delta}, \term{x}: \type{T}, \term{y}: \type{T'} \vdashM \term{\mbb{N}}: \type{U}
            }{
                \type{\Gamma}, \type{\Delta} \vdashM \term{\sff{let}\, (x,y) = \mbb{M}\, \sff{in}\, \mbb{N}}: \type{U}
            }
            \and
            \inferrule[T-New]{
            }{
                \type{\emptyset} \vdashM \term{\sff{new}}: \type{S \times \ol{S}}
            }
            \and
            \inferrule[T-Spawn]{
                \type{\Gamma} \vdashM \term{\mbb{M}}: \type{\1 \times T}
            }{
                \type{\Gamma} \vdashM \term{\sff{spawn}~\mbb{M}}: \type{T}
            }
            \and
            \inferrule[T-EndL]{
                \type{\Gamma} \vdashM \term{\mbb{M}}: \type{T}
            }{
                \type{\Gamma}, \term{x}: \type{\sff{end}} \vdashM \term{\mbb{M}}: \type{T}
            }
            \and
            \inferrule[T-EndR]{ }{
                \type{\emptyset} \vdashM \term{x}: \type{\sff{end}}
            }
            \and
            \inferrule[T-Send]{
                \type{\Gamma} \vdashM \term{\mbb{M}}: \type{T \times {!}T \sdot S}
            }{
                \type{\Gamma} \vdashM \term{\sff{send}~\mbb{M}}: \type{S}
            }
            \and
            \inferrule[T-Recv]{
                \type{\Gamma} \vdashM \term{\mbb{M}}: \type{{?}T \sdot S}
            }{
                \type{\Gamma} \vdashM \term{\sff{recv}~\mbb{M}}: \type{T \times S}
            }
            \and
            \inferrule[T-Select]{
                \type{\Gamma} \vdashM \term{\mbb{M}}: \type{\oplus\{i:T_i\}_{i \in I}}
                \\
                j \in I
            }{
                \type{\Gamma} \vdashM \term{\sff{select}\, j\, \mbb{M}}: \type{T_j}
            }
            \and
            \inferrule[T-Case]{
                \type{\Gamma} \vdashM \term{\mbb{M}}: \type{{\&}\{i:T_i\}_{i \in I}}
                \\
                \forall i \in I.~ \type{\Delta} \vdashM \term{\mbb{N}_i}: \type{T_i \lolli U}
            }{
                \type{\Gamma}, \type{\Delta} \vdashM \term{\sff{case}\, \mbb{M}\, \sff{of}\, \{i:\mbb{N}_i\}_{i \in I}}: \type{U}
            }
            \and
            \inferrule[T-Sub]{
                \type{\Gamma}, \term{x}: \type{T} \vdashM \term{\mbb{M}}: \type{U}
                \\
                \type{\Delta} \vdashM \term{\mbb{N}}: \type{T}
            }{
                \type{\Gamma}, \type{\Delta} \vdashM \term{\mbb{M}\xsub{ \mbb{N}/x }}: \type{U}
            }
            \and
            \inferrule[T-Send']{
                \type{\Gamma} \vdashM \term{\mbb{M}}: \type{T}
                \\
                \type{\Delta} \vdashM \term{\mbb{N}}: \type{{!}T \sdot S}
            }{
                \type{\Gamma}, \type{\Delta} \vdashM \term{\sff{send}'(\mbb{M},\mbb{N})}: \type{S}
            }
        \end{mathpar}

        \dashes

        \vspace{-3ex}
        \begin{mathpar}
            \inferrule[T-Buf]{ }{
                \type{\emptyset} \vdashB \term{\bfr{\epsilon}}: \type{S'} > \type{S'}
            }
            \and
            \inferrule[T-BufSend]{
                \type{\Gamma} \vdashM \term{M}: \type{T}
                \\
                \type{\Delta} \vdashB \term{\bfr{\vec{m}}}: \type{S'} > \type{S}
            }{
                \type{\Gamma}, \type{\Delta} \vdashB \term{\bfr{\vec{m},M}}: \type{S'} > \type{{!}T \sdot S}
            }
            \and
            \inferrule[T-BufSelect]{
                \type{\Gamma} \vdashB \term{\bfr{\vec{m}}}: \type{S'} > \type{S_j}
                \\
                j \in I
            }{
                \type{\Gamma} \vdashB \term{\bfr{\vec{m},j}}: \type{S'} > \type{\oplus\{i:S_i\}_{i \in I}}
            }
        \end{mathpar}

        \dashes

        \vspace{-3ex}
        \begin{mathpar}
            \inferrule[T-Main]{
                \type{\Gamma} \vdashM \term{\mbb{M}}: \type{T}
            }{
                \type{\Gamma} \vdashC{\main} \term{\main\, \mbb{M}}: \type{T}
            }
            \and
            \inferrule[T-Child]{
                \type{\Gamma} \vdashM \term{\mbb{M}}: \type{\1}
            }{
                \type{\Gamma} \vdashC{\child} \term{\child\, \mbb{M}}: \type{\1}
            }
            \and
            \inferrule[T-ParL]{
                \type{\Gamma} \vdashC{\child} \term{C}: \type{\1}
                \\
                \type{\Delta} \vdashC{\phi} \term{D}: \type{T}
            }{
                \type{\Gamma}, \type{\Delta} \vdashC{\child + \phi} \term{C \prl D}: \type{T}
            }
            \and
            \inferrule[T-ParR]{
                \type{\Gamma} \vdashC{\phi} \term{C}: \type{T}
                \\
                \type{\Delta} \vdashC{\child} \term{D}: \type{\1}
            }{
                \type{\Gamma}, \type{\Delta} \vdashC{\phi + \child} \term{C \prl D}: \type{T}
            }
            \and
            \inferrule[T-Res]{
                \type{\Gamma} \vdashB \term{\bfr{\vec{m}}}: \type{S'} > \type{S}
                \\
                \type{\Delta}, \term{x}: \type{S'}, \term{y}: \type{\ol{S}} \vdashC{\phi} \term{C}: \type{T}
            }{
                \type{\Gamma}, \type{\Delta} \vdashC{\phi} \term{\nu{x\bfr{\vec{m}}y}C}: \type{T}
            }
            \and
            \inferrule[T-ResBuf]{
                \type{\Gamma}, \term{y}: \type{\ol{S}} \vdashB \term{\bfr{\vec{m}}}: \type{S'} > \type{S}
                \\
                \type{\Delta}, \term{x}: \type{S'} \vdashC{\phi} \term{C}: \type{T}
            }{
                \type{\Gamma}, \type{\Delta} \vdashC{\phi} \term{\nu{x\bfr{\vec{m}}y}C}: \type{T}
            }
            \and
            \inferrule[T-ConfSub]{
                \type{\Gamma}, \term{x}: \type{T} \vdashC{\phi} \term{C}: \type{U}
                \\
                \type{\Delta} \vdashM \term{\mbb{M}}: \type{T}
            }{
                \type{\Gamma}, \type{\Delta} \vdashC{\phi} \term{C\xsub{ \mbb{M}/x }}: \type{U}
            }
        \end{mathpar}
    \end{mdframed}
    \caption{Typing rules for terms (top), buffers (center), and configurations (bottom).}\label{f:gvType}
\end{figure}

Typing judgments use typing environments ($\type{\Gamma},\type{\Delta},\type{\Lambda}$) consisting of types assigned to variables ($\term{x}: \type{T}$).
We write $\type{\emptyset}$ to denote the empty  environment; in writing `$\type{\Gamma},\type{\Delta}$', we assume that the variables in $\type{\Gamma}$ and $\type{\Delta}$ are pairwise distinct.
\Cref{f:gvType} (top) gives the type system for (runtime) terms.
Judgments are denoted $\type{\Gamma} \vdashM \term{\mbb{M}}: \type{T}$ and have a \emph{use-provide} reading: term $\term{\mbb{M}}$ \emph{uses} the variables in $\type{\Gamma}$ to \emph{provide} a behavior of type $\type{T}$ (cf.\ Caires and Pfenning~\cite{conf/concur/CairesP10}).
When a term provides type $\type{T}$, we often say that the term is of type $\type{T}$.

Typing rules~\scc{T-Var}, \scc{T-Abs}, \scc{T-App}, \scc{T-Unit}, \scc{T-Pair}, and \scc{T-Split} are standard.
Rule~\scc{T-New} types a pair of dual session types $\type{S \times \ol{S}}$.
Rule~\scc{T-Spawn} types spawning a $\type{\1}$-typed term as a child thread, continuing as a term of type $\type{T}$.
Rules~\scc{T-EndL} and \scc{T-EndR} type finished sessions.
Rule~\scc{T-Send} (resp.\ \scc{T-Recv}) uses a term of type $\type{{!}T \sdot S}$ (resp.\ $\type{{?}T \sdot S}$) to type a send (resp.\ receive) of a term of type $\type{T}$, continuing as type $\type{S}$.
Rule~\scc{T-Select} uses a term of type $\type{\oplus\{i:T_i\}_{i \in I}}$ to type selecting a label $j \in I$, continuing as type $\type{T_j}$.
Rule~\scc{T-Case} uses a term of type $\type{{\&}\{i:T_i\}_{i \in I}}$ to type branching on labels $i \in I$, continuing as type $\type{U}$---each branch is typed $\type{T_i \lolli U}$.
Rule~\scc{T-Sub} types an explicit substitution.
Rule~\mbox{\scc{T-Send'}} types sending directly, not requiring a pair but two separate terms.

\Cref{f:gvType} (bottom) gives the typing rules for configurations.
The typing judgments here are annotated with a thread marker: $\type{\Gamma} \vdashC{\phi} \term{C}: \type{T}$.
The thread marker serves to keep track of whether the typed configuration contains the main thread or not (i.e.\ $\term{\phi} = \term{\main}$ if so, and $\term{\phi} = \term{\child}$ otherwise).
When typing the parallel composition of two configurations, we thus have to combine the thread markers of their judgments.
This combination of thread markers ($\term{\phi + \phi'}$) is defined as follows:
\begin{mathpar}
    \term{\main + \child} = \term{\main}
    \and
    \term{\child + \main} = \term{\main}
    \and
    \term{\child + \child} = \term{\child}
    \and
    (\term{\main + \main} ~\text{is undefined})
\end{mathpar}

Typing rules~\scc{T-Main} and \scc{T-Child} turn a typed term into a thread, where child threads may only be of type $\type{\1}$.
Rules~\scc{T-ParL} and \scc{T-ParR} compose configurations: one configuration must be of type~$\type{\1}$ and have thread marker $\term{\child}$ (i.e., it does not contain a main thread), providing the other configuration's type.
Rule~\scc{T-Res} types buffered restriction, with output endpoint $\term{x}$ and input endpoint $\term{y}$ used in the configuration.
It is possible to send the endpoint $\term{y}$ on $\term{x}$, so there is also a Rule~\scc{T-ResBuf} where $\term{y}$ is used in the buffer.
Unlike with usual typing rules for restriction, the types $\type{S'}$ of $\term{x}$ and $\type{\ol{S}}$ of $\term{y}$ do not necessarily have to be duals.
This is because the restriction's buffer may already contain messages sent on $\term{x}$ but not yet received on $\term{y}$, such that the restricted configuration only needs to use $\term{x}$ according to a continuation of $\type{S}$.
To ensure that $\type{S'}$ is indeed a continuation of $\type{S}$ in accordance with the messages in the buffer, we have additional typing rules for buffers, which we explain hereafter.
Finally, Rule~\scc{T-ConfSub} types an explicit substitution on the level of configurations.

For typing buffers, in \Cref{f:gvType} (center), we have judgments of the form: $\type{\Gamma} \vdashB \term{\bfr{\vec{m}}}: \type{S'} > \type{S}$.
The judgment denotes that $\type{S'}$ is a continuation of $\type{S}$, in accordance with the messages $\term{\vec{m}}$, which use the variables in $\type{\Gamma}$.
The idea of the typing rules is that, starting with an empty buffer at the top of the typing derivation (Rule~\scc{T-Buf}) where $\type{S'} = \type{S}$, Rules~\scc{T-BufSend} and \scc{T-BufSelect} add messages to the end of the buffer.
Rule~\scc{T-BufSend} then prefixes $\type{S}$ with an output of the sent term's type, and Rule~\scc{T-BufSelect} prefixes $\type{S}$ with a selection such that the sent label's continuation is $\type{S}$.

\begin{figure}[t]
    \begin{mdframed}
        \vspace{-1.5em}
        \begin{mathpar}
            \mprset{sep=0.8em}
            \inferrule
            {
                \inferrule*
                {
                    \inferrule*
                    { }
                    { \type{\emptyset} \vdashM \term{z'}: \type{\sff{end}} }
                    \\
                    \inferrule*
                    { }
                    { \type{\emptyset} \vdashB \term{\bfr{\epsilon}}: \type{\sff{end}} > \type{\sff{end}} }
                }
                { \type{\emptyset} \vdashB \term{\bfr{z'}}: \type{\sff{end}} > \type{\ol{S}} }
                \\
                \inferrule*
                {
                    \inferrule*
                    {
                        \inferrule*
                        {
                            \inferrule*
                            {
                                \inferrule*
                                { }
                                { \term{y}:\type{S} \vdashM \term{y} : \type{S} }
                            }
                            { \term{y}:\type{S} \vdashM \term{\sff{revc}~y} : \type{\sff{end} \times \sff{end}} }
                            \\
                            \inferrule*
                            {
                                \inferrule*
                                {
                                    \inferrule*
                                    { }
                                    { \type{\emptyset} \vdashM \term{()}: \type{\1} }
                                }
                                { \term{y_0}: \type{\sff{end}} \vdashM \term{()}: \type{\1} }
                            }
                            { \term{z}: \type{\sff{end}}, \term{y_0}: \type{\sff{end}} \vdashM \term{()}: \type{\1} }
                        }
                        { \term{y}: \type{S} \vdashM \term{\sff{let}\, (z,y_0) = \sff{recv}~y\, \sff{in}\, ()}: \type{\1} }
                    }
                    { \term{y'}: \type{\sff{end}}, \term{y}: \type{S} \vdashM \term{\sff{let}\, (z,y_0) = \sff{recv}~y\, \sff{in}\, ()}: \type{\1} }
                }
                { \term{y'}: \type{\sff{end}}, \term{y}: \type{S} \vdashC{\main} \term{\main\, (\sff{let}\, (z,y_0) = \sff{recv}~y\, \sff{in}\, ())}: \type{\1} }
            }
            { \type{\emptyset} \vdashC{\main} \term{\nu{y'\bfr{z'}y}(\main\, (\sff{let}\, (z,y_0) = \sff{recv}~y\, \sff{in}\, ()))}: \type{\1} }
        \end{mathpar}

        \dashes

        \vspace{-2.5em}
        \begin{mathpar}
            \mprset{sep=1.4em}
            \inferrule
            {
                \inferrule*
                {
                    \inferrule*
                    { }
                    { \type{\emptyset} \vdashM \term{()}: \type{\1} }
                    \\
                    \inferrule*
                    {
                        \inferrule*
                        {
                            \type{\Gamma} \vdashM \term{M}: \type{T}
                            \\
                            \inferrule*
                            { }
                            { \type{\emptyset} \vdashB \term{\bfr{\epsilon}}: \type{S'} > \type{S'} }
                        }
                        { \type{\Gamma} \vdashB \term{\bfr{M}}: \type{S'} > \type{{!}T \sdot S'} }
                    }
                    { \type{\Gamma} \vdashB \term{\bfr{M,\ell}}: \type{S'} > \type{\oplus\{\ell: {!}T \sdot S', \ell': S''\}} }
                }
                { \type{\Gamma} \vdashB \term{\bfr{M,\ell,()}}: \type{S'} > \type{{!}\1 \sdot \oplus\{\ell: {!}T \sdot S', \ell': S''\}} }
                \\
                \type{\Delta}, \term{x}:\type{S'}, \term{y}:\type{{?}\1 \sdot {\&}\{\ell: {?}T \sdot \ol{S'}, \ell': \ol{S''}\}} \vdashC{\phi} \term{C}: \type{U}
            }
            { \type{\Gamma}, \type{\Delta} \vdashC{\phi} \term{\nu{x\bfr{M,\ell,()}y}C}: \type{U} }
        \end{mathpar}
    \end{mdframed}
    \caption{Derivation of configurations (cf.\ \Cref{x:typings}): (top) reduced from the initial one in \Cref{x:gvRed} ($\type{S} = \type{{?}\sff{end} \sdot \sff{end}}$); (bottom) a buffer containing several messages.}
    \label{f:gvExType}
\end{figure}

\begin{example}
    \label{x:typings}
    \Cref{f:gvExType} (top) shows the typing derivation of a configuration reduced from $\term{C_1}$ in \Cref{x:gvRed} (following an alternative path after Reduction~\labelcref{eq:exGvRedBuf}).
    \mbox{\Cref{f:gvExType}} (bottom) shows the typing of the configuration $\term{\nu{x\bfr{M,\ell,()}y}C}$, which has some messages in a buffer; notice how the type of $\term{x}$ in $\term{C}$ is a continuation of the dual of the type of $\term{y}$.

    In the configuration \mbox{$\term{ \nu{x \bfr{ \sff{let}\, (z,y) = \sff{recv}~y\, \sff{in}\, y } y} C }$}, the endoint $\term{y}$ is inside the buffer connecting it with $\term{x}$; to type it, we need \mbox{Rule~$\scc{T-ResBuf}$} (omitting the derivation of the buffer):
    \[
        \inferrule*{%
                \term{y}:\type{{?}\sff{end} \sdot \sff{end}} \vdashB \term{\bfr{ \sff{let}\, (z,y) = \sff{recv}~y\, \sff{in}\, y }}:\type{\sff{end}} > \type{{!}\sff{end} \sdot \sff{end}}
            \\
            \type{\Gamma}, \term{x}:\type{\sff{end}} \vdashC{\phi} \term{C}:\type{U}
        }{%
            \type{\Gamma} \vdashC{\phi} \term{\nu{x \bfr{\sff{let}\, (z,y) = \sff{recv}~y\, \sff{in}\, y} y} C}:\type{U}
        }
    \]
    Note that such buffers will always deadlock: the message in the buffer can never be received.
\end{example}

\paragraph*{Type Preservation}

Well-typed \ourGV terms and configurations satisfy protocol fidelity and communication safety.
These properties follow from type preservation: typing is consistent across structural congruence and reduction.
In both cases the proof is   by induction on the derivation of the congruence and reduction, respectively;
\ifappendix we include full proofs in \mbox{\Cref{a:ourGVtypePresProofs}}.
\else we include full proofs in the extended version of this paper~\cite{report/vdHeuvelP22}.
\fi

\begin{theorem}
    If $\type{\Gamma} \vdashC{\phi} \term{C}: \type{T}$ and $\term{C} \equivC \term{D}$ or $\term{C} \reddC \term{D}$, then $\type{\Gamma} \vdashC{\phi} \term{D}: \type{T}$.
\end{theorem}


\section{APCP (Asynchronous Priority-based Classical Processes)}
\label{s:apcp}

APCP~\cite{report/vdHeuvelP21B} is a linear type system for $\pi$-calculus processes that communicate asynchronously (i.e., the output of messages is non-blocking) on connected channel endpoints.
The  type system assigns to endpoints types that specify two-party protocols, in the style of binary session types~\cite{conf/concur/Honda93}.
In APCP, well-typed processes may be cyclically connected: types rely on \emph{priority} annotations, which enable cyclic connections while ruling out circular dependencies between sessions.
Properties of well-typed APCP processes are \emph{type preservation} (\Cref{t:APCPsubjRed}) and \emph{deadlock-freedom} (\Cref{t:APCPdlFree}).




\begin{figure}[t]
    \begin{mdframed}\small
        Process syntax:
        \begin{align*}
            P,Q ::=
            &~ x[y,z]
            & \text{(output)}
            & ~~~ \sepr ~~
              x(y,z) \sdot P
            & \text{(input)}
            \\[-3pt]
            \sepr\!
            &~ x[z] \triangleleft i
            & \text{(selection)}
            & ~~~ \sepr ~~
              x(z) \triangleright \{i: P\}_{i \in I}
            & \text{(branching)}
            & ~~~ \sepr ~~
              \nu{x y}P
            & \text{(restriction)}
            \\[-3pt]
            \sepr\!
            &~ P \| Q
            & \text{(parallel)}
            & ~~~ \sepr ~~
              \0
            & \text{(inaction)}
            & ~~~ \sepr ~~
              x \fwd y
            & \text{(forwarder)}
        \end{align*}

        \vspace{-1em}

        \smallskip
        Structural congruence:
        \begin{align*}
            P
            &\equiv P'
            \quad\text{(if $P \equiv_\alpha P'$)}
            &
            P \| Q
            &\equiv Q \| P
            &
            x \fwd y
            &\equiv y \fwd x
            \\
            P \| (Q \| R)
            &\equiv (P \| Q) \| R
            &
            P \| \0
            &\equiv P
            &
            \nu{x y} x \fwd y
            &\equiv \0
            \\
            P \| \nu{x y}Q
            &\equiv \nu{x y}(P \| Q)
            \quad\text{(if $x,y \notin \fn(P)$)}
            &
            \nu{x y}\0
            &\equiv \0
            \\
            \nu{x y}\nu{z w} P
            &\equiv \nu{z w}\nu{x y} P
            &
            \nu{x y}P
            &\equiv \nu{y x}P
        \end{align*}

        \vspace{-1em}

        \smallskip
        Reduction:
        \begin{mathpar}
            \inferrule*[right=\rLab{$\rred{\scc{Id}}$},vcenter]{%
                z,y \neq x
            }{
                \nu{yz}(x \fwd y \| P) \redd P \{x/z\}
            }
            \and
            \inferrule*[right=\rLab{$\rred{\tensor \parr}$},vcenter]{ }{%
                \nu{xy}(x[a,b] \| y(v,z) \sdot P) \redd P \{a/v,b/z\}
            }
            \and
            \inferrule*[right=\rLab{$\rred{\oplus \&}$},vcenter]{%
                j \in I
            }{
                \nu{xy}(x[b] \puts j \| y(z) \gets \{i:P_i\}_{i \in I}) \redd P_j \{b/z\}
            }
            \and
            \inferrule*[right=\rLab{$\rred{\equiv}$},vcenter]
            {
                P \equiv P'
                \\
                P' \redd Q'
                \\
                Q' \equiv Q
            }
            { P \redd Q }
            \and
            \inferrule*[right=\rLab{$\rred{\onu}$},vcenter]
            { P \redd Q }
            { \nu{x y} P \redd \nu{x y} Q }
            \and
            \inferrule*[right=\rLab{$\rred{\|}$},vcenter]
            { P \redd Q }
            { P \| R \redd Q \| R }
        \end{mathpar}
    \end{mdframed}

    \caption{Definition of APCP's process language.}
    \label{f:procdef}
\end{figure}

\paragraph*{Syntax and Semantics}
We write $x, y, z, \ldots$ to denote \emph{endpoints} (or \emph{names}), and write $\tilde{x}, \tilde{y}, \tilde{z}, \ldots$ to denote sequences of endpoints.
Also, we write $i, j, k, \ldots$ to denote \emph{labels} and $I, J, K, \ldots$ to denote sets of labels.

\Cref{f:procdef} (top) gives the syntax and meaning of processes.
In APCP, all endpoints are used strictly linearly: each endpoint can be used for exactly one communication only.
However, we want to assign session types to endpoints, so we have to be able to implement sequences of communications.
Therefore, each communication action carries an additional \emph{continuation endpoint} to continue the session on.

The output action $x[y,z]$ sends a message endpoint $y$ and a continuation endpoint $z$ along $x$.
The input prefix $x(y,z) \sdot  P$ blocks until a message and a continuation endpoint are received on $x$, binding $y$ and $z$ in $P$.
The selection action $x[z] \puts i$ sends a label $i$ and a continuation endpoint $z$ along $x$.
The branching prefix $x(z) \gets \{i: P_i\}_{i \in I}$ blocks until it receives a label $i \in I$ and a continuation endpoint $z$ on $x$, binding $z$ in each $P_i$.
Restriction $\nu{x y} P$ binds $x$ and $y$ in $P$ to form a channel for communication.
The process $P \| Q$ denotes parallel composition.
The process $\0$ denotes inaction.
The forwarder process $x \fwd y$ is a primitive copycat process that links together $x$ and $y$.

Endpoints are free unless they are bound somehow.
We write $\fn(P)$ for the set of free names of $P$.
Also, we write $P \subst{x/y}$ to denote the capture-avoiding substitution of the free occurrences of $y$ in $P$ for $x$.
We write sequences of substitutions $P \subst{x_1/y_1} \ldots \subst{x_n/y_n}$ as $P \subst{x_1/y_1, \ldots, x_n/y_n}$.


The reduction relation for processes ($P \redd Q$)  formalizes how complementary actions on connected endpoints may synchronize.
As usual for $\pi$-calculi, reduction relies on \emph{structural congruence} ($P \equiv Q$), which relates processes with minor syntactic differences; it is the smallest congruence on the syntax of processes (\figref{f:procdef} (top)) satisfying the axioms in \Cref{f:procdef} (center).


We define the reduction relation $P \redd Q$ by the axioms and closure rules in \Cref{f:procdef} (bottom).
Rule~$\rred{\scc{Id}}$ implements the forwarder as a substitution.
Rule~$\rred{\tensor \parr}$ synchronizes an output and an input on connected endpoints and substitutes the message and continuation endpoints.
Rule~$\rred{\oplus \&}$ synchronizes a selection and a branch:
the received label determines the continuation process, substituting the continuation endpoint appropriately.
Rules~$\rred{\equiv}$, $\rred{\onu}$, and $\rred{\|}$ close reduction under congruence, restriction, and parallel composition, respectively.
We write $\redd^\ast$ for the reflexive, transitive closure of~$\redd$.

\paragraph*{The Type System}

APCP types processes by assigning binary session types to channel endpoints.
Following Curry-Howard interpretations, we present session types as linear logic propositions (cf. Caires \etal~\cite{journal/mscs/CairesPT16} and  Wadler~\cite{conf/icfp/Wadler12})
extended with  \emph{priority} annotations.
Intuitively, actions typed with lower priority cannot be blocked by those with higher priority.

We write $\pri, \opri, \pi, \rho, \ldots$ to denote priorities, and $\omega$ to denote the ultimate priority that is greater than all other priorities  and cannot be increased further.
That is, $\forall \pri \in \mbb{N}.~\omega > \pri$ and $\forall \pri \in \mbb{N}.~\omega + \pri = \omega$.

\begin{definition}
\label{d:props}
    The following grammar defines the syntax of \emph{session types} $A,B$.
    Let $\pri \in \mbb{N}$.
    \begin{align*}
        A,B &::=
        A \tensor^\pri B & \text{(output)}
        & \sepr
        A \parr^\pri B & \text{(input)}
        & \sepr
        \oplus^\pri \{i: A\}_{i \in I} & \text{(select)}
        & \sepr
        \&^\pri \{i: A\}_{i \in I} & \text{(branch)}
        & \sepr
        \bullet & \text{(end)}
    \end{align*}
\end{definition}

\noindent
Note that type $\bullet$ does not require a priority.

\emph{Duality}, the cornerstone of session types and linear logic, ensures that the two endpoints of a channel have matching actions.
Furthermore, dual types must have matching priority annotations.

\begin{definition}
\label{d:duality}
    The \emph{dual} of session type $A$, denoted $\ol{A}$, is defined inductively as follows:
    \begin{align*}
        \ol{A \tensor^\pri B}
        &:= \ol{A} \parr^\pri \ol{B}
        &
        \ol{\oplus^\pri \{ i: A_i \}_{i \in I}}
        &:= \&^\pri \{ i: \ol{A_i} \}_{i \in I}
        &
        \ol{\bullet}
        &:= \bullet
        \\
        \ol{A \parr^\pri B}
        &:= \ol{A} \tensor^\pri \ol{B}
        &
        \ol{\&^\pri \{ i: A_i \}_{i \in I}}
        &:= \oplus^\pri \{ i: \ol{A_i} \}_{i \in I}
        &
        &
    \end{align*}
\end{definition}

The priority of a type is determined by the priority of the type's outermost connective:

\begin{definition}
\label{d:priority}
    For session type $A$, $\pr(A)$ denotes its \emph{priority}:
    \begin{align*}
        \pr(A \tensor^\pri B)
        := \pr(A \parr^\pri B)
        := \pr(\oplus^\pri\{i:A_i\}_{i \in I})
        := \pr(\&^\pri\{i:A_i\}_{i \in I})
        &:= \pri
        &
        \pr(\bullet)
        &:= \omega
    \end{align*}
\end{definition}

\noindent
The priority of $\bullet$ is $\omega$: it denotes a ``final'' action of protocols without blocking behavior.
Although associated with non-blocking behavior, $\tensor$ and $\oplus$ do have a non-constant priority: they are connected to $\parr$ and $\&$, respectively, which denote blocking actions.


The typing rules of APCP ensure that actions with lower priority are not blocked by those with higher priority (cf.\ Dardha and Gay~\cite{conf/fossacs/DardhaG18}).
To this end, typing rules enforce the following laws:
\begin{enumerate}
    \item\label{i:prioLawLower}
        An action with priority $\pri$ must be prefixed only by inputs and branches with priority strictly smaller than $\pri$---this law does not hold for output and selection, as they are not prefixes;

    \item
        dual actions leading to a synchronization must have equal priorities (cf.\ Def.\ \labelcref{d:duality}).
\end{enumerate}
Judgments are of the form $P \vdash \Gamma$, where $P$ is a process and $\Gamma$ is a context that assigns types to endpoints ($x: A$).
A judgment $P \vdash \Gamma$ then means that $P$ can be typed in accordance with the type assignments for names recorded in $\Gamma$.
The context $\Gamma$ obeys \emph{exchange}: assignments may be silently reordered.
$\Gamma$ is \emph{linear}, disallowing \emph{weakening} (i.e., all assignments must be used) and \emph{contraction} (i.e., assignments may not be duplicated).
The empty context is written $\emptyset$.
In writing $\Gamma, x: A$ we assume that $x \notin \dom(\Gamma)$.
We write $\pr(\Gamma)$ to denote the least priority of all types in $\Gamma$ (cf.\ Def.\ \labelcref{d:priority}).

\begin{figure}[t]
    \begin{mdframed}
        \vspace{-4mm}
        {\small%
            \begin{mathpar}
                \inferrule*[right=\rLab{\scc{Empty}}]
                { }
                { \0 \vdash \emptyset }
                \and
                \inferrule*[right=\rLab{$\bullet$}]
                { P \vdash \Gamma }
                { P \vdash \Gamma, x: \bullet }
                \and
                \inferrule*[right=\rLab{\scc{Id}}]
                { }
                { x \fwd y \vdash x: \ol{A}, y: A }
                \and
                \inferrule*[right=\rLab{\scc{Mix}}]
                {
                    P \vdash \Gamma
                    \\
                    Q \vdash \Delta
                }
                { P \| Q \vdash \Gamma, \Delta }
                \and
                \inferrule*[right=\rLab{\scc{Cycle}}]
                { P \vdash \Gamma, x: A, y: \ol{A} }
                { \nu{x y} P \vdash \Gamma }
                \and
                \inferrule*[right=\rLab{$\tensor$}]
                { }
                { x[y,z] \vdash x: A \tensor^\pri B, y: \ol{A}, z: \ol{B} }
                \and
                \inferrule*[right=\rLab{$\parr$}]
                {
                    P \vdash \Gamma, y: A, z: B
                    \\
                    \pri < \pr(\Gamma)
                }
                { x(y, z) \sdot P \vdash \Gamma, x: A \parr^\pri B }
                \and
                \inferrule*[right=\rLab{$\oplus$}]
                { j \in I }
                { x[z] \puts j \vdash x: \oplus^\pri\{i: A_i\}_{i \in I}, z: \ol{A_j} }
                \and
                \inferrule*[right=\rLab{$\&$}]
                {
                    \forall i \in I.~ P_i \vdash \Gamma, z: A_i
                    \\
                    \pri < \pr(\Gamma)
                }
                { x(z) \gets \{i: P_i\}_{i \in I} \vdash \Gamma, x: \&^\pri\{i: A_i\}_{i \in I} }
            \end{mathpar}
        }%
    \end{mdframed}

    \caption{The typing rules of APCP.}
    \label{f:apcpInf}
\end{figure}

\Cref{f:apcpInf} gives the typing rules.
Rule~\scc{Empty} types an inactive process with no endpoints.
Rule~$\bullet$ silently removes a closed endpoint from the typing context.
Rule~\scc{Id} types forwarding between endpoints of dual type.
Rule~\scc{Mix} types the parallel composition of two processes that do not share assignments on the same endpoints.
Rule~\scc{Cycle} types a restriction, where the two restricted endpoints must be of dual type.
Rule~$\tensor$ types an output action; this rule does not have premises to provide a continuation process, leaving the free endpoints to be bound to a continuation process using \scc{Mix} and \scc{Cycle}.
Similarly, Rule~$\oplus$ types an unbound selection action.
Priority checks are confined to Rules~$\parr$ and $\&$, which type input and branching prefixes, respectively.
In both cases, the used endpoint's priority must be lower than the priorities of the other types in the continuation's typing context, thus enforcing Law~\labelcref{i:prioLawLower} above.


Well-typed processes satisfy protocol fidelity, communication safety, and deadlock-freedom.
The first two properties follow from \emph{type preservation}.
Here we only state these results; see~\cite{report/vdHeuvelP21B} for details.

\begin{theorem}[Type Preservation]\label{t:APCPsubjRed}
    If $P \vdash \Gamma$ and $P \equiv Q$ or $P \redd Q$, then $Q \vdash \Gamma$.
\end{theorem}

\begin{theorem}[Deadlock-freedom]\label{t:APCPdlFree}
    If $P \vdash \emptyset$, then either $P \equiv \0$ or $P \redd Q$ for some $Q$.
\end{theorem}

\section{Translating \ourGV into APCP}
\label{s:translation}

\subsection{The Translation}
\label{ss:ourGVintoAPCP}

In this section, we translate \ourGV into APCP.
We translate entire typing derivations, following, e.g., Wadler~\cite{conf/icfp/Wadler12}.
Given the structure of \ourGV and its type system, the translation is defined in parts: for (runtime) terms, for configurations, and for buffers.
The translation is defined on well-typed configurations which may be deadlocked, so our  translation does not consider priority requirements.
As we will see, typability in APCP will enable us to identify deadlock-free configurations in \ourGV (cf.\ Sec.\ \labelcref{ss:ourGVdf}).

The translation is informed by the semantics of \ourGV.
It is crucial that subterms may only reduce when they occur in reduction contexts.
For example, $\term{M_1}$ and $\term{M_2}$ may not reduce if they appear in a pair $\term{(M_1,M_2)}$.
The translation must thus ensure that subterms are blocked when they do not occur in reduction contexts.
Translations such as Wadler's hinge on blocking outputs and inputs;
for example, the pair $\term{(M_1,M_2)}$ is translated as an output that blocks the translations of $\term{M_1}$ and $\term{M_2}$.
However, outputs in APCP are non-blocking and so
we use additional inputs to disable the reduction of  subterms.
For example, the translation of $\term{(M_1,M_2)}$ adds extra inputs to block the translations of $\term{M_1}$ and $\term{M_2}$.

\begin{figure}[t!]
    \begin{mdframed}
    \vspace{-6mm}
        \begin{align*}
            \enct{T \times U} &= (\enct{T} \parr \bullet) \tensor (\enct{U} \parr \bullet)
            &&\dashline
            &
            \enct{{!}T \sdot S} &= (\ol{\enct{T}} \tensor \bullet) \parr \enct{S}
            &
            \enct{\oplus\{i:T_i\}_{i \in I}} &= {\&}\{i:\enct{T_i}\}_{i \in I}
            \\[-.9ex]
            \enct{T \lolli U} &= (\ol{\enct{T}} \tensor \bullet) \parr \enct{U}
            &&\dashline
            &
            \enct{{?}T \sdot S} &= (\enct{T} \parr \bullet) \tensor \enct{S}
            &
            \enct{{\&}\{i:T_i\}_{i \in I}} &= \oplus\{i:\enct{T_i}\}_{i \in I}
            \\[-.9ex]
            \enct{\1} &= \bullet
            &&\dashline
            &
            \enct{\sff{end}} &= \bullet
        \end{align*}
    \end{mdframed}
    \caption{Translation of \ourGV types into session types.}\label{f:transTypes}
\end{figure}

\Cref{f:transTypes} gives the translation of \ourGV types into APCP types ($\enct{T}$), which already captures the operation of the translation: our translation is similar to the one by Wadler, but includes the aforementioned additional inputs.
It may seem odd that this translation dualizes \mbox{\ourGV} session types (e.g., an output `$\type{!}$' becomes an input `$\parr$').
To understand this, consider that a variable $\term{x}$ typed $\type{{!}T \sdot S}$ represents access to a session which expects the user to send a term of type $\type{T}$ and continue as $\type{S}$, but not the output itself.
Hence, to translate an output on $\term{x}$ into APCP, we need to connect the translation of $\term{x}$ to an actual output.
Since this actual output would be typed with $\tensor$, this means that the translation of $\term{x}$ would need to be dually typed, i.e., typed with $\parr$.
A more technical explanation is that the translation moves from two-sided \mbox{\ourGV} judgments to one-sided APCP judgments, which requires dualization (see, e.g., \mbox{\cite{journal/apal/Girard93,conf/places/vdHeuvelP20}}).

Importantly, the translation preserves duality of session types (by induction on their structure):

\begin{proposition}\label{p:transDuality}
    Given a \ourGV session type $\type{S}$, $\ol{\enct{S}} = \enct{\ol{S}}$.
\end{proposition}

We extend the translation of types to typing environments,  defined as expected.
Similarly, we extend duality to typing environments: $\ol{\Gamma}$ denotes $\Gamma$ with each type dualized.
In this section, we give simplified presentations of the translations, showing only the conclusions of the source and target derivations;
\ifappendix \mbox{\Cref{a:transFull}} presents the translations with full derivations.
\else we include the translations with full derivations in the extended version of this paper~\cite{report/vdHeuvelP22}.
\fi

A remark on notation.
Some translated terms include  annotated restrictions $\nuf{xy}$.
These so-called \emph{forwarder-enabled} restrictions can be ignored in this subsection, but will be useful later when proving soundness (one of the correctness properties of the translation; cf.\ \Cref{ss:opCorr}).

We define the translation of (the typing rules of) terms.
Since a term has a provided type, the translation takes as a parameter a name on which the translation provides this type.
\Cref{f:transTermShort} gives the translation of terms, denoted $\enc{z}{\type{\Gamma} \vdashM \term{\mbb{M}}: \type{T}}$,
where the type $\type{T}$ is provided on $z$.
By abuse of notation, we write $\encc{z}{\mbb{M}}$ to denote the process translation of the term $\term{\mbb{M}}$, and similary for configurations and buffers.
Notice the aforementioned additional inputs to block behavior of subterms in rules such as Rule~\scc{T-Pair}.
Before moving to buffers and configurations, we illustrate the translation of terms by an example:

\begin{example}
    \label{x:termTrans}
    Consider the following subterm from \Cref{x:nontrivial}: $\term{ \big( \lambda z \sdot \sff{send}~((),z) \big)~y }$.
    We gradually discuss how this term translates to APCP, and how the translation is set up to mimick the term's behavior.
    \[
        \encc{q}{\big( \lambda z \sdot \sff{send}~((),z) \big)~y} = \nu{ab}\big( \encc{a}{\lambda z \sdot \sff{send}~((),z)} \| \nu{cd}(b[c,q] \| d(e,\_) \sdot \encc{e}{y}) \big)
    \]
    The function application translates the function on $a$, which is connected to $b$.
    The output on $b$ serves to activate the function, which will subsequently activate the functions parameter ($\encc{e}{y} = y \fwd e$) by means of an output that will be received on $d$.
    \[
        \encc{a}{\lambda z \sdot \sff{send}~((),z)} = a(f,g) \sdot \nuf{hz}( \nunil f[h,\_] \| \encc{g}{\sff{send}~((),z)} )
    \]
    The translation of the function is indeed blocked until it receives on $a$.
    It then outputs on $f$ to activate the function's parameter (which receives on $d$), while the function's body appears in parallel.
    \[
        \encc{g}{\sff{send}~((),z)} = \nu{kl}\Big( \encc{k}{((),z)} \| l(m,n) \sdot \nu{op}\big( \nunil n[o,\_] \| \nu{rs}( p[m,r] \| s \fwd g ) \big) \Big)
    \]
    The translation of the $\term{\sff{send}}$ primitive connects the translation of the pair $\term{((),z)}$ on $k$ to an input on $l$, receiving endpoints for the output term ($m$) and the output endpoint ($n$).
    Once activated by the input on $l$, the term representing the output endpoint is activated by means of an output on~$n$.
    In parallel, the actual output (on $p$) sends the endpoint of the output term ($m$) and a fresh endpoint ($r$) representing the continuation channel after the message has been placed in a buffer (the forwarder~$s \fwd g$).
    \[
        \encc{k}{((),z)} = \nu{tu}\nu{vw}( k[t,v] \| u(a',\_) \sdot \encc{a'}{()} \| w(b',\_) \sdot \encc{b'}{z} )
    \]
    The translation of the pair outputs on $k$ two endpoints for the two terms it contains (to be received by whatever intends to use the pair in the context, e.g., the $\term{\sff{send}}$ primitive on $l$).
    The translations of the two terms inside the pair ($\encc{a'}{()} = \0$ and $\encc{b'}{z} = z \fwd b'$) are both guarded by an input, preventing the terms from reducing until the context explicitly activates them by means of outputs.

    Analogously to the reductions from \Cref{x:nontrivial}---$\term{ \big(\lambda z \sdot \sff{send}~((),z) \big)~y } \reddM^3 \term{ \sff{send'}((),y) }$---we have
    \[
        \encc{q}{\big(\lambda z \sdot \sff{send}~((),z) \big)~y} \redd^5 \encc{q}{\sff{send'}((),y)}.
    \]
\end{example}

\Cref{f:transConfBufShort} (top) gives the translation of configurations, denoted $\encc{z}{\type{\Gamma} \vdashC{\phi} \term{C}: \type{T}}$.
We omit the translation of Rule~\scc{T-ParR}.
Noteworthy are the translations of buffered restrictions: the translation of $\term{\nu{x\bfr{\vec{m}}y}C}$ relies on the translation of   $\term{\bfr{\vec{m}}}$, which is given the translation of $\term{C}$ as its continuation.

\begin{figure}[t]
    \begin{mdframed}[innerrightmargin=0ex,innerleftmargin=.1ex]
        \begin{align*}
            &\sccsm{T-Var}
            &
            \enc{z}{
                \term{x}: \type{T} \vdashM \term{x}: \type{T}
            }
            &= x \fwd z \vdash x: \ol{\enct{\type{T}}}, z: \enct{\type{T}}
            \qquad
            \sccsm{T-Unit}
            \quad
            \enc{z}{
                \type{\emptyset} \vdashM \term{()}: \type{\1}
            }
            = \0 \vdash z: \bullet
            \\
            &\sccsm{T-Abs}
            &
            \enc{z}{
                \type{\Gamma} \vdashM \term{\lambda x \sdot M}: \type{T \lolli U}
            }
            &= z(a,b) \sdot \nuf{cx}(\nu{ef}a[c,e] \| \encc{b}{M}) \vdash \ol{\enct{\type{\Gamma}}}, z: (\ol{\enct{\type{T}}} \tensor \bullet) \parr \enct{\type{U}}
            \\
            &\sccsm{T-App}
            &
            \enc{z}{
                \type{\Gamma}, \type{\Delta} \vdashM \term{M~N}: \type{U}
            }
            &= \nu{ab}(\encc{a}{M} \| \nu{cd}(b[c,z] \| d(e,f) \sdot \encc{e}{N})) \vdash \ol{\enct{\type{\Gamma}}}, \ol{\enct{\type{\Delta}}}, z: \enct{\type{U}}
            \\
            &\sccsm{T-Pair}
            &
            \enc{z}{
                \begin{array}{@{}r@{}}
                    \type{\Gamma}, \type{\Delta} \vdashM \term{(M,N)} \\
                    {}: \type{T \times U}
                \end{array}
            }
            &= \begin{array}{@{}l@{}}
                \nu{ab}\nu{cd}(z[a,c] \|
                b(e,f) \sdot \encc{e}{M} \| d(g,h) \sdot \encc{g}{N})
                \\
                \quad {} \vdash \ol{\enct{\type{\Gamma}}}, \ol{\enct{\type{\Delta}}}, z: (\enct{\type{T}} \parr \bullet) \tensor (\enct{\type{U}} \parr \bullet)
            \end{array}
            \\
            &\sccsm{T-Split}
            &
            \enc{z}{
                \begin{array}{@{}r@{}}
                    \type{\Gamma}, \type{\Delta} \vdashM \term{\sff{let}\, (x,y)} \\
                    \term{{}= M\, \sff{in}\, N}: \type{U}
                \end{array}
            }
            &= \begin{array}{@{}l@{}}
                \nu{ab}(\encc{a}{M} \| b(c,d) \sdot \nuf{ex}\nuf{fy}
                ( \\
                \qquad \nu{gh}c[e,g] \| \nu{kl}d[f,k] \| \encc{z}{N})) \vdash \ol{\enct{\type{\Gamma}}}, \ol{\enct{\type{\Delta}}}, z: \enct{\type{U}}
            \end{array}
            \\
            &\sccsm{T-New}
            &
            \enc{z}{
                \type{\emptyset} \vdashM \term{\sff{new}}: \type{S \times \ol{S}}
            }
            &= \begin{array}[t]{@{}l@{}}
                \nu{ab}(\nu{cd}a[c,d] \| b(e,f) \sdot
                \nu{xy}\encc{z}{(x,y)})
                \\
                \quad {} \vdash z: (\enct{\type{S}} \parr \bullet) \tensor (\enct{\type{\ol{S}}} \parr \bullet)
            \end{array}
            \\
            &\sccsm{T-Spawn}
            &
            \enc{z}{
                \type{\Gamma} \vdashM \term{\sff{spawn}~M}: \type{T}
            }
            &= \nu{ab}(\encc{a}{M} \| b(c,d) \sdot (\nu{ef}c[e,f] \| \nu{gh}d[z,g])) \vdash \ol{\enct{\type{\Gamma}}}, z: \enct{\type{T}}
            \\
            &\sccsm{T-EndL}
            &
            \enc{z}{
                \type{\Gamma}, \term{x}: \type{\sff{end}} \vdashM \term{M}: \type{T}
            }
            &= \encc{z}{M} \vdash \ol{\enct{\type{\Gamma}}}, x: \bullet, z: \enct{\type{T}}
            \qquad
            \sccsm{T-EndR}
            \quad
            \enc{z}{
                \type{\emptyset} \vdashM \term{x}: \type{\sff{end}}
            }
            = \0 \vdash z: \bullet
            \\
            &\sccsm{T-Send}
            &
            \enc{z}{
                \type{\Gamma} \vdashM \term{\sff{send}~M}: \type{S}
            }
            &= \begin{array}[t]{@{}l@{}}
                \nu{ab}(\encc{a}{M} \| b(c,d) \sdot \nu{ef}
                (\nu{gh}d[e,g] \\
                \qquad {}\| \nu{kl}(f[c,k] \| l \fwd z))) \vdash \ol{\enct{\type{\Gamma}}}, z: \enct{\type{S}}
            \end{array}
            \\
            &\sccsm{T-Recv}
            &
            \enc{z}{
                \type{\Gamma} \vdashM \term{\sff{recv}~M}: \type{T \times S}
            }
            &= \begin{array}[t]{@{}l@{}}
                \nu{ab}(\encc{a}{M} \| b(c,d) \sdot \nu{ef}(
                z[c,e] \| f(g,h) \sdot d \fwd g))
                \\
                \quad {} \vdash \ol{\enct{\type{\Gamma}}}, z: (\enct{\type{T}} \parr \bullet) \tensor (\enct{\type{S}} \parr \bullet)
            \end{array}
            \\
            &\sccsm{T-Select}
            &
            \enc{z}{
                \type{\Gamma} \vdashM \term{\sff{select}\, j\, M}: \type{T_j}
            }
            &= \nu{ab}(\encc{a}{M} \| \nu{cd}(b[c] \puts j \| d \fwd z)) \vdash \ol{\enct{\type{\Gamma}}}, z: \enct{\type{T_j}}
            \\
            &\sccsm{T-Case}
            &
            \enc{z}{
                {\begin{array}{@{}r@{}}
                    \type{\Gamma}, \type{\Delta} \vdashM \term{\sff{case}\, M} \\
                    \term{\sff{of}\, \{i:N_i\}_{i \in I}}: \type{U}
                \end{array}}
            }
            &= \nu{ab}(\encc{a}{M} \| b(c) \gets \{i:\encc{z}{N_i~c}\}_{i \in I}) \vdash \ol{\enct{\type{\Gamma}}}, \ol{\enct{\type{\Delta}}}, z:\enct{\type{U}}
            \\
            &\sccsm{T-Sub}
            &
            \enc{z}{
                \type{\Gamma}, \type{\Delta} \vdashM \term{\mbb{M}\xsub{ \mbb{N}/x }}: \type{U}
            }
            &= \nuf{xa}(\encc{z}{\mbb{M}} \| \encc{a}{\mbb{N}}) \vdash \ol{\enct{\type{\Gamma}}}, \ol{\enct{\type{\Delta}}}, z: \enct{\type{U}}
            \\
            &\sccsm{T-Send'}
            &
            \enc{z}{
                \begin{array}{@{}r@{}}
                    \type{\Gamma}, \type{\Delta} \vdashM \term{\sff{send}'} \\
                    \term{(M,\mbb{N})}: \type{S}
                \end{array}
            }
            &= \begin{array}{@{}l@{}}
                \nu{ab}(a(c,d) \sdot \encc{c}{M} \| \nu{ef}(\encc{e}{\mbb{N}} \| \nu{gh}(
                f[b,g] \| h \fwd z)))
                \\
                \quad {} \vdash \ol{\enct{\type{\Gamma}}}, \ol{\enct{\type{\Delta}}}, z: \enct{\type{S}}
            \end{array}
        \end{align*}
    \end{mdframed}\vspace{-.9em}
    \caption{%
        Translation of (runtime) term typing rules.
        \ifappendix See \Cref{a:transFull} for typing derivations.
        \else See~\cite{report/vdHeuvelP22} for typing derivations.
        \fi%
    }\label{f:transTermShort}
\end{figure}

The translation of buffers requires care: each message in the buffer is translated as an output in APCP, where the output of the following messages is on the former output's continuation endpoint.
Once there are no more messages in the buffer, the translation uses a typed APCP process---a parameter of the translation---to provide the behavior of the continuation of the lastmost output.
The translation has no requirements for the continuation process and its typing, except for the type of the buffer's endpoint.
With this in mind, \Cref{f:transConfBufShort} (bottom) gives the translation of the typing rules of buffers, denoted
$
\benc{x}{P \vdashAst \Lambda, x: \ol{\enct{\type{S'}}}}{\type{\Gamma} \vdashB \term{\bfr{\vec{m}}}: \type{S'} > \type{S}}
$,
where $x$ is the endpoint on which the buffer outputs, and $P$ is the continuation of the buffer's last message.
Note that we never use the typing rules for buffers by themselves: they always accompany the typing of endpoint restriction, of which the translation properly instantiates the continuation process.

\begin{figure}[t!]
    \begin{mdframed}
        \begin{align*}
            &\sccsm{T-Main/Child}
            &
            \enc{z}{
                \type{\Gamma} \vdashC{\phi} \term{\phi\, \mbb{M}}: \type{T}
            }
            &= \encc{z}{\mbb{M}} \vdash \ol{\enct{\type{\Gamma}}}, z: \enct{\type{T}}
            \\
            &\sccsm{T-ParL}
            &
            \enc{z}{
                \type{\Gamma}, \type{\Delta} \vdashC{\child + \phi} \term{C \prl D}: \type{T}
            }
            &= \nu{ab}\encc{a}{C} \| \encc{z}{D} \vdash \ol{\enct{\type{\Gamma}}}, \ol{\enct{\type{\Delta}}}, z: \enct{\type{T}}
            \\
            &\sccsm{T-Res/T-ResBuf}
            &
            \enc{z}{
                \type{\Gamma}, \type{\Delta} \vdashC{\phi} \term{\nu{x\bfr{\vec{m}}y}C}: \type{T}
            }
            &= \inferrule{}{
                \nu{xy}\bencb{x}{\encc{z}{C}}{\bfr{\vec{m}}} \vdash \ol{\enct{\type{\Gamma}}}, \ol{\enct{\type{\Delta}}}, z: \enct{T}
            }
            \\
            &\sccsm{T-ConfSub}
            &
            \enc{z}{
                \type{\Gamma}, \type{\Delta} \vdashC{\phi} \term{C\xsub{ \mbb{M}/x }}: \type{U}
            }
            &= \nuf{xa}(\encc{z}{C} \| \encc{z}{\mbb{M}}) \vdash \ol{\enct{\type{\Gamma}}}, \ol{\enct{\type{\Delta}}}, z: \enct{\type{U}}
            \\[8pt]
            &\sccsm{T-Buf}
            &
            \benc{x}{P \vdash \Lambda, x: \ol{\enct{\type{S'}}}}{
                \type{\emptyset} \vdashB \term{\bfr{\epsilon}}: \type{S'} > \type{S'}
            }
            &= P \vdash \Lambda, x: \ol{\enct{\type{S'}}}
            \\
            &\sccsm{T-BufSend}
            &
            \benc{x}{P \vdash \Lambda, x: \ol{\enct{\type{S'}}}}{
                \begin{array}{@{}r@{}}
                    \type{\Gamma}, \type{\Delta} \vdashB \term{\bfr{\vec{m},M}} \\
                    {}: \type{S'}
                    {}> \type{{!}T \sdot S}
                \end{array}
            }
            &= \begin{array}{@{}l@{}}
                \nu{ab}\nu{cd}
                (\nu{gh}(x \fwd g \| h[a,c]) \| b(e,f) \sdot \encc{e}{M} \\
                \quad {}\| \bencb{d}{P\{d/x\}}{\bfr{\vec{m}}}) \vdash
                    \ol{\enct{\type{\Gamma}}}, \ol{\enct{\type{\Delta}}}, \Lambda,
                    x: (\enct{\type{T}} \parr \bullet) \tensor \ol{\enct{\type{S}}}
            \end{array}
            \\
            &\sccsm{T-BufSelect}
            &
            \benc{x}{P \vdash \Lambda, x: \ol{\enct{\type{S'}}}}{
                \begin{array}{@{}r@{}}
                    \type{\Gamma} \vdashB \term{\bfr{\vec{m},j}} \\
                    {}: \type{S'}
                    {}> \type{\oplus\{i:S_i\}_{i \in I}}
                \end{array}
            }
            &= \begin{array}{@{}l@{}}
                \nu{ab}
                (\nu{cd}(x \fwd c \| d[a] \puts j) \\
                \quad {}\| \bencb{b}{P\{b/x\}}{\bfr{\vec{m}}}) \vdash \ol{\enct{\type{\Gamma}}}, \Lambda, x: \oplus\{i:\ol{\enct{\type{S_i}}}\}_{i \in I}
            \end{array}
        \end{align*}
    \end{mdframed}
    \caption{%
        Translation of configuration and buffer typing rules.
        \ifappendix See \Cref{a:transFull} for typing derivations.
        \else See~\cite{report/vdHeuvelP22} for typing derivations.
        \fi%
    }\label{f:transConfBufShort}
\end{figure}

Because \ourGV configurations may deadlock, the type preservation result of our translation holds up to priority requirements.
To formalize this, we have the following definition:

\begin{definition}
\label{d:vdashAst}
    Let $P$ be a process.
    We write $P \vdashAst \Gamma$ to denote that $P$ is well-typed according to the typing rules in \Cref{f:apcpInf} where Rules~$\parr$ and $\&$ are modified by erasing priority checks.
\end{definition}

Hence, if $P \vdash \Gamma$ then $P \vdashAst \Gamma$ but the converse does not hold.
Our translation  correctly preserves the typing of terms, configurations, and buffers:

\begin{theorem}[Type Preservation for the Translation]\label{t:transTypePres}
    ~
    \begin{itemize}
        \item[$\bullet$]
            $\enc{z}{\type{\Gamma} \vdashM \term{\mbb{M}}: \type{T}} = \encc{z}{\mbb{M}} \vdashAst \ol{\enct{\Gamma}}, z: \enct{T}$
         \qquad \qquad \qquad \qquad     $\bullet ~~\encc{z}{\type{\Gamma} \vdashC{\phi} \term{C}: \type{T}} = \encc{z}{C} \vdashAst \ol{\enct{\Gamma}}, z: \enct{T}$
        \item[$\bullet$]
            $\benc{x}{P\, \vdashAst \Lambda, x: \ol{\enct{S'}}}{\type{\Gamma} \vdashB \term{\bfr{\vec{m}}}: \type{S'} > \type{S}} = \bencb{x}{P}{\bfr{\vec{m}}} \vdashAst \ol{\enct{\Gamma}}, \Lambda, x: \enct{S}$
    \end{itemize}
\end{theorem}


\begin{example}\label{x:trans}
    Consider again the configuration
    $
    \term{\nu{x\bfr{M,\ell,()}y}C}
    $.
    We illustrate the translation of buffers into APCP by giving the translation of this configuration (writing $\fwded{x}[a,b]$ to denote the forwarded output $\nu{cd}(x \fwd c \| d[a,b])$):
    \begin{align*}
        & \encc{z}{\nu{x\bfr{M,\ell,()}y}C}
        = \nu{xy}\bencb{x}{\encc{z}{C}}{\bfr{M,\ell,()}}
        = \nu{xy} ~ \nu{ab}\nu{cx'} ( \fwded{x}[a,c] \| b(d,\_) \sdot \0
        \\
        &\qquad {} \| \nu{ex''} ( \fwded{x'}[e] \puts \ell \| \nu{fg}\nu{hx'''} ( \fwded{x''}[f,h] \| g(k,\_) \sdot \encc{k}{M} \| \encc{z}{C} \{x'''/x\} ) ) )
    \end{align*}
    Notice how the (forwarded) outputs are sequenced by continuation endpoints, and how the translation of~$\term{C}$ uses the last continuation endpoint $x'''$ to interact with the buffer.
\end{example}

\subsection{Operational Correctness}
\label{ss:opCorr}

Following Gorla~\cite{journal/ic/Gorla10}, we focus   on \emph{operational correspondence}: a translated configuration can reproduce all of the source configuration's reductions (completeness; \Cref{t:completeness}), and any of the translated configuration's reductions can be traced back to reductions of the source configuration (soundness; \Cref{t:soundness}).
With the soundness result, our translation is stronger than related prior translations~\cite{book/Milner89,book/SangiorgiW03,conf/esop/LindleyM15}.


Our completeness result states that the reductions of a well-typed configuration   can be mimicked by its translation in zero or more steps.

\begin{restatable}[Completeness] 
{theorem}{thmTransConfRed}\label{t:completeness}
    Given $\type{\Gamma} \vdashC{\phi} \term{C}: \type{T}$, if $\term{C} \reddC \term{D}$, then $\encc{z}{C} \redd^\ast \encc{z}{D}$.
\end{restatable}

\begin{proof}[Proof (Sketch)]
    By induction on the derivation of the configuration's reduction.
    In each case, we infer the shape of the configuration from the reduction and well-typedness.
    We then consider the translation of the configuration, and show that the resulting process reduces in zero or more steps to the translation of the reduced configuration.
    \ifappendix See \mbox{\Cref{as:completeness}} for a full proof.
    \else See the extended version of this paper~\cite{report/vdHeuvelP22} for a full proof.
    \fi
\end{proof}


Soundness states that any sequence of reductions from the translation of a well-typed configuration eventually leads to the translation of another configuration, which the initial configuration also reduces to.
Asynchrony in APCP requires us to be careful, specifically concerning the semantics of variables in \ourGV.
Variables can only cause reductions under specific circumstances.
On the other hand, variables translate to forwarders in APCP, which reduce as soon as they are bound by restriction.
This semantics for forwarders turns out to be too eager for soundness.
As a result, soundness only holds for an alternative, so-called \emph{lazy semantics} for APCP, denoted $\reddL$, in which forwarders may only cause reductions under specific circumstances.
It is here that the forwarder-enabled restrictions $\nuf{xy}$ anticipated in \Cref{ss:ourGVintoAPCP} come into play.
As we will see in \Cref{ss:ourGVdf}, this alternative semantics does not prevent us from identifying a class of deadlock-free \ourGV configurations through the translation into APCP.
\ifappendix Due to space limitations,  the definition of the lazy semantics appears in \mbox{\Cref{as:soundness}}.
\else Due to space limitations, the definitions of the lazy semantics only appears in the extended version of this paper~\cite{report/vdHeuvelP22}.
\fi

\begin{restatable}[Soundness] 
{theorem}{thmTransSndConf}
    \label{t:soundness}
    Given $\type{\Gamma} \vdashC{\phi} \term{C}: \type{T}$, if $\encc{z}{C} \reddL^\ast Q$, then  $\term{C} \reddC^\ast \term{D}$ and $Q \reddL^\ast \encc{z}{D}$ for some~$\term{D}$.
\end{restatable}

\begin{proof}[Proof (Sketch)]
    By induction on the structure of $\term{C}$.
    In each case, we additionally apply induction on the number $k$ of steps $\encc{z}{C} \reddL^k Q$.
    We then consider which reductions might occur from $\encc{z}{C}$ to $Q$.
    Considering the structure of $\term{C}$, we then isolate a sequence of $k'$ possible steps, such that $\encc{z}{C} \reddL^{k'} \encc{z}{D'}$ for some $\term{D'}$ where $\term{C} \reddC \term{D'}$.
    Since $\encc{z}{D'} \reddL^{k-k'} Q$, it then follows from the induction hypothesis that there exists $\term{D}$ such that $\term{D'} \reddC^\ast \term{D}$ and $\encc{z}{D'} \reddL^\ast \encc{z}{D}$.

    Key here is the independence of reductions in APCP:
    if two or more reductions are enabled from a (well-typed) process, they must originate from independent parts of the process, and so they do not interfere with each other.
    This essentially means that the order in which independent reductions occur does not affect the resulting process.
    Hence, we can pick ``desirable'' sequences of reductions, postponing other possible reductions.
    \ifappendix See \mbox{\Cref{as:soundness}} for a full proof of soundness.
    \else See the extended version of this paper~\cite{report/vdHeuvelP22} for a full proof of soundness.
    \fi
\end{proof}

\noindent
From the proof above we can deduce that if the translation takes at least one step, then so does the source:

\begin{corollary}
    \label{cor:soundnessPlus}
    Given $\type{\Gamma} \vdashC{\phi} \term{C}: \type{T}$, if $\encc{z}{C} \reddL^+ Q$, then $\term{C} \reddC^+ \term{D}$ and $Q \reddL^\ast \encc{z}{D}$ for some~$\term{D}$.
\end{corollary}

\subsection{Transferring Deadlock-freedom from APCP to \ourGV}
\label{ss:ourGVdf}

In APCP, well-typed processes typable under empty contexts ($P \vdash \emptyset$) are deadlock-free.
By appealing to the operational correctness of our translation, we transfer this result to \ourGV configurations.
Each   deadlock-free configuration in \ourGV obtained via transference satisfies two requirements:
\begin{itemize}
    \item The configuration is typable $\type{\emptyset} \vdashC{\main} \term{C} : \type{\1}$: it needs no external resources and has no external behavior.
    \item The typed translation of the configuration satisfies APCP's priority requirements: it is well-typed under `$\vdash$', not only under `$\vdashAst$' (cf.\ Def.\ \labelcref{d:vdashAst}).
\end{itemize}

We rely on soundness (\Cref{t:soundness}) to transfer deadlock-freedom to configurations.
However, APCP's deadlock-freedom (\Cref{t:APCPdlFree}) considers standard semantics ($\redd$), whereas soundness considers the lazy semantics ($\reddL$).
Therefore, we first must show that if the translation of a configuration satisfying the requirements above reduces under $\redd$, it also reduces under $\reddL$; this is \Cref{t:confTransReddL} below.
The deadlock-freedom of these configurations (\Cref{t:ourGVdfFree}) then follows from \Cref{t:APCPdlFree,t:confTransReddL}.
\ifappendix See \mbox{\Cref{a:ourGVdlFree}} for detailed proofs of these results.
\else See the extended version of this paper~\cite{report/vdHeuvelP22} for detailed proofs of these results.
\fi

\begin{restatable}{theorem}{thmConfTransReddL}
\label{t:confTransReddL}
    Given $\type{\emptyset} \vdashC{\main} \term{C}: \type{\1}$, if $\encc{z}{C} \vdash \Gamma$ for some $\Gamma$ and $\encc{z}{C} \redd Q$, then  $\encc{z}{C} \reddL Q'$, for some $Q'$.
\end{restatable}

\begin{proof}[Proof (Sketch)]
    By inspecting the derivation of $\encc{z}{C} \redd Q$.
    If the reduction is not derived from $\rred{\scc{Id}}$, it can be directly replicated under $\reddL$.
    Otherwise, we analyze the possible shapes of $\term{C}$ and show that a different reduction under $\reddL$ is possible.
\end{proof}

\begin{restatable}[Deadlock-freedom for \ourGV]
{theorem}{thmOurGVdfFree}
\label{t:ourGVdfFree}
    \!Given $\type{\emptyset} \vdashC{\main} \term{C}: \type{\1}$, if $\encc{z}{C} \vdash \Gamma$ for some $\Gamma$, then $\term{C} \equiv \term{\main\,()}$ or $\term{C} \reddC \term{D}$ for some~$\term{D}$.
\end{restatable}

\begin{proof}[Proof (Sketch)]
    By assumption and \Cref{t:transTypePres}, $\encc{z}{C} \vdash z:\bullet$.
    Then $\nu{z\_}\encc{z}{C} \vdash \emptyset$.
    By \Cref{t:APCPdlFree}, (i)~$\nu{z\_}\encc{z}{C} \equiv \0$ or (ii)~$\nu{z\_}\encc{z}{C} \redd Q$ for some $Q$.
    In case~(i) it follows from the well-typedness and translation of $\term{C}$ that $\term{C} \equivC \term{\main\, ()}$.
    In case~(ii) we deduce that the reduction of $\nu{z\_}\encc{z}{C}$ cannot involve the endpoint $z$.
    Hence, $\encc{z}{C} \redd Q_0$ for some $Q_0$.
    By \Cref{t:confTransReddL}, then $\encc{z}{C} \reddL Q'$ for some $Q'$.
    Then, by \mbox{\Cref{cor:soundnessPlus}}, there exists $\term{D'}$ such that $\term{C} \reddC^+ \term{D'}$.
    Hence, $\term{C} \reddC \term{D}$ for some $\term{D}$, proving the thesis.
\end{proof}

\smallskip \noindent
As an example, using \Cref{t:ourGVdfFree} we can show  that $\term{C_1}$ from \Cref{x:gvRed} is deadlock-free;
\ifappendix see \mbox{App.~\labelcref{a:example}}.
\else see~\cite{report/vdHeuvelP22}.
\fi

\section{Conclusion}
\label{s:conclusion}

We have presented \ourGV, a new functional language with asynchronous session-typed communication.
As illustrated in \Cref{s:intro}, \ourGV is strictly more expressive than its predecessors, thanks to a highly asynchronous semantics (compared to GV and PGV), its support for cyclic thread configurations (compared to EGV), and the ability to send whole terms and not just values (compared to all the mentioned calculi).
\Cref{tbl:comparison} summarizes the features of \ourGV compared to its predecessors.

An operationally correct translation into APCP solidifies the design of \ourGV, and enables identifying a class of deadlock-free \ourGV programs.
Interestingly, the asynchronous semantics of \ourGV is reminiscent of \emph{future}/\emph{promise} programming paradigms (see, e.g.,~\cite{journal/pls/Halstead85,report/OstheimerD93,journal/cl/TremblayM00}), which have been little studied in the context of session-typed communication.

The alternative to establishing deadlock-freedom in \ourGV via translation into APCP would be to enhance \ourGV's type system with priorities (in the spirit of, e.g., work by Padovani and Novara~\cite{conf/forte/PadovaniN15}).
Another useful addition concerns recursion / recursive types.
We leave these extensions to future work.

\begin{table}[!t]
    \begin{center}
        \begin{tabular}{|c|c|c|c|c|c|}
            \hline
            & $\lGV$~\cite{journal/jfp/GayV10}
            & GV~\cite{conf/icfp/Wadler12}
            & EGV~\cite{conf/popl/FowlerLMD19}
            & PGV~\cite{conf/forte/KokkeD21,report/KokkeD21}
            & \textbf{\ourGV (this paper)}
            \\ \hline
            Communication
            & Asynch.
            & Synch.
            & Asynch.
            & Synch.
            & Asynch.
            \\ \hline
            Cyclic Topologies
            & Yes
            & No
            & No
            & Yes
            & Yes
            \\ \hline
            Deadlock-Freedom
            & No
            & Yes (typing)
            & Yes (typing)
            & Yes (typing)
            & Yes (via APCP)
            \\ \hline
        \end{tabular}
    \end{center}
    \caption{The features of \ourGV compared to its predecessors.}
    \label{tbl:comparison}
\end{table}

\paragraph{Acknowledgments}
Thanks to Simon Fowler and the anonymous reviewers for their helpful feedback.
We gratefully acknowledge the support of the Dutch Research Council (NWO) under project No.\,016.Vidi.189.046 (Unifying Correctness for Communicating Software).

\bibliographystyle{eptcs}
\bibliography{refs}

\ifappendix

\appendix
\newpage

\tableofcontents

\newpage

\section{Type Preservation for \ourGV: Full Proofs}
\label{a:ourGVtypePresProofs}

\begin{lemma}\label{l:termWeaken}
    If $\type{\Gamma}, \term{x}: \type{U} \vdashM \term{\mbb{M}}: \type{T}$ and $\term{x} \notin \fn(\term{\mbb{M}})$, then $\type{U} = \type{\sff{end}}$ and $\type{\Gamma} \vdashM \term{\mbb{M}}: \type{T}$.
\end{lemma}

\begin{proof}
    The only possibility of having a name $\term{x}$ in the typing context that is not free in $\term{\mbb{M}}$ is by application of Rule~\scc{T-EndL}.
    Hence, $\type{U} = \type{\sff{end}}$.
    Moreover, since the rule does not modify terms, we can simply remove the application of the Rule~\scc{T-EndL} from the typing derivation of $\term{\mbb{M}}$, such that $\type{\Gamma} \vdashM \term{\mbb{M}}: \type{T}$.
\end{proof}

\begin{restatable}[Subject Congruence for Terms]{theorem}{thmTermSubjCong}
\label{t:termSubjCong}
    If $\type{\Gamma} \vdashM \term{\mbb{M}}: \type{T}$ and $\term{\mbb{M}} \equivM \term{\mbb{N}}$, then $\type{\Gamma} \vdashM \term{\mbb{N}}: \type{T}$.
\end{restatable}

\begin{proof}
    By induction on the derivation of $\term{\mbb{M}} \equivM \term{\mbb{N}}$.
    The inductive cases follow from the IH directly.
    We consider the only rule \scc{SC-SubExt}: $\term{x} \notin \fn(\term{\mcl{R}}) \implies \term{(\mcl{R}[\mbb{M}])\xsub{ \mbb{N}/x }} \equivM \term{\mcl{R}[\mbb{M}\xsub{ \mbb{N}/x }]}$.

    We apply induction on the structure of the reduction context $\term{\mcl{R}}$.
    As an interesting, representative case, consider $\term{\mcl{R}} = \term{\mbb{L}\xsub{ \mcl{R}'/y }}$.
    Assuming $\term{x} \notin \fn(\term{\mcl{R}})$, we have $\term{x} \notin \fn(\term{\mbb{L}}) \cup \fn(\term{\mcl{R}'})$.
    We apply inversion of typing:
    \begin{align*}
        \inferrule*{
            \inferrule*{
                \type{\Gamma}, \term{y}: \type{U} \vdashM \term{\mbb{L}}: \type{T}
                \\
                \type{\Delta}, \term{x}: \type{U'} \vdashM \term{\mcl{R}'[\mbb{M}]}: \type{U}
            }{
                \type{\Gamma}, \type{\Delta}, \term{x}: \type{U'} \vdashM \term{\mbb{L}\xsub{ (\mcl{R}'[\mbb{M}])/y }}: \type{T}
            }
            \\
            \type{\Delta'} \vdashM \term{\mbb{N}}: \type{U'}
        }{
            \type{\Gamma}, \type{\Delta}, \type{\Delta'} \vdashM \term{(\mbb{L}\xsub{ (\mcl{R}'[\mbb{M}])/y })\xsub{ \mbb{N}/x }}: \type{T}
        }
    \end{align*}
    We can derive $\inferrule{
        \type{\Delta}, \term{x}: \type{U'} \vdashM \term{\mcl{R}'[\mbb{M}]}: \type{U}
        \\
        \type{\Delta'} \vdashM \term{\mbb{N}}: \type{U'}
    }{
        \type{\Delta}, \type{\Delta'} \vdashM \term{(\mcl{R}'[\mbb{M}])\xsub{ \mbb{N}/x }}: \type{U}
    }$.
    Since $\term{x} \notin \fn(\term{\mcl{R}'})$, by Rule~\scc{SC-SubExt}, $\term{(\mcl{R}'[\mbb{M}])\xsub{ \mbb{N}/x }} \equivM \term{\mcl{R}'[\mbb{M}\xsub{ \mbb{N}/x }]}$.
    Then, by the IH, $\type{\Delta}, \type{\Delta'} \vdashM \term{\mcl{R}'[\mbb{M}\xsub{ \mbb{N}/x }]}: \type{U}$.
    Hence, we can conclude the following:
    \begin{align*}
        \inferrule*{
            \type{\Gamma}, \term{y}: \type{U} \vdashM \term{\mbb{L}}: \type{T}
            \\
            \type{\Delta}, \type{\Delta'} \vdashM \term{\mcl{R}'[\mbb{M}\xsub{ \mbb{N}/x }]}: \type{U}
        }{
            \type{\Gamma}, \type{\Delta}, \type{\Delta'} \vdashM \term{\mbb{L}\xsub{ (\mcl{R}'[\mbb{M}\xsub{ \mbb{N}/x }])/y }}: \type{T}
        }
        \tag*{\qedhere}
    \end{align*}
\end{proof}

\begin{restatable}[Subject Reduction for Terms]{theorem}{thmTermSubjRed}
\label{t:termSubjRed}
    If $\type{\Gamma} \vdashM \term{\mbb{M}}: \type{T}$ and $\term{\mbb{M}} \reddM \term{\mbb{N}}$, then $\type{\Gamma} \vdashM \term{\mbb{N}}: \type{T}$.
\end{restatable}

\begin{proof}
    By induction on the derivation of $\term{\mbb{M}} \reddM \term{\mbb{N}}$ (\ih{1}).
    The case of Rule~\scc{E-Lift} follows by induction on the structure of the reduction context $\term{\mcl{R}}$, where the base case ($\term{\mcl{R}} = \term{[]}$) follows from \ih{1}.
    The case of Rule~\scc{E-LiftSC} follows from \ih{1} and \Cref{t:termSubjCong} (Subject Congruence for Terms).
    We consider the other cases, applying inversion of typing and deriving the typing of the term after reduction:
    \begin{itemize}
        \item
            Rule~\scc{E-Lam}: $\term{(\lambda x \sdot M)\, \mbb{N}} \reddM \term{M\xsub{ \mbb{N}/x }}$.
            \begin{align*}
                \inferrule*{
                    \inferrule*{
                        \type{\Gamma}, \term{x}: \type{T} \vdashM \term{M}: \type{U}
                    }{
                        \type{\Gamma} \vdashM \term{\lambda x \sdot M}: \type{T \lolli U}
                    }
                    \\
                    \type{\Delta} \vdashM \term{\mbb{N}}: \type{T}
                }{
                    \type{\Gamma}, \type{\Delta} \vdashM \term{(\lambda x \sdot M)~\mbb{N}}: \type{U}
                }
                \reddM
                \inferrule*{
                    \type{\Gamma}, \term{x}: \type{T} \vdashM \term{M}: \type{U}
                    \\
                    \type{\Delta} \vdashM \term{\mbb{N}}: \type{T}
                }{
                    \type{\Gamma}, \type{\Delta} \vdashM \term{M\xsub{ \mbb{N}/x }}: \type{U}
                }
            \end{align*}

        \item
            Rule~\scc{E-Pair}: $\term{\sff{let}\, (x,y) = (\mbb{M}_1,\mbb{M}_2)\, \sff{in}\, N} \reddM \term{N\xsub{ \mbb{M}_1/x,\mbb{M}_2/y }}$.
            \begin{mathpar}
                \inferrule*{
                    \inferrule*{
                        \type{\Gamma} \vdashM \term{\mbb{M}_1}: \type{T}
                        \\
                        \type{\Gamma'} \vdashM \term{\mbb{M}_2}: \type{T'}
                    }{
                        \type{\Gamma}, \type{\Gamma'} \vdashM \term{(\mbb{M}_1,\mbb{M}_2)}: \type{T \times T'}
                    }
                    \\
                    \type{\Delta}, \term{x}: \type{T}, \term{y}: \type{T'} \vdashM \term{N}: \type{U}
                }{
                    \type{\Gamma}, \type{\Gamma'}, \type{\Delta} \vdashM \term{\sff{let}\, (x,y) = (\mbb{M}_1,\mbb{M}_2)\, \sff{in}\, N}: \type{U}
                }
                \reddM
                \inferrule*{
                    \inferrule*{
                        \type{\Delta}, \term{x}: \type{T}, \term{y}: \type{T'} \vdashM \term{N}: \type{U}
                        \\
                        \type{\Gamma} \vdashM \term{\mbb{M}_1}: \type{T}
                    }{
                        \type{\Gamma}, \type{\Delta}, \term{y}: \type{T'} \vdashM \term{N\xsub{ \mbb{M}_1/x }}: \type{U}
                    }
                    \\
                    \type{\Gamma'} \vdashM \term{\mbb{M}_2}: \type{T'}
                }{
                    \type{\Gamma}, \type{\Gamma'}, \type{\Delta} \vdashM \term{N\xsub{ \mbb{M}_1/x,\mbb{M}_2/y }}: \type{U}
                }
            \end{mathpar}

        \item
            Rule~\scc{E-SubstName}: $\term{\mbb{M}\xsub{ y/x }} \reddM \term{\mbb{M}\{y/x\}}$.
            \begin{mathpar}
                \inferrule*{
                    \type{\Gamma}, \term{x}: \type{T} \vdashM \term{\mbb{M}}: \type{U}
                    \\
                    \term{y}: \type{T} \vdashM \term{y}: \type{T}
                }{
                    \type{\Gamma}, \term{y}: \type{T} \vdashM \term{\mbb{M}\xsub{ y/x }}: \type{U}
                }
                \reddM
                \type{\Gamma}, \term{y}: \type{T} \vdashM \term{\mbb{M}\{y/x\}}: \type{U}
            \end{mathpar}

        \item
            Rule~\scc{E-NameSubst}: $\term{x\xsub{ \mbb{M}/x }} \reddM \term{\mbb{M}}$.
            \begin{mathpar}
                \inferrule*{
                    \term{x}: \type{U} \vdashM \term{x}: \type{U}
                    \\
                    \type{\Gamma} \vdashM \term{\mbb{M}}: \type{U}
                }{
                    \type{\Gamma} \vdashM \term{x\xsub{ \mbb{M}/x }}: \type{U}
                }
                \reddM
                \type{\Gamma} \vdashM \term{\mbb{M}}: \type{U}
            \end{mathpar}

        \item
            Rule~\scc{E-Send}: $\term{\sff{send}~(\mbb{M},\mbb{N})} \reddM \term{\sff{send}'(\mbb{M},\mbb{N})}$.
            \begin{align*}
                \inferrule*{
                    \inferrule*{
                        \type{\Gamma} \vdashM \term{\mbb{M}}: \type{T}
                        \\
                        \type{\Delta} \vdashM \term{\mbb{N}}: \type{{!}T \sdot S}
                    }{
                        \type{\Gamma}, \type{\Delta} \vdashM \term{(\mbb{M},\mbb{N})}: \type{T \times {!}T \sdot S}
                    }
                }{
                    \type{\Gamma}, \type{\Delta} \vdashM \term{\sff{send}~(\mbb{M},\mbb{N})}: \type{S}
                }
                \reddM
                \inferrule*{
                    \type{\Gamma} \vdashM \term{\mbb{M}}: \type{T}
                    \\
                    \type{\Delta} \vdashM \term{\mbb{N}}: \type{{!}T \sdot S}
                }{
                    \type{\Gamma}, \type{\Delta} \vdashM \term{\sff{send}'(\mbb{M},\mbb{N})}: \type{S}
                }
                \tag*{\qedhere}
            \end{align*}
    \end{itemize}
\end{proof}

\begin{theorem}
    \label{t:confSubjCong}
    If $\type{\Gamma} \vdashC{\phi} \term{C}:\type{T}$ and $\term{C} \equivC \term{D}$, then $\type{\Gamma} \vdashC{\phi} \term{D} : \type{T}$.
\end{theorem}

\begin{proof}
    By induction on the derivation of $\term{C} \equivC \term{D}$.
    The inductive cases follow from the IH directly.
    The case for Rule~\scc{SC-TermSC} follows from \Cref{t:termSubjCong} (Subject Congruence for Terms).
    The cases for Rules~\scc{SC-ResSwap}, \scc{SC-ResComm}, \scc{SC-ParNil}, \scc{SC-ParComm}, and \scc{SC-ParAssoc} are straightforward.
    We consider the other cases:
    \begin{itemize}
        \item
            Rule~\scc{SC-ResExt} $\term{x},\term{y} \notin \fn(\term{C}) \implies \term{\nu{x\bfr{\vec{m}}y}(C \prl D)} \equivC \term{C \prl \nu{x\bfr{\vec{m}}y}D}$.

            The analysis depends on whether $\term{C}$ or $\term{D}$ are child threads.
            W.l.o.g., we assume $\term{C}$ is a child thread.
            Assuming $\term{x},\term{y} \notin \fn(\term{C})$, we apply inversion of typing:
            \begin{mathpar}
                \inferrule*{
                    \type{\Gamma} \vdashB \term{\bfr{\vec{m}}}: \type{S'} > \type{S}
                    \\
                    \inferrule*{
                        \type{\Delta} \vdashC{\circ} \term{C}: \type{\1}
                        \\
                        \type{\Lambda}, \term{x}: \type{S'}, \term{y}: \type{\ol{S}} \vdashC{\phi} \term{D}: \type{T}
                    }{
                        \type{\Delta}, \type{\Lambda}, \term{x}: \type{S'}, \term{y}: \type{\ol{S}} \vdashC{\circ+\phi} \term{C \prl D}: \type{T}
                    }
                }{
                    \type{\Gamma}, \type{\Delta}, \type{\Lambda} \vdashC{\circ+\phi} \term{\nu{x\bfr{\vec{m}}y}(C \prl D)}: \type{T}
                }
            \end{mathpar}
            Then, we derive the typing of the structurally congruent configuration:
            \begin{mathpar}
                \inferrule*{
                    \type{\Delta} \vdashC{\circ} \term{C}: \type{\1}
                    \\
                    \inferrule*{
                        \type{\Gamma} \vdashB \term{\bfr{\vec{m}}}: \type{S'} > \type{S}
                        \\
                        \type{\Lambda}, \term{x}: \type{S'}, \term{y}: \type{\ol{S}} \vdashC{\phi} \term{D}: \type{T}
                    }{
                        \type{\Gamma}, \type{\Lambda} \vdashC{\phi} \term{\nu{x\bfr{\vec{m}}y}D}: \type{T}
                    }
                }{
                    \type{\Gamma}, \type{\Delta}, \type{\Lambda} \vdashC{\circ+\phi} \term{C \prl \nu{x\bfr{\vec{m}}y}D}: \type{T}
                }
            \end{mathpar}

        \item
            Rule~\scc{SC-ResNil}: $\term{x},\term{y} \notin \fn(\term{C}) \implies \term{\nu{x\bfr{\epsilon}y}C} \equivC \term{C}$.

            Assuming $\term{x},\term{y} \notin \fn(\term{C})$, we apply inversion of typing:
            \begin{mathpar}
                \inferrule*{
                    \inferrule*{ }{
                        \type{\emptyset} \vdashB \term{\bfr{\epsilon}}: \type{S} > \type{S}
                    }
                    \\
                    \type{\Gamma}, \term{x}: \type{S}, \term{y}: \type{\ol{S}} \vdashC{\phi} \term{C}: \type{T}
                }{
                    \type{\Gamma} \vdashC{\phi} \term{\nu{x\bfr{\epsilon}y}C}: \type{T}
                }
            \end{mathpar}
            By induction on the structure of $\term{C}$, we show that $\type{\Gamma} \vdashC{\phi} \term{C}: \type{T}$, proving the thesis.
            The inductive cases follow from the IH directly.
            The base cases (Rules~\scc{T-Main} and \scc{T-Child}) follow from \Cref{l:termWeaken}.

        \item
            Rule~\scc{SC-Send'}: $\term{\nu{x\bfr{\vec{m}}y}(\hat{\mcl{F}}[\sff{send}'(M,x)] \prl C)} \equivC \term{\nu{x\bfr{M,\vec{m}}y}(\hat{\mcl{F}}[x] \prl C)}$.

            This case follows by induction on the structure of $\term{\hat{\mcl{F}}}$.
            The inductive cases follow from the IH straightforwardly.
            The fact that the hole in $\term{\hat{\mcl{F}}}$ does not occur under an explicit substitution, guarantees that we can move $\term{M}$ out of the context of $\term{\hat{\mcl{F}}}$ and into the buffer.
            We consider the base case ($\term{\hat{\mcl{F}}} = \term{\phi\,[]}$).
            We apply inversion of typing, w.l.o.g.\ assuming that $\term{\phi} = \term{\bullet}$ and $\term{y} \in \fn(\term{C})$:
            \begin{mathpar}
                \inferrule*{
                    \type{\Gamma} \vdashB \term{\bfr{\vec{m}}}: \type{{!}T \sdot S'} > \type{S}
                    \\
                    \inferrule*{
                        \inferrule*{
                            \inferrule*{
                                \type{\Delta} \vdashM \term{M}: \type{T}
                                \\
                                \inferrule*{ }{
                                    \term{x}: \type{{!}T \sdot S'} \vdashM \term{x}: \type{{!}T \sdot S'}
                                }
                            }{
                                \type{\Delta}, \term{x}: \type{{!}T \sdot S'} \vdashM \term{\sff{send}'(M,x)}: \type{S'}
                            }
                        }{
                            \type{\Delta}, \term{x}: \type{{!}T \sdot S'} \vdashC{\bullet} \term{\bullet\,\sff{send}'(M,x)}: \type{S'}
                        }
                        \\
                        \type{\Lambda}, \term{y}: \type{\ol{S}} \vdashC{\circ} \term{C}: \type{\1}
                    }{
                        \type{\Delta}, \type{\Lambda}, \term{x}: \type{{!}T \sdot S'}, \term{y}: \type{\ol{S}} \vdashC{\bullet} \term{\bullet\,\sff{send}'(M,x) \prl C}: \type{S'}
                    }
                }{
                    \type{\Gamma}, \type{\Delta}, \type{\Lambda} \vdashC{\bullet} \term{\nu{x\bfr{\vec{m}}y}(\bullet\,\sff{send}'(M,x) \prl C)}: \type{S'}
                }
            \end{mathpar}

            Note that the derivation of $\type{\Gamma} \vdashB \term{\bfr{\vec{m}}}: \type{{!}T \sdot S'} > \type{S}$ depends on the size of $\term{\vec{m}}$.
            By induction on the size of $\term{\vec{m}}$ (\ih{2}), we derive $\type{\Gamma}, \type{\Delta} \vdashB \term{\bfr{M,\vec{m}}}: \type{S'} > \type{S}$:
            \begin{itemize}
                \item
                    If $\term{\vec{m}}$ is empty, it follows by inversion of typing that $\type{\Gamma} = \type{\emptyset}$ and $\type{S} = \type{{!}T \sdot S'}$:
                    \begin{mathpar}
                        \inferrule*{ }{
                            \type{\emptyset} \vdashB \term{\bfr{\epsilon}}: \type{{!}T \sdot S'} > \type{{!}T \sdot S'}
                        }
                    \end{mathpar}
                    Then, we derive the following:
                    \begin{mathpar}
                        \inferrule*{
                            \type{\Delta} \vdashM \term{M}: \type{T}
                            \\
                            \inferrule*{ }{
                                \type{\emptyset} \vdashB \term{\bfr{\epsilon}}: \type{S'} > \type{S'}
                            }
                        }{
                            \type{\Delta} \vdashB \term{\bfr{M}}: \type{S'} > \type{{!}T \sdot S'}
                        }
                    \end{mathpar}

                \item
                    If $\term{\vec{m}} = \term{\vec{m}',N}$, it follows by inversion of typing that $\type{\Gamma} = \type{\Gamma'},\type{\Gamma''}$ and $\type{S} = \type{{!}T' \sdot S''}$:
                    \begin{mathpar}
                        \inferrule*{
                            \type{\Gamma'} \vdashM \term{N}: \type{T'}
                            \\
                            \type{\Gamma''} \vdashB \term{\bfr{\vec{m}'}}: \type{{!}T \sdot S'} > \type{S''}
                        }{
                            \type{\Gamma'}, \type{\Gamma''} \vdashB \term{\bfr{\vec{m}',N}}: \type{{!}T \sdot S'} > \type{{!}T' \sdot S''}
                        }
                    \end{mathpar}
                    By \ih{2}, $\type{\Gamma''}, \type{\Delta} \vdashB \term{\bfr{M,\vec{m}'}}: \type{S'} > \type{S''}$, allowing us to derive the following:
                    \begin{mathpar}
                        \inferrule*{
                            \type{\Gamma'} \vdashM \term{N}: \type{T'}
                            \\
                            \type{\Gamma''}, \type{\Delta} \vdashB \term{\bfr{M,\vec{m}'}}: \type{S'} > \type{S''}
                        }{
                            \type{\Gamma'}, \type{\Gamma''}, \type{\Delta} \vdashB \term{\bfr{M,\vec{m}',N}}: \type{S'} > \type{{!}T' \sdot S''}
                        }
                    \end{mathpar}

                \item
                    If $\term{\vec{m}} = \term{\vec{m}', j}$, it follows by inversion of typing that there exist types $\type{S_i}$ for each $i$ in a set of labels $I$, where $j \in I$, such that $\type{S} = \type{\oplus\{i:S_i\}_{i \in I}}$:
                    \begin{mathpar}
                        \inferrule*{
                            \type{\Gamma} \vdashB \term{\bfr{\vec{m}'}}: \type{{!}T \sdot S'} > \type{S_j}
                        }{
                            \type{\Gamma} \vdashB \term{\bfr{\vec{m},j}}: \type{{!}T \sdot S'} > \type{\oplus\{i:S_i\}_{i \in I}}
                        }
                    \end{mathpar}
                    By \ih{2}, $\type{\Gamma}, \type{\Delta} \vdashB \term{\bfr{M,\vec{m}'}}: \type{S'} > \type{S_j}$, so we derive the following:
                    \begin{mathpar}
                        \inferrule*{
                            \type{\Gamma}, \type{\Delta} \vdashB \term{\bfr{M,\vec{m}'}}: \type{S'} > \type{S_j}
                        }{
                            \type{\Gamma}, \type{\Delta} \vdashB \term{\bfr{M,\vec{m}',j}}: \type{S'} > \type{\oplus\{i:S_i\}_{i \in I}}
                        }
                    \end{mathpar}
            \end{itemize}

            Now, we can derive the typing of the structurally congruent configuration:
            \begin{mathpar}
                \inferrule*{
                    \type{\Gamma}, \type{\Delta} \vdashB \term{\bfr{M,\vec{m}}}: \type{S'} > \type{S}
                    \\
                    \inferrule*{
                        \inferrule*{
                            \inferrule*{ }{
                                \term{x}: \type{S'} \vdashM \term{x}: \type{S'}
                            }
                        }{
                            \term{x}: \type{S'} \vdashC{\bullet} \term{\bullet\,x}: \type{S'}
                        }
                        \\
                        \type{\Lambda}, \term{y}: \type{\ol{S}} \vdashC{\circ} \term{C}: \type{\1}
                    }{
                        \type{\Lambda}, \term{x}: \type{S'}, \term{y}: \type{\ol{S}} \vdashC{\bullet} \term{\bullet\,x \prl C}: \type{S'}
                    }
                }{
                    \type{\Gamma}, \type{\Delta}, \type{\Lambda} \vdashC{\bullet} \term{\nu{x\bfr{M,\vec{m}}y}(\bullet\,x \prl C)}: \type{S'}
                }
            \end{mathpar}

        \item
            The case for Rule~\scc{SC-Select} ($\term{\nu{x\bfr{\vec{m}}y}(\mcl{F}[\sff{select}\, \ell\, x] \prl C)} \equivC \term{\nu{x\bfr{\ell,\vec{m}}y}(\mcl{F}[x] \prl C)}$) is similar to the case above.
            In this case, there is no restriction on the context $\term{\mcl{F}}$: it is fine for the selection to occur under an explicit substitution, because the rule is not moving anything out of the scope of the context.

        \item
            Rule~\scc{SC-ConfSubst}: $\term{\phi\,(\mbb{M}\xsub{ \mbb{N}/x })} \equivC \term{(\phi\,\mbb{M})\xsub{ \mbb{N}/x }}$.

            This case follows by a straightforward inversion of typing on both terms:
            \begin{mathpar}
                \inferrule*{
                    \inferrule*{
                        \type{\Gamma}, \term{x}: \type{T} \vdashM \term{\mbb{M}}: \type{U}
                        \\
                        \type{\Delta} \vdashM \term{N}: \type{T}
                    }{
                        \type{\Gamma}, \type{\Delta} \vdashM \term{\mbb{M}\xsub{ \mbb{N}/x }}: \type{U}
                    }
                }{
                    \type{\Gamma}, \type{\Delta} \vdashC{\phi} \term{\phi\,(\mbb{M}\xsub{ \mbb{N}/x })}: \type{U}
                }
                \equiv
                \inferrule*{
                    \inferrule*{
                        \type{\Gamma}, \term{x}: \type{T} \vdashM \term{\mbb{M}}: \type{U}
                    }{
                        \type{\Gamma}, \term{x}: \type{T} \vdashC{\phi} \term{\phi\,\mbb{M}}: \type{U}
                    }
                    \\
                    \type{\Delta} \vdashM \term{N}: \type{T}
                }{
                    \type{\Gamma}, \type{\Delta} \vdashC{\phi} \term{(\phi\,\mbb{M})\xsub{ \mbb{N}/x }}: \type{U}
                }
            \end{mathpar}

        \item
            Rule~\scc{SC-ConfSubstExt}: $\term{x} \notin \fn(\term{\mcl{G}}) \implies \term{(\mcl{G}[C])\xsub{ \mbb{M}/x }} \equivC \term{\mcl{G}[C\xsub{ \mbb{M}/x }]}$.

            This case follows by induction on the structure of $\term{\mcl{G}}$.
            The inductive cases follow from the IH straightforwardly.
            For the base case ($\term{\mcl{G}} = \term{[]}$), the structural congruence is simply an equality.
            \qedhere
    \end{itemize}
\end{proof}

\begin{restatable}
{theorem}{thmConfSubjRed}
\label{t:confSubjRed}
    If $\type{\Gamma} \vdashC{\phi} \term{C}: \type{T}$ and $\term{C} \reddC \term{D}$, then $\type{\Gamma} \vdashC{\phi} \term{D}: \type{T}$.
\end{restatable}

\begin{proof}
    By induction on the derivation of $\term{C} \reddC \term{D}$ (\ih{1}).
    The case of Rule~\scc{E-LiftC} ($\term{C} \reddC \term{C'} \implies \term{\mcl{G}[C]} \reddC \term{\mcl{G}[C']}$) follows by induction on the structure of $\term{\mcl{G}}$, directly from \ih{1}.
    The case of Rule~\scc{E-LiftM} ($\term{\mbb{M}} \reddM \term{\mbb{M}'} \implies \term{\mcl{F}[\mbb{M}]} \reddC \term{\mcl{F}[\mbb{M'}]}$) follows by induction on the structure of $\term{\mcl{F}}$, where the base case ($\term{\mcl{F}} = \term{\phi\,[]}$) follows from \Cref{t:termSubjRed} (Subject Reduction for Terms).
    The case for Rule~\mbox{\scc{E-ConfLiftSC}} ($\term{C} \equivC \term{C'} \wedge \term{C'} \reddC \term{D'} \wedge \term{D'} \equivC \term{D} \implies \term{C} \reddC \term{D}$) follows from \ih{1} and \Cref{t:confSubjCong} (Subject Congruence for Configurations).
    We consider the other cases:
    \begin{itemize}
        \item
            Rule~\scc{E-New}: $\term{\mcl{F}[\sff{new}]} \reddC \term{\nu{x\bfr{\epsilon}y}(\mcl{F}[(x,y)])}$.

            This case follows by induction on the structure of $\term{\mcl{F}}$.
            The inductive cases follow from the IH directly.
            For the base case ($\term{\mcl{F}} = \term{\phi\,[]}$), we apply inversion of typing and derive the typing of the reduced configuration:
            \begin{mathpar}
                \inferrule*{
                    \inferrule*{
                    }{
                        \type{\emptyset} \vdashM \term{\sff{new}}: \type{S \times \ol{S}}
                    }
                }{
                    \type{\emptyset} \vdashC{\phi} \term{\phi\,\sff{new}}: \type{S \times \ol{S}}
                }
                \reddC
                \inferrule*{
                    \inferrule*{ }{
                        \type{\emptyset} \vdashB \term{\bfr{\epsilon}}: \type{S} > \type{S}
                    }
                    \\
                    \inferrule*{
                        \inferrule*{
                            \inferrule*{ }{
                                \term{x}: \type{S} \vdashM \term{x}: \type{S}
                            }
                            \\
                            \inferrule*{ }{
                                \term{y}: \type{\ol{S}} \vdashM \term{y}: \type{\ol{S}}
                            }
                        }{
                            \term{x}: \type{S}, \term{y}: \type{\ol{S}} \vdashM \term{(x,y)}: \type{S \times \ol{S}}
                        }
                    }{
                        \term{x}: \type{S}, \term{y}: \type{\ol{S}} \vdashC{\phi} \term{\phi\,(x,y)}: \type{S \times \ol{S}}
                    }
                }{
                    \type{\emptyset} \vdashC{\phi} \term{\nu{x\bfr{\epsilon}y}(\phi\,(x,y))}: \type{S \times \ol{S}}
                }
            \end{mathpar}

        \item
            Rule~\scc{E-Spawn}: $\term{\hat{\mcl{F}}[\sff{spawn}~(M,N)]} \reddC \term{\hat{\mcl{F}}[N] \prl \circ\, M}$.

            Similar to the case above, this case follows by induction on the structure of $\term{\hat{\mcl{F}}}$, which excludes holes under explicit substitution.
            For the base case ($\term{\hat{\mcl{F}}} = \term{\phi\,[]}$), we apply inversion of typing and derive the typing of the reduced configuration:
            \begin{mathpar}
                \inferrule*{
                    \inferrule*{
                        \inferrule*{
                            \type{\Gamma} \vdashM \term{M}: \type{\1}
                            \\
                            \type{\Delta} \vdashM \term{N}: \type{T}
                        }{
                            \type{\Gamma}, \type{\Delta} \vdashM \term{(M,N)}: \type{\1 \times T}
                        }
                    }{
                        \type{\Gamma}, \type{\Delta} \vdashM \term{\sff{spawn}~(M,N)}: \type{T}
                    }
                }{
                    \type{\Gamma}, \type{\Delta} \vdashC{\phi} \term{\phi\,(\sff{spawn}~(M,N))}: \type{T}
                }
                \reddC
                \inferrule*{
                    \inferrule*{
                        \type{\Delta} \vdashM \term{N}: \type{T}
                    }{
                        \type{\Delta} \vdashC{\phi} \term{\phi\,N}: \type{T}
                    }
                    \\
                    \inferrule*{
                        \type{\Gamma} \vdashM \term{M}: \type{\1}
                    }{
                        \type{\Gamma} \vdashC{\circ} \term{\circ\,M}: \type{\1}
                    }
                }{
                    \type{\Gamma}, \type{\Delta} \vdashC{\phi} \term{\phi\,N \prl \circ\,M}: \type{T}
                }
            \end{mathpar}

        \item
            Rule~\scc{E-Recv}: $\term{\nu{x\bfr{\vec{m},M}y}(\mcl{F}[\sff{recv}~y] \prl C)} \reddC \term{\nu{x\bfr{\vec{m}}y}(\mcl{F}[(M,y)] \prl C)}$.

            For this case, we apply induction on the structure of $\term{\mcl{F}}$.
            The inductive cases follow from the IH directly.
            We consider the base case ($\term{\mcl{F}} = \term{\phi\,[]}$).
            We apply inversion of typing, w.l.o.g.\ assuming that $\term{\phi} = \term{\bullet}$, and then derive the typing of the reduced configuration:
            \begin{mathpar}
                \inferrule*{
                    \inferrule*{
                        \type{\Gamma} \vdashM \term{M}: \type{T}
                        \\
                        \type{\Delta} \vdashB \term{\bfr{\vec{m}}}: \type{S'} > \type{S}
                    }{
                        \type{\Gamma}, \type{\Delta} \vdashB \term{\bfr{\vec{m},M}}: \type{S'} > \type{{!}T \sdot S}
                    }
                    \\
                    \inferrule*{
                        \inferrule*{
                            \inferrule*{
                                \inferrule*{ }{
                                    \term{y}: \type{{?}T \sdot \ol{S}} \vdashM \term{y}: \type{{?}T \sdot \ol{S}}
                                }
                            }{
                                \term{y}: \type{{?}T \sdot \ol{S}} \vdashM \term{\sff{recv}~y}: \type{T \times \ol{S}}
                            }
                        }{
                            \term{y}: \type{{?}T \sdot \ol{S}} \vdashC{\bullet} \term{\bullet\,(\sff{recv}~y)}: \type{T \times \ol{S}}
                        }
                        \\
                        \type{\Lambda}, \term{x}: \type{S'} \vdashC{\circ} \term{C}: \type{\1}
                    }{
                        \type{\Lambda}, \term{x}: \type{S'}, \term{y}: \type{{?}T \sdot \ol{S}} \vdashC{\bullet} \term{\bullet\,(\sff{recv}~y) \prl C}: \type{T \times \ol{S}}
                    }
                }{
                    \type{\Gamma}, \type{\Delta}, \type{\Lambda} \vdashC{\bullet} \term{\nu{x\bfr{\vec{m},M}y}(\bullet\,(\sff{recv}~y) \prl C)}: \type{T \times \ol{S}}
                }
                \reddC
                \inferrule*{
                    \type{\Delta} \vdashB \term{\bfr{\vec{m}}}: \type{S'} > \type{S}
                    \\
                    \inferrule*{
                        \inferrule*{
                            \inferrule*{
                                \type{\Gamma} \vdashM \term{M}: \type{T}
                                \\
                                \inferrule*{ }{
                                    \term{y}: \type{\ol{S}} \vdashM \term{y}: \type{\ol{S}}
                                }
                            }{
                                \type{\Gamma}, \term{y}: \type{\ol{S}} \vdashM \term{(M,y)}: \type{T \times \ol{S}}
                            }
                        }{
                            \type{\Gamma}, \term{y}: \type{\ol{S}} \vdashC{\bullet} \term{\bullet\,(M,y)}: \type{T \times \ol{S}}
                        }
                        \\
                        \type{\Lambda}, \term{x}: \type{S'} \vdashC{\circ} \term{C}: \type{\1}
                    }{
                        \type{\Gamma}, \type{\Lambda}, \term{x}: \type{S'}, \term{y}: \type{\ol{S}} \vdashC{\bullet} \term{\bullet\,(M,y) \prl C}: \type{T \times \ol{S}}
                    }
                }{
                    \type{\Gamma}, \type{\Delta}, \type{\Lambda} \vdashC{\bullet} \term{\nu{x\bfr{\vec{m}}y}(\bullet\,(M,y) \prl C)}: \type{T \times \ol{S}}
                }
            \end{mathpar}

        \item
            The case of Rule~\scc{E-Case} ($j \in I \implies \term{\nu{x\bfr{\vec{m},j}y}(\mcl{F}[\sff{case}\, y\, \sff{of}\, \{i:M_i\}_{i \in I}] \prl C)} \reddC \term{\nu{x\bfr{\vec{m}}y}(\mcl{F}[M_j~y] \prl C)}$) is similar to the above case.
            \qedhere
    \end{itemize}
\end{proof}

\section{Translations of \ourGV Typing Rules with Full Derivations}\label{a:transFull}

Here we give the translation of \ourGV into APCP, presented in \Cref{ss:ourGVintoAPCP}, with full derivation trees.
\begin{itemize}
    \item
        \refwpage{f:transTerm1} gives the translations of the Rules~\scc{T-Var}, \scc{T-Unit}, \scc{T-Abs}, \scc{T-App}, and \scc{T-New}.

    \item
        \refwpage{f:transTerm2} gives the translations of the Rules~\scc{T-Pair}, \scc{T-Spawn}, and \scc{T-Sub}.

    \item
        \refwpage{f:transTerm3} gives the translations of the Rules~\scc{T-Split}, \scc{T-EndL}, and \scc{T-EndR}.

    \item
        \refwpage{f:transTerm4} gives the translation of the Rules~\scc{T-Send}, and \scc{T-Recv}.

    \item
        \refwpage{f:transTerm5} gives the translations of the Rules~\scc{T-Select}, \scc{T-Case}, and \scc{T-Send'}.

    \item
        \refwpage{f:transConf} gives the translations of the typing rules for configurations.

    \item
        \refwpage{f:transBuf} gives the translations of the typing rules for buffers.
\end{itemize}

\begin{figure}[t!]
    \begin{mdframed}\small
        \begin{mathpar}
            \enc{z}{
                \inferrule[T-Var]{ }{
                    \term{x}: \type{T} \vdashM \term{x}: \type{T}
                }
            }
            = \inferrule{ }{
                x \fwd z \vdashAst x: \ol{\enct{\type{T}}}, z: \enct{\type{T}}
            }
            \and
            \enc{z}{
                \inferrule[T-Unit]{ }{
                    \type{\emptyset} \vdashM \term{()}: \type{\1}
                }
            }
            = \inferrule{
                \inferrule{ }{
                    \0 \vdashAst \emptyset
                }
            }{
                \0 \vdashAst z: \bullet
            }
        \end{mathpar}

        \begin{mathpar}
            \enc{z}{
                \inferrule[T-Abs]{
                    \type{\Gamma}, \term{x}: \type{T} \vdashM \term{M}: \type{U}
                }{
                    \type{\Gamma} \vdashM \term{\lambda x \sdot M}: \type{T \lolli U}
                }
            }
            = \inferrule{
                \inferrule{
                    \inferrule{
                        \inferrule*{
                            \inferrule{
                                \inferrule{ }{
                                    a[c,e] \vdashAst a: \ol{\enct{\type{T}}} \tensor \bullet, c: \enct{\type{T}}, e: \bullet
                                }
                            }{
                                a[c,e] \vdashAst a: \ol{\enct{\type{T}}} \tensor \bullet, c: \enct{\type{T}}, e: \bullet, f: \bullet
                            }
                        }{
                            \nu{ef}a[c,e] \vdashAst a: \ol{\enct{\type{T}}} \tensor \bullet, c: \enct{\type{T}}
                        }
                        \\
                        \encc{b}{M} \vdashAst \ol{\enct{\type{\Gamma}}}, x: \ol{\enct{\type{T}}}, b: \enct{\type{U}}
                    }{
                        \nu{ef}a[c,e] \| \encc{b}{M} \vdashAst \ol{\enct{\type{\Gamma}}}, a: \ol{\enct{\type{T}}} \tensor \bullet, b: \enct{\type{U}}, c: \enct{\type{T}}, x: \ol{\enct{\type{T}}}
                    }
                }{
                    \nuf{cx}(\nu{ef}a[c,e] \| \encc{b}{M}) \vdashAst \ol{\enct{\type{\Gamma}}}, a: \ol{\enct{\type{T}}} \tensor \bullet, b: \enct{\type{U}}
                }
            }{
                z(a,b) \sdot \nuf{cx}(\nu{ef}a[c,e] \| \encc{b}{M}) \vdashAst \ol{\enct{\type{\Gamma}}}, z: (\ol{\enct{\type{T}}} \tensor \bullet) \parr \enct{\type{U}}
            }
        \end{mathpar}

        \begin{mathpar}
            \enc{z}{
                \inferrule[T-App]{
                    \type{\Gamma} \vdashM \term{M}: \type{T \lolli U}
                    \\
                    \type{\Delta} \vdashM \term{N}: \type{T}
                }{
                    \type{\Gamma}, \type{\Delta} \vdashM \term{M~N}: \type{U}
                }
            }
            = \inferrule{
                \inferrule*{
                    \inferrule*{}{
                        {\begin{array}[t]{@{}l@{}}
                                \encc{a}{M} \vdashAst \ol{\enct{\type{\Gamma}}},
                                \\
                                a: {\begin{array}[t]{@{}l@{}}
                                        (\ol{\enct{\type{T}}} \tensor \bullet)
                                        \\
                                        {} \parr \enct{\type{U}}
                                \end{array}}
                        \end{array}}
                    }
                    \\
                    \inferrule*{
                        \inferrule*{
                            \inferrule*{ }{
                                b[c,z] \vdashAst {\begin{array}[t]{@{}l@{}}
                                        b: (\enct{\type{T}} \parr \bullet) \tensor \ol{\enct{\type{U}}},
                                        \\
                                        c: \ol{\enct{\type{T}}} \tensor \bullet, z: \enct{\type{U}}
                                \end{array}}
                            }
                            \\
                            \inferrule*{
                                \inferrule*{
                                    \encc{e}{N} \vdashAst \ol{\enct{\type{\Delta}}}, e: \enct{\type{T}}
                                }{
                                    \encc{e}{N} \vdashAst \ol{\enct{\type{\Delta}}}, e: \enct{\type{T}}, f: \bullet
                                }
                            }{
                                d(e,f) \sdot \encc{e}{N} \vdashAst \ol{\enct{\type{\Delta}}}, d: \enct{\type{T}} \parr \bullet
                            }
                        }{
                            b[c,z] \| d(e,f) \sdot \encc{e}{N} \vdashAst \ol{\enct{\type{\Delta}}}, z: \enct{\type{U}}, {\begin{array}[t]{@{}l@{}}
                                    b: (\enct{\type{T}} \parr \bullet) \tensor \ol{\enct{\type{U}}},
                                    \\
                                    c: \ol{\enct{\type{T}}} \tensor \bullet, d: \enct{\type{T}} \parr \bullet
                            \end{array}}
                        }
                    }{
                        \nu{cd}(b[c,z] \| d(e,f) \sdot \encc{e}{N}) \vdashAst \ol{\enct{\type{\Delta}}}, z: \enct{\type{U}}, b: (\enct{\type{T}} \parr \bullet) \tensor \ol{\enct{\type{U}}}
                    }
                }{
                    \encc{a}{M} \| \nu{cd}(b[c,z] \| d(e,f) \sdot \encc{e}{N}) \vdashAst \ol{\enct{\type{\Gamma}}}, \ol{\enct{\type{\Delta}}}, z: \enct{\type{U}}, a: (\ol{\enct{\type{T}}} \tensor \bullet) \parr \enct{\type{U}}, b: (\enct{\type{T}} \parr \bullet) \tensor \ol{\enct{\type{U}}}
                }
            }{
                \nu{ab}(\encc{a}{M} \| \nu{cd}(b[c,z] \| d(e,f) \sdot \encc{e}{N})) \vdashAst \ol{\enct{\type{\Gamma}}}, \ol{\enct{\type{\Delta}}}, z: \enct{\type{U}}
            }
        \end{mathpar}

        \begin{mathpar}
            \enc{z}{
                \inferrule[T-New]{
                }{
                    \type{\emptyset} \vdashM \term{\sff{new}}: \type{S \times \ol{S}}
                }
            }
            = \inferrule{
                \inferrule{
                    \inferrule*{
                        \inferrule{ }{
                            a[c,d] \vdashAst {\begin{array}[t]{@{}l@{}}
                                    a: \bullet \tensor \bullet,
                                    \\
                                    c: \bullet, d: \bullet
                            \end{array}}
                        }
                    }{
                        \nu{cd}a[c,d] \vdashAst a: \bullet \tensor \bullet
                    }
                    \\
                    \inferrule*{
                        \inferrule{
                            \inferrule{
                                \inferrule*{
                                    \inferrule{ }{
                                        \encc{z}{(x,y)} \vdashAst z: (\enct{\type{S}} \parr \bullet) \tensor (\enct{\type{\ol{S}}} \parr \bullet), x: \ol{\enct{\type{S}}}, y: \enct{\type{S}}
                                    }
                                }{
                                    \nu{xy}\encc{z}{(x,y)} \vdashAst z: (\enct{\type{S}} \parr \bullet) \tensor (\enct{\type{\ol{S}}} \parr \bullet)
                                }
                            }{
                                \nu{xy}\encc{z}{(x,y)} \vdashAst z: (\enct{\type{S}} \parr \bullet) \tensor (\enct{\type{\ol{S}}} \parr \bullet), e:\bullet
                            }
                        }{
                            \nu{xy}\encc{z}{(x,y)} \vdashAst z: (\enct{\type{S}} \parr \bullet) \tensor (\enct{\type{\ol{S}}} \parr \bullet), e: \bullet, f: \bullet
                        }
                    }{
                        b(e,f) \sdot \nu{xy}\encc{z}{(x,y)} \vdashAst z: (\enct{\type{S}} \parr \bullet) \tensor (\enct{\type{\ol{S}}} \parr \bullet), b: \bullet \parr \bullet
                    }
                }{
                    \nu{cd}a[c,d] \| b(e,f) \sdot \nu{xy}\encc{z}{(x,y)} \vdashAst z: (\enct{\type{S}} \parr \bullet) \tensor (\enct{\type{\ol{S}}} \parr \bullet), a: \bullet \tensor \bullet, b: \bullet \parr \bullet
                }
            }{
                \nu{ab}(\nu{cd}a[c,d] \| b(e,f) \sdot \nu{xy}\encc{z}{(x,y)}) \vdashAst z: (\enct{\type{S}} \parr \bullet) \tensor (\enct{\type{\ol{S}}} \parr \bullet)
            }
        \end{mathpar}
    \end{mdframed}
    \caption{Translation of (runtime) term typing rules with full derivations (part 1 of 5).}\label{f:transTerm1}
\end{figure}

\begin{figure}[t!]
    \begin{mdframed}\small
        \begin{mathpar}
            \enc{z}{
                \inferrule[T-Pair]{
                    \type{\Gamma} \vdashM \term{M}: \type{T}
                    \\
                    \type{\Delta} \vdashM \term{N}: \type{U}
                }{
                    \type{\Gamma}, \type{\Delta} \vdashM \term{(M,N)}: \type{T \times U}
                }
            }
            = \inferrule{
                \inferrule{
                    \inferrule*{
                        \inferrule*{ }{
                            z[a,c] \vdashAst {\begin{array}[t]{@{}l@{}}
                                    z: {\begin{array}[t]{@{}l@{}}
                                            (\enct{\type{T}} \parr \bullet)
                                            \\
                                            {} \tensor (\enct{\type{U}} \parr \bullet),
                                    \end{array}}
                                    \\
                                    a: \ol{\enct{\type{T}}} \tensor \bullet,
                                    \\
                                    c: \ol{\enct{\type{U}}} \tensor \bullet
                            \end{array}}
                        }
                        \\
                        \inferrule*{
                            \inferrule*{
                                \encc{e}{M} \vdashAst \ol{\enct{\type{\Gamma}}}, e: \enct{\type{T}}
                            }{
                                \encc{e}{M} \vdashAst \ol{\enct{\type{\Gamma}}}, e: \enct{\type{T}}, f: \bullet
                            }
                        }{
                            b(e,f) \sdot \encc{e}{M} \vdashAst \ol{\enct{\type{\Gamma}}}, b: \enct{\type{T}} \parr \bullet
                        }
                        \\
                        \inferrule*{
                            \inferrule*{
                                \encc{g}{N} \vdashAst \ol{\enct{\type{\Delta}}}, g: \enct{\type{U}}
                            }{
                                \encc{g}{N} \vdashAst \ol{\enct{\type{\Delta}}}, g: \enct{\type{U}}, h: \bullet
                            }
                        }{
                            d(g,h) \sdot \encc{g}{N} \vdashAst \ol{\enct{\type{\Delta}}}, d: \enct{\type{U}} \parr \bullet
                        }
                    }{
                        z[a,c] \| b(e,f) \sdot \encc{e}{M} \| d(g,h) \sdot \encc{g}{N} \vdashAst {\begin{array}[t]{@{}l@{}}
                                \ol{\enct{\type{\Gamma}}}, \ol{\enct{\type{\Delta}}}, z: (\enct{\type{T}} \parr \bullet) \tensor (\enct{\type{U}} \parr \bullet),
                                \\
                                a: \ol{\enct{\type{T}}} \tensor \bullet, b: \enct{\type{T}} \parr \bullet, c: \ol{\enct{\type{U}}} \tensor \bullet, d: \enct{\type{U}} \parr \bullet
                        \end{array}}
                    }
                }{
                    \nu{cd}(z[a,c] \| b(e,f) \sdot \encc{e}{M} \| d(g,h) \sdot \encc{g}{N}) \vdashAst \ol{\enct{\type{\Gamma}}}, \ol{\enct{\type{\Delta}}}, z: (\enct{\type{T}} \parr \bullet) \tensor (\enct{\type{U}} \parr \bullet), a: \ol{\enct{\type{T}}} \tensor \bullet, b: \enct{\type{T}} \parr \bullet
                }
            }{
                \nu{ab}\nu{cd}(z[a,c] \| b(e,f) \sdot \encc{e}{M} \| d(g,h) \sdot \encc{g}{N}) \vdashAst \ol{\enct{\type{\Gamma}}}, \ol{\enct{\type{\Delta}}}, z: (\enct{\type{T}} \parr \bullet) \tensor (\enct{\type{U}} \parr \bullet)
            }
        \end{mathpar}

        \begin{mathpar}
            \enc{z}{
                \inferrule[T-Spawn]{
                    \type{\Gamma} \vdashM \term{M}: \type{\1 \times T}
                }{
                    \type{\Gamma} \vdashM \term{\sff{spawn}~M}: \type{T}
                }
            }
            \mprset{sep=1.6em}
            = \inferrule{
                \inferrule*{
                    \inferrule*{}{
                        \encc{a}{M} \vdashAst {\begin{array}[t]{@{}l@{}}
                                \ol{\enct{\type{\Gamma}}},
                                \\
                                a: {\begin{array}[t]{@{}l@{}}
                                        (\bullet \parr \bullet)
                                        \\
                                        {} \tensor (\enct{\type{T}} \parr \bullet)
                                \end{array}}
                        \end{array}}
                    }
                    \\
                    \inferrule*{
                        \inferrule*{
                            \inferrule*{
                                \inferrule*{ }{
                                    c[e,f] \vdashAst c: \bullet \tensor \bullet, e: \bullet, f: \bullet
                                }
                            }{
                                \nu{ef}c[e,f] \vdashAst c: \bullet \tensor \bullet
                            }
                            \\
                            \inferrule*{
                                \inferrule*{
                                    \inferrule*{ }{
                                        d[z,g] \vdashAst z: \enct{\type{T}}, d: \ol{\enct{\type{T}}} \tensor \bullet, g: \bullet
                                    }
                                }{
                                    d[z,g] \vdashAst z: \enct{\type{T}}, d: \ol{\enct{\type{T}}} \tensor \bullet, g: \bullet, h: \bullet
                                }
                            }{
                                \nu{gh}d[z,g] \vdashAst z: \enct{\type{T}}, d: \ol{\enct{\type{T}}} \tensor \bullet
                            }
                        }{
                            \nu{ef}c[e,f] \| \nu{gh}d[z,g] \vdashAst z: \enct{\type{T}}, c: \bullet \tensor \bullet, d: \ol{\enct{\type{T}}} \tensor \bullet
                        }
                    }{
                        b(c,d) \sdot (\nu{ef}c[e,f] \| \nu{gh}d[z,g]) \vdashAst z: \enct{\type{T}}, b: (\bullet \tensor \bullet) \parr (\ol{\enct{\type{T}}} \tensor \bullet)
                    }
                }{
                    \encc{a}{M} \| b(c,d) \sdot (\nu{ef}c[e,f] \| \nu{gh}d[z,g]) \vdashAst \ol{\enct{\type{\Gamma}}}, z: \enct{\type{T}}, a: (\bullet \parr \bullet) \tensor (\enct{\type{T}} \parr \bullet), b: (\bullet \tensor \bullet) \parr (\ol{\enct{\type{T}}} \tensor \bullet)
                }
            }{
                \nu{ab}(\encc{a}{M} \| b(c,d) \sdot (\nu{ef}c[e,f] \| \nu{gh}d[z,g])) \vdashAst \ol{\enct{\type{\Gamma}}}, z: \enct{\type{T}}
            }
        \end{mathpar}

        \begin{mathpar}
            \enc{z}{
                \inferrule[T-Sub]{
                    \type{\Gamma}, \term{x}: \type{T} \vdashM \term{\mbb{M}}: \type{U}
                    \\
                    \type{\Delta} \vdashM \term{\mbb{N}}: \type{T}
                }{
                    \type{\Gamma}, \type{\Delta} \vdashM \term{\mbb{M}\xsub{ \mbb{N}/x }}: \type{U}
                }
            }
            = \inferrule{
                \inferrule*{
                    \encc{z}{\mbb{M}} \vdashAst \ol{\enct{\type{\Gamma}}}, z: \enct{\type{U}}, x: \ol{\enct{\type{T}}}
                    \\
                    \encc{a}{\mbb{N}} \vdashAst \ol{\enct{\type{\Delta}}}, a: \enct{\type{T}}
                }{
                    \encc{z}{\mbb{M}} \| \encc{a}{\mbb{N}} \vdashAst \ol{\enct{\type{\Gamma}}}, \ol{\enct{\type{\Delta}}}, z: \enct{\type{U}}, x: \ol{\enct{\type{T}}}, a: \enct{\type{T}}
                }
            }{
                \nuf{xa}(\encc{z}{\mbb{M}} \| \encc{a}{\mbb{N}}) \vdashAst \ol{\enct{\type{\Gamma}}}, \ol{\enct{\type{\Delta}}}, z: \enct{\type{U}}
            }
        \end{mathpar}
    \end{mdframed}
    \caption{Translation of (runtime) term typing rules with full derivations (part 2 of 5).}\label{f:transTerm2}
\end{figure}

\begin{figure}[t!]
    \begin{mdframed}\small
        \begin{mathpar}
            \enc{z}{
                \inferrule[T-Split]{
                    \type{\Gamma} \vdashM \term{M}: \type{T \times T'}
                    \\
                    \type{\Delta}, \term{x}: \type{T}, \term{y}: \type{T'} \vdashM \term{N}: \type{U}
                }{
                    \type{\Gamma}, \type{\Delta} \vdashM \term{\sff{let}\, (x,y) = M\, \sff{in}\, N}: \type{U}
                }
            }
            \mprset{sep=0.4em}
            = \inferrule{
                \inferrule*{
                    \inferrule*{}{
                        {\begin{array}[t]{@{}l@{}}
                                \encc{a}{M} \vdashAst {}
                                \\
                                \ol{\enct{\type{\Gamma}}},
                                \\
                                a: {\begin{array}[t]{@{}l@{}}
                                        (\enct{\type{T}} \parr \bullet)
                                        \\
                                        {} \tensor (\enct{\type{T'}} \parr \bullet)
                                \end{array}}
                        \end{array}}
                    }
                    \\
                    \inferrule*{
                        \inferrule*{
                            \inferrule*{
                                \inferrule*{
                                    \inferrule*{
                                        \inferrule*{
                                            \inferrule*{ }{
                                                c[e,g] \vdashAst {\begin{array}[t]{@{}l@{}}
                                                        c: (\ol{\enct{\type{T}}} \tensor \bullet),
                                                        \\
                                                        e: \enct{\type{T}}, g: \bullet
                                                \end{array}}
                                            }
                                        }{
                                            c[e,g] \vdashAst {\begin{array}[t]{@{}l@{}}
                                                    c: (\ol{\enct{\type{T}}} \tensor \bullet),
                                                    \\
                                                    e: \enct{\type{T}},
                                                    \\
                                                    g: \bullet, h: \bullet
                                            \end{array}}
                                        }
                                    }{
                                        \nu{gh}c[e,g] \vdashAst {\begin{array}[t]{@{}l@{}}
                                                c: (\ol{\enct{\type{T}}} \tensor \bullet),
                                                \\
                                                e: \enct{\type{T}}
                                        \end{array}}
                                    }
                                    \\
                                    \inferrule*{
                                        \inferrule*{
                                            \inferrule*{ }{
                                                d[f,k] \vdashAst {\begin{array}[t]{@{}l@{}}
                                                        d: (\ol{\enct{\type{T'}}} \tensor \bullet),
                                                        \\
                                                        f: \enct{\type{T'}}, k: \bullet
                                                \end{array}}
                                            }
                                        }{
                                            d[f,k] \vdashAst {\begin{array}[t]{@{}l@{}}
                                                    d: (\ol{\enct{\type{T'}}} \tensor \bullet),
                                                    \\
                                                    f: \enct{\type{T'}},
                                                    \\
                                                    k: \bullet, l: \bullet
                                            \end{array}}
                                        }
                                    }{
                                        \nu{kl}d[f,k] \vdashAst {\begin{array}[t]{@{}l@{}}
                                                d: (\ol{\enct{\type{T'}}} \tensor \bullet),
                                                \\
                                                f: \enct{\type{T'}}
                                        \end{array}}
                                    }
                                    \\
                                    \inferrule*{}{
                                        \encc{z}{N} \vdashAst {\begin{array}[t]{@{}l@{}}
                                                \ol{\enct{\type{\Delta}}},
                                                \\
                                                x: \ol{\enct{\type{T}}},
                                                \\
                                                y: \ol{\enct{\type{T'}}},
                                                \\
                                                z: \enct{\type{U}}
                                        \end{array}}
                                    }
                                }{
                                    \nu{gh}c[e,g] \| \nu{kl}d[f,k] \| \encc{z}{N} \vdashAst {\begin{array}[t]{@{}l@{}}
                                            \ol{\enct{\type{\Delta}}}, z: \enct{\type{U}},
                                            \\
                                            c: (\ol{\enct{\type{T}}} \tensor \bullet), d: (\ol{\enct{\type{T'}}} \tensor \bullet)
                                            \\
                                            e: \enct{\type{T}}, x: \ol{\enct{\type{T}}}, f: \enct{\type{T'}}, y: \ol{\enct{\type{T'}}}
                                    \end{array}}
                                }
                            }{
                                \nuf{fy}\left({\begin{array}{@{}l@{}}
                                            \nu{gh}c[e,g]
                                            \\
                                            {} \| \nu{kl}d[f,k]
                                            \\
                                            {} \| \encc{z}{N}
                                \end{array}}\right) \vdashAst {\begin{array}[t]{@{}l@{}}
                                        \ol{\enct{\type{\Delta}}}, z: \enct{\type{U}},
                                        \\
                                        c: (\ol{\enct{\type{T}}} \tensor \bullet), d: (\ol{\enct{\type{T'}}} \tensor \bullet)
                                        \\
                                        e: \enct{\type{T}}, x: \ol{\enct{\type{T}}}
                                \end{array}}
                            }
                        }{
                            \nuf{ex}\nuf{fy}\left({\begin{array}{@{}l@{}}
                                        \nu{gh}c[e,g]
                                        \\
                                        {} \| \nu{kl}d[f,k]
                                        \\
                                        {} \| \encc{z}{N}
                            \end{array}}\right) \vdashAst {\begin{array}[t]{@{}l@{}}
                                    \ol{\enct{\type{\Delta}}}, z: \enct{\type{U}},
                                    \\
                                    c: (\ol{\enct{\type{T}}} \tensor \bullet), d: (\ol{\enct{\type{T'}}} \tensor \bullet)
                            \end{array}}
                        }
                    }{
                        b(c,d) \sdot \nuf{ex}\nuf{fy}\left({\begin{array}{@{}l@{}}
                                    \nu{gh}c[e,g]
                                    \\
                                    {} \| \nu{kl}d[f,k]
                                    \\
                                    {} \| \encc{z}{N}
                        \end{array}}\right) \vdashAst {\begin{array}[t]{@{}l@{}}
                                \ol{\enct{\type{\Delta}}}, z: \enct{\type{U}},
                                \\
                                b: (\ol{\enct{\type{T}}} \tensor \bullet) \parr (\ol{\enct{\type{T'}}} \tensor \bullet)
                        \end{array}}
                    }
                }{
                    \encc{a}{M} \| b(c,d) \sdot \nuf{ex}\nuf{fy}(\nu{gh}c[e,g] \| \nu{kl}d[f,k] \| \encc{z}{N}) \vdashAst \ol{\enct{\type{\Gamma}}}, \ol{\enct{\type{\Delta}}}, z: \enct{\type{U}}, {\begin{array}[t]{@{}l@{}}
                            a: (\enct{\type{T}} \parr \bullet) \tensor (\enct{\type{T'}} \parr \bullet),
                            \\
                            b: (\ol{\enct{\type{T}}} \tensor \bullet) \parr (\ol{\enct{\type{T'}}} \tensor \bullet)
                    \end{array}}
                }
            }{
                \nu{ab}(\encc{a}{M} \| b(c,d) \sdot \nuf{ex}\nuf{fy}(\nu{gh}c[e,g] \| \nu{kl}d[f,k] \| \encc{z}{N})) \vdashAst \ol{\enct{\type{\Gamma}}}, \ol{\enct{\type{\Delta}}}, z: \enct{\type{U}}
            }
        \end{mathpar}

        \begin{mathpar}
            \enc{z}{
                \inferrule[T-EndL]{
                    \type{\Gamma} \vdashM \term{M}: \type{T}
                }{
                    \type{\Gamma}, \term{x}: \type{\sff{end}} \vdashM \term{M}: \type{T}
                }
            }
            = \inferrule{
                \encc{z}{M} \vdashAst \ol{\enct{\type{\Gamma}}}, z: \enct{\type{T}}
            }{
                \encc{z}{M} \vdashAst \ol{\enct{\type{\Gamma}}}, x: \bullet, z: \enct{\type{T}}
            }
            \and
            \enc{z}{
                \inferrule[T-EndR]{ }{
                    \type{\emptyset} \vdashM \term{x}: \type{\sff{end}}
                }
            }
            = \inferrule{
                \inferrule*{ }{
                    \0 \vdashAst \emptyset
                }
            }{
                \0 \vdashAst z: \bullet
            }
        \end{mathpar}
    \end{mdframed}
    \caption{Translation of (runtime) term typing rules with full derivations (part 3 of 5).}\label{f:transTerm3}
\end{figure}

\begin{figure}[t!]
    \begin{mdframed}\small
        \begin{mathpar}
            \enc{z}{
                \inferrule[T-Send]{
                    \type{\Gamma} \vdashM \term{M}: \type{T \times {!}T \sdot S}
                }{
                    \type{\Gamma} \vdashM \term{\sff{send}~M}: \type{S}
                }
            }
            \mprset{sep=0.7em}
            = \inferrule{
                \inferrule*{
                    \inferrule*{}{
                        {\begin{array}[t]{@{}l@{}}
                                \encc{a}{M} \vdashAst \ol{\enct{\type{\Gamma}}},
                                \\
                                a: (\enct{\type{T}} \parr \bullet)
                                \\
                                {} \tensor (((\ol{\enct{\type{T}}} \tensor \bullet)
                                \\
                                {} \parr \enct{\type{S}}) \parr \bullet)
                        \end{array}}
                    }
                    \\
                    \inferrule*{
                        \inferrule*{
                            \inferrule*{
                                \inferrule*{
                                    \inferrule*{
                                        \inferrule*{ }{
                                            d[e,g] \vdashAst {\begin{array}[t]{@{}l@{}}
                                                    d: ((\enct{\type{T}} \parr \bullet) \tensor \ol{\enct{\type{S}}}) \tensor \bullet,
                                                    \\
                                                    e: (\ol{\enct{\type{T}}} \tensor \bullet) \parr \enct{\type{S}}, g: \bullet
                                            \end{array}}
                                        }
                                    }{
                                        d[e,g] \vdashAst {\begin{array}[t]{@{}l@{}}
                                                d: ((\enct{\type{T}} \parr \bullet) \tensor \ol{\enct{\type{S}}}) \tensor \bullet,
                                                \\
                                                e: (\ol{\enct{\type{T}}} \tensor \bullet) \parr \enct{\type{S}}, g: \bullet, h: \bullet
                                        \end{array}}
                                    }
                                }{
                                    \nu{gh}d[e,g] \vdashAst {\begin{array}[t]{@{}l@{}}
                                            d: ({\begin{array}[t]{@{}l@{}}
                                                    (\enct{\type{T}} \parr \bullet)
                                                    \\
                                                    {} \tensor \ol{\enct{\type{S}}}) \tensor \bullet,
                                            \end{array}}
                                            \\
                                            e: (\ol{\enct{\type{T}}} \tensor \bullet) \parr \enct{\type{S}}
                                    \end{array}}
                                }
                                \\
                                \inferrule*{
                                    \inferrule*{
                                        \inferrule*{ }{
                                            {\begin{array}[t]{@{}l@{}}
                                                    f[c,k] \vdash
                                                    \\
                                                    f: (\enct{\type{T}} \parr \bullet) \tensor \ol{\enct{\type{S}}},
                                                    \\
                                                    c: \ol{\enct{\type{T}}} \tensor \bullet,
                                                    k: \enct{\type{S}}
                                            \end{array}}
                                        }
                                        \\
                                        \inferrule*{ }{
                                            l \fwd z \vdashAst {\begin{array}[t]{@{}l@{}}
                                                    z: \enct{\type{S}},
                                                    \\
                                                    l: \ol{\enct{\type{S}}}
                                            \end{array}}
                                        }
                                    }{
                                        {\begin{array}[t]{@{}c@{}}
                                                f[c,k] \| l \fwd z \vdash
                                                z: \enct{\type{S}}, c: \ol{\enct{\type{T}}} \tensor \bullet,
                                                \\
                                                f: (\enct{\type{T}} \parr \bullet) \tensor \ol{\enct{\type{S}}},
                                                k: \enct{\type{S}}, l: \ol{\enct{\type{S}}}
                                        \end{array}}
                                    }
                                }{
                                    {\begin{array}[t]{@{}c@{}}
                                            \nu{kl}(f[c,k] \| l \fwd z) \vdashAst {}
                                            \\
                                            z: \enct{\type{S}}, c: \ol{\enct{\type{T}}} \tensor \bullet, f: (\enct{\type{T}} \parr \bullet) \tensor \ol{\enct{\type{S}}}
                                    \end{array}}
                                }
                            }{
                                \nu{gh}d[e,g] \| \nu{kl}(f[c,k] \| l \fwd z) \vdashAst z: \enct{\type{S}}, {\begin{array}[t]{@{}l@{}}
                                        c: \ol{\enct{\type{T}}} \tensor \bullet, d: ((\enct{\type{T}} \parr \bullet) \tensor \ol{\enct{\type{S}}}) \tensor \bullet,
                                        \\
                                        e: (\ol{\enct{\type{T}}} \tensor \bullet) \parr \enct{\type{S}}, f: (\enct{\type{T}} \parr \bullet) \tensor \ol{\enct{\type{S}}}
                                \end{array}}
                            }
                        }{
                            \nu{ef}(\nu{gh}d[e,g] \| \nu{kl}(f[c,k] \| l \fwd z)) \vdashAst z: \enct{\type{S}}, c: \ol{\enct{\type{T}}} \tensor \bullet, d: ({\begin{array}[t]{@{}l@{}}
                                    (\enct{\type{T}} \parr \bullet)
                                    \\
                                    {} \tensor \ol{\enct{\type{S}}}) \tensor \bullet
                            \end{array}}
                        }
                    }{
                        {\begin{array}[t]{@{}c@{}}
                                b(c,d) \sdot \nu{ef}(\nu{gh}d[e,g] \| \nu{kl}(f[c,k] \| l \fwd z)) \vdashAst {}
                                \\
                                z: \enct{\type{S}}, b: (\ol{\enct{\type{T}}} \tensor \bullet) \parr (((\enct{\type{T}} \parr \bullet) \tensor \ol{\enct{\type{S}}}) \tensor \bullet)
                        \end{array}}
                    }
                }{
                    {\begin{array}[t]{@{}c@{}}
                            \encc{a}{M} \| b(c,d) \sdot \nu{ef}(\nu{gh}d[e,g] \| \nu{kl}(f[c,k] \| l \fwd z)) \vdashAst {}
                            \\
                            \ol{\enct{\type{\Gamma}}}, z: \enct{\type{S}}, a: (\enct{\type{T}} \parr \bullet) \tensor (((\ol{\enct{\type{T}}} \tensor \bullet) \parr \enct{\type{S}}) \parr \bullet),
                            b: (\ol{\enct{\type{T}}} \tensor \bullet) \parr (((\enct{\type{T}} \parr \bullet) \tensor \ol{\enct{\type{S}}}) \tensor \bullet)
                    \end{array}}
                }
            }{
                \nu{ab}(\encc{a}{M} \| b(c,d) \sdot \nu{ef}(\nu{gh}d[e,g] \| \nu{kl}(f[c,k] \| l \fwd z))) \vdashAst \ol{\enct{\type{\Gamma}}}, z: \enct{\type{S}}
            }
        \end{mathpar}

        \begin{mathpar}
            \enc{z}{
                \inferrule[T-Recv]{
                    \type{\Gamma} \vdashM \term{M}: \type{{?}T \sdot S}
                }{
                    \type{\Gamma} \vdashM \term{\sff{recv}~M}: \type{T \times S}
                }
            }
            \mprset{sep=1.4em}
            = \inferrule{
                \inferrule*{
                    \inferrule*{}{
                        {\begin{array}[t]{@{}l@{}}
                                \encc{a}{M} \vdashAst {}
                                \\
                                \ol{\enct{\type{\Gamma}}},
                                \\
                                a: (\enct{\type{T}} \parr \bullet) \tensor \enct{\type{S}}
                        \end{array}}
                    }
                    \\
                    \inferrule*{
                        \inferrule*{
                            \inferrule*{
                                \inferrule*{ }{
                                    z[c,e] \vdashAst {\begin{array}[t]{@{}l@{}}
                                            z: (\enct{\type{T}} \parr \bullet) \tensor (\enct{\type{S}} \parr \bullet),
                                            \\
                                            c: \ol{\enct{\type{T}}} \tensor \bullet, e: \ol{\enct{\type{S}}} \tensor \bullet
                                    \end{array}}
                                }
                                \\
                                \inferrule*{
                                    \inferrule*{
                                        \inferrule*{ }{
                                            d \fwd g \vdashAst d: \ol{\enct{\type{S}}}, g: \enct{\type{S}}
                                        }
                                    }{
                                        d \fwd g \vdashAst d: \ol{\enct{\type{S}}}, g: \enct{\type{S}}, h: \bullet
                                    }
                                }{
                                    f(g,h) \sdot d \fwd g \vdashAst {\begin{array}[t]{@{}l@{}}
                                            d: \ol{\enct{\type{S}}},
                                            \\
                                            f: \enct{\type{S}} \parr \bullet
                                    \end{array}}
                                }
                            }{
                                z[c,e] \| f(g,h) \sdot d \fwd g \vdashAst {\begin{array}[t]{@{}l@{}}
                                        z: (\enct{\type{T}} \parr \bullet) \tensor (\enct{\type{S}} \parr \bullet),
                                        \\
                                        c: \ol{\enct{\type{T}}} \tensor \bullet, d: \ol{\enct{\type{S}}}, e: \ol{\enct{\type{S}}} \tensor \bullet, f: \enct{\type{S}} \parr \bullet
                                \end{array}}
                            }
                        }{
                            \nu{ef}(z[c,e] \| f(g,h) \sdot d \fwd g) \vdashAst {\begin{array}[t]{@{}l@{}}
                                    z: (\enct{\type{T}} \parr \bullet) \tensor (\enct{\type{S}} \parr \bullet),
                                    \\
                                    c: \ol{\enct{\type{T}}} \tensor \bullet, d: \ol{\enct{\type{S}}}
                            \end{array}}
                        }
                    }{
                        b(c,d) \sdot \nu{ef}(z[c,e] \| f(g,h) \sdot d \fwd g) \vdashAst {\begin{array}[t]{@{}l@{}}
                                z: (\enct{\type{T}} \parr \bullet) \tensor (\enct{\type{S}} \parr \bullet),
                                \\
                                b: (\ol{\enct{\type{T}}} \tensor \bullet) \parr \ol{\enct{\type{S}}}
                        \end{array}}
                    }
                }{
                    \encc{a}{M} \| b(c,d) \sdot \nu{ef}(z[c,e] \| f(g,h) \sdot d \fwd g) \vdashAst \ol{\enct{\type{\Gamma}}}, z: (\enct{\type{T}} \parr \bullet) \tensor (\enct{\type{S}} \parr \bullet), {\begin{array}[t]{@{}l@{}}
                            a: (\enct{\type{T}} \parr \bullet) \tensor \enct{\type{S}},
                            \\
                            b: (\ol{\enct{\type{T}}} \tensor \bullet) \parr \ol{\enct{\type{S}}}
                    \end{array}}
                }
            }{
                \nu{ab}(\encc{a}{M} \| b(c,d) \sdot \nu{ef}(z[c,e] \| f(g,h) \sdot d \fwd g)) \vdashAst \ol{\enct{\type{\Gamma}}}, z: (\enct{\type{T}} \parr \bullet) \tensor (\enct{\type{S}} \parr \bullet)
            }
        \end{mathpar}
    \end{mdframed}
    \caption{Translation of (runtime) term typing rules with full derivations (part 4 of 5).}\label{f:transTerm4}
\end{figure}

\begin{figure}[t!]
    \begin{mdframed}
        \begin{mathpar}\small
            \enc{z}{
                \inferrule[T-Select]{
                    \type{\Gamma} \vdashM \term{M}: \type{\oplus\{i:T_i\}_{i \in I}}
                    \\
                    j \in I
                }{
                    \type{\Gamma} \vdashM \term{\sff{select}\, j\, M}: \type{T_j}
                }
            }
            = \inferrule{
                \inferrule*{
                    \encc{a}{M} \vdashAst \ol{\enct{\type{\Gamma}}}, a: {\&}\{i:\enct{\type{T_i}}\}_{i \in I}
                    \\
                    \inferrule*{
                        \inferrule*{
                            \inferrule*{ }{
                                b[c] \puts j \vdashAst b: \oplus\{i:\ol{\enct{\type{T_i}}}\}_{i \in I}, c: \enct{\type{T_j}}
                            }
                            \\
                            \inferrule*{ }{
                                d \fwd z \vdashAst z: \enct{\type{T_j}}, d: \ol{\enct{\type{T_j}}}
                            }
                        }{
                            b[c] \puts j \| d \fwd z \vdashAst z: \enct{\type{T_j}}, b: \oplus\{i:\ol{\enct{\type{T_i}}}\}_{i \in I}, c: \enct{\type{T_j}}, d: \ol{\enct{\type{T_j}}}
                        }
                    }{
                        \nu{cd}(b[c] \puts j \| d \fwd z) \vdashAst z: \enct{\type{T_j}}, b: \oplus\{i:\ol{\enct{\type{T_i}}}\}_{i \in I}
                    }
                }{
                    \encc{a}{M} \| \nu{cd}(b[c] \puts j \| d \fwd z) \vdashAst \ol{\enct{\type{\Gamma}}}, z: \enct{\type{T_j}}, a: {\&}\{i:\enct{\type{T_i}}\}_{i \in I}, b: \oplus\{i:\ol{\enct{\type{T_i}}}\}_{i \in I}
                }
            }{
                \nu{ab}(\encc{a}{M} \| \nu{cd}(b[c] \puts j \| d \fwd z)) \vdashAst \ol{\enct{\type{\Gamma}}}, z: \enct{\type{T_j}}
            }
        \end{mathpar}

        \begin{mathpar}
            \enc{z}{
                \inferrule[T-Case]{
                    \type{\Gamma} \vdashM \term{M}: \type{{\&}\{i:T_i\}_{i \in I}}
                    \\
                    \forall i \in I.~ \type{\Delta} \vdashM \term{N_i}: \type{T_i \lolli U}
                }{
                    \type{\Gamma}, \type{\Delta} \vdashM \term{\sff{case}\, M\, \sff{of}\, \{i:N_i\}_{i \in I}}: \type{U}
                }
            }
            = \inferrule{
                \inferrule*{
                    \encc{a}{M} \vdashAst \ol{\enct{\type{\Gamma}}}, a: \oplus\{i:\enct{\type{T_i}}\}_{i \in I}
                    \\
                    \inferrule*{
                        \forall i \in I.~ \encc{z}{N_i~c} \vdashAst \ol{\enct{\type{\Delta}}}, z:\enct{\type{U}}, c: \ol{\enct{\type{T_i}}}
                    }{
                        b(c) \gets \{i:\encc{z}{N_i~c}\}_{i \in I} \vdashAst \ol{\enct{\type{\Delta}}}, z:\enct{\type{U}}, b: {\&}\{i:\ol{\enct{\type{T_i}}}\}_{i \in I}
                    }
                }{
                    \encc{a}{M} \| b(c) \gets \{i:\encc{z}{N_i~c}\}_{i \in I} \vdashAst \ol{\enct{\type{\Gamma}}}, \ol{\enct{\type{\Delta}}}, z:\enct{\type{U}}, a: \oplus\{i:\enct{\type{T_i}}\}_{i \in I}, b: {\&}\{i:\ol{\enct{\type{T_i}}}\}_{i \in I}
                }
            }{
                \nu{ab}(\encc{a}{M} \| b(c) \gets \{i:\encc{z}{N_i~c}\}_{i \in I}) \vdashAst \ol{\enct{\type{\Gamma}}}, \ol{\enct{\type{\Delta}}}, z:\enct{\type{U}}
            }
        \end{mathpar}

        \begin{mathpar}
            \enc{z}{
                \inferrule[T-Send']{
                    \type{\Gamma} \vdashM \term{M}: \type{T}
                    \\
                    \type{\Delta} \vdashM \term{\mbb{N}}: \type{{!}T \sdot S}
                }{
                    \type{\Gamma}, \type{\Delta} \vdashM \term{\sff{send}'(M,\mbb{N})}: \type{S}
                }
            }
            \mprset{sep=1.4em}
            = \inferrule{
                \inferrule*{
                    \inferrule*{
                        \inferrule*{
                            \encc{c}{M} \vdashAst \ol{\enct{\type{\Gamma}}}, c: \enct{\type{T}}
                        }{
                            \encc{c}{M} \vdashAst \ol{\enct{\type{\Gamma}}}, c: \enct{\type{T}}, d: \bullet
                        }
                    }{
                        a(c,d) \sdot \encc{c}{M} \vdashAst \ol{\enct{\type{\Gamma}}}, a: \enct{\type{T}} \parr \bullet
                    }
                    \\
                    \inferrule*{
                        \inferrule*{
                            \inferrule*{}{
                                {\begin{array}[t]{@{}l@{}}
                                        \encc{e}{\mbb{N}} \vdashAst \ol{\enct{\type{\Delta}}},
                                        \\
                                        e: (\ol{\enct{\type{T}}} \tensor \bullet) \parr \enct{\type{S}}
                                \end{array}}
                            }
                            \\
                            \inferrule*{
                                \inferrule*{
                                    \inferrule*{ }{
                                        {\begin{array}[t]{@{}l@{}}
                                                f[b,g] \vdash
                                                \\
                                                f: (\enct{\type{T}} \parr \bullet) \tensor \ol{\enct{\type{S}}},
                                                \\
                                                b: \ol{\enct{\type{T}}} \tensor \bullet, g: \enct{\type{S}}
                                        \end{array}}
                                    }
                                    \\
                                    \inferrule*{ }{
                                        {\begin{array}[t]{@{}l@{}}
                                                h \fwd z \vdash
                                                \\
                                                z: \enct{\type{S}},
                                                \\
                                                h: \ol{\enct{\type{S}}}
                                        \end{array}}
                                    }
                                }{
                                    {\begin{array}[t]{@{}l@{}}
                                            f[b,g] \| h \fwd z \vdash
                                            \\
                                            z: \enct{\type{S}}, b: \ol{\enct{\type{T}}} \tensor \bullet,
                                            \\
                                            f: (\enct{\type{T}} \parr \bullet) \tensor \ol{\enct{\type{S}}},
                                            \\
                                            g: \enct{\type{S}}, h: \ol{\enct{\type{S}}}
                                    \end{array}}
                                }
                            }{
                                {\begin{array}[t]{@{}l@{}}
                                        \nu{gh}(f[b,g] \| h \fwd z) \vdash
                                        \\
                                        z: \enct{\type{S}}, b: \ol{\enct{\type{T}}} \tensor \bullet,
                                        \\
                                        f: (\enct{\type{T}} \parr \bullet) \tensor \ol{\enct{\type{S}}}
                                \end{array}}
                            }
                        }{
                            {\begin{array}[t]{@{}c@{}}
                                    \encc{e}{\mbb{N}} \| \nu{gh}(f[b,g] \| h \fwd z) \vdashAst \ol{\enct{\type{\Delta}}}, z: \enct{\type{S}}, b: \ol{\enct{\type{T}}} \tensor \bullet,
                                    \\
                                    e: (\ol{\enct{\type{T}}} \tensor \bullet) \parr \enct{\type{S}}, f: (\enct{\type{T}} \parr \bullet) \tensor \ol{\enct{\type{S}}}
                            \end{array}}
                        }
                    }{
                        \nu{ef}(\encc{e}{\mbb{N}} \| f[b,z]) \vdashAst \ol{\enct{\type{\Delta}}}, z: \enct{\type{S}}, b: \ol{\enct{\type{T}}} \tensor \bullet
                    }
                }{
                    a(c,d) \sdot \encc{c}{M} \| \nu{ef}(\encc{e}{\mbb{N}} \| \nu{gh}(f[b,g] \| h \fwd z)) \vdashAst \ol{\enct{\type{\Gamma}}}, \ol{\enct{\type{\Delta}}}, z: \enct{\type{S}}, a: \enct{\type{T}} \parr \bullet, b: \ol{\enct{\type{T}}} \tensor \bullet
                }
            }{
                \nu{ab}(a(c,d) \sdot \encc{c}{M} \| \nu{ef}(\encc{e}{\mbb{N}} \| \nu{gh}(f[b,g] \| h \fwd z))) \vdashAst \ol{\enct{\type{\Gamma}}}, \ol{\enct{\type{\Delta}}}, z: \enct{\type{S}}
            }
        \end{mathpar}
    \end{mdframed}
    \caption{Translation of (runtime) term typing rules with full derivations (part 5 of 5).}\label{f:transTerm5}
\end{figure}

\begin{figure}[t!]
    \begin{mdframed}\small
        \begin{mathpar}
            \enc{z}{
                \inferrule[T-Main]{
                    \type{\Gamma} \vdashM \term{\mbb{M}}: \type{T}
                }{
                    \type{\Gamma} \vdashC{\main} \term{\main\, \mbb{M}}: \type{T}
                }
            }
            = \encc{z}{\mbb{M}} \vdashAst \ol{\enct{\type{\Gamma}}}, z: \enct{\type{T}}
            \and
            \enc{z}{
                \inferrule[T-Child]{
                    \type{\Gamma} \vdashM \term{\mbb{M}}: \type{\1}
                }{
                    \type{\Gamma} \vdashC{\child} \term{\child\, \mbb{M}}: \type{\1}
                }
            }
            = \encc{z}{\mbb{M}} \vdashAst \ol{\enct{\type{\Gamma}}}, z: \enct{\type{T}}
        \end{mathpar}

        \begin{mathpar}
            \enc{z}{
                \inferrule[T-ParL]{
                    \type{\Gamma} \vdashC{\child} \term{C}: \type{\1}
                    \\
                    \type{\Delta} \vdashC{\phi} \term{D}: \type{T}
                }{
                    \type{\Gamma}, \type{\Delta} \vdashC{\child + \phi} \term{C \prl D}: \type{T}
                }
            }
            = \inferrule{
                \inferrule*{
                    \inferrule*{
                        \encc{a}{C} \vdashAst \ol{\enct{\type{\Gamma}}}, a: \bullet
                    }{
                        \encc{a}{C} \vdashAst \ol{\enct{\type{\Gamma}}}, a: \bullet, b: \bullet
                    }
                }{
                    \nu{ab}\encc{a}{C} \vdashAst \ol{\enct{\type{\Gamma}}}
                }
                \\
                \encc{z}{D} \vdashAst \ol{\enct{\type{\Delta}}}, z: \enct{\type{T}}
            }{
                \nu{ab}\encc{a}{C} \| \encc{z}{D} \vdashAst \ol{\enct{\type{\Gamma}}}, \ol{\enct{\type{\Delta}}}, z: \enct{\type{T}}
            }
        \end{mathpar}

        \begin{mathpar}
            \enc{z}{
                \inferrule[T-ParR]{
                    \type{\Gamma} \vdashC{\phi} \term{C}: \type{T}
                    \\
                    \type{\Delta} \vdashC{\child} \term{D}: \type{\1}
                }{
                    \type{\Gamma}, \type{\Delta} \vdashC{\phi + \child} \term{C \prl D}: \type{T}
                }
            }
            = \inferrule{
                \encc{z}{C} \vdashAst \ol{\enct{\type{\Gamma}}}, z: \enct{\type{T}}
                \\
                \inferrule*{
                    \inferrule*{
                        \encc{a}{D} \vdashAst \ol{\enct{\type{\Delta}}}, a: \bullet
                    }{
                        \encc{a}{D} \vdashAst \ol{\enct{\type{\Delta}}}, a: \bullet, b: \bullet
                    }
                }{
                    \nu{ab}\encc{a}{D} \vdashAst \ol{\enct{\type{\Delta}}}
                }
            }{
                \encc{z}{C} \| \nu{ab}\encc{a}{D} \vdashAst \ol{\enct{\type{\Gamma}}}, \ol{\enct{\type{\Delta}}}, z: \enct{\type{T}}
            }
        \end{mathpar}

        \begin{mathpar}
            \enc{z}{
                \inferrule[T-ResBuf]{
                    \type{\Gamma}, \term{y}: \type{\ol{S}} \vdashB \term{\bfr{\vec{m}}}: \type{S'} > \type{S}
                    \\
                    \type{\Delta}, \term{x}: \type{S'} \vdashC{\phi} \term{C}: \type{T}
                }{
                    \type{\Gamma}, \type{\Delta} \vdashC{\phi} \term{\nu{x\bfr{\vec{m}}y}C}: \type{T}
                }
            }
            = \inferrule{
                \bencb{x}{\encc{z}{C}}{\bfr{\vec{m}}} \vdashAst \ol{\enct{\type{\Gamma}}}, \ol{\enct{\type{\Delta}}}, x: \ol{\enct{\type{S}}}, y: \ol{\enct{\type{\ol{S}}}}, z: \enct{T}
            }{
                \nu{xy}\bencb{x}{\encc{z}{C}}{\bfr{\vec{m}}} \vdashAst \ol{\enct{\type{\Gamma}}}, \ol{\enct{\type{\Delta}}}, z: \enct{T}
            }
        \end{mathpar}

        \begin{mathpar}
            \enc{z}{
                \inferrule[T-Res]{
                    \type{\Gamma} \vdashB \term{\bfr{\vec{m}}}: \type{S'} > \type{S}
                    \\
                    \type{\Delta}, \term{x}: \type{S'}, \term{y}: \type{\ol{S}} \vdashC{\phi} \term{C}: \type{T}
                }{
                    \type{\Gamma}, \type{\Delta} \vdashC{\phi} \term{\nu{x\bfr{\vec{m}}y}C}: \type{T}
                }
            }
            = \inferrule{
                \bencb{x}{\encc{z}{C}}{\bfr{\vec{m}}} \vdashAst \ol{\enct{\type{\Gamma}}}, \ol{\enct{\type{\Delta}}}, x: \ol{\enct{\type{S}}}, y: \ol{\enct{\type{\ol{S}}}}, z: \enct{T}
            }{
                \nu{xy}\bencb{x}{\encc{z}{C}}{\bfr{\vec{m}}} \vdashAst \ol{\enct{\type{\Gamma}}}, \ol{\enct{\type{\Delta}}}, z: \enct{T}
            }
        \end{mathpar}

        \begin{mathpar}
            \enc{z}{
                \inferrule[T-ConfSub]{
                    \type{\Gamma}, \term{x}: \type{T} \vdashC{\phi} \term{C}: \type{U}
                    \\
                    \type{\Delta} \vdashM \term{\mbb{M}}: \type{T}
                }{
                    \type{\Gamma}, \type{\Delta} \vdashC{\phi} \term{C\xsub{ \mbb{M}/x }}: \type{U}
                }
            }
            = \inferrule{
                \inferrule*{
                    \encc{z}{C} \vdashAst \ol{\enct{\type{\Gamma}}}, x: \ol{\enct{\type{T}}}, z: \enct{\type{U}}
                    \\
                    \encc{z}{\mbb{M}} \vdashAst \ol{\enct{\type{\Delta}}}, a: \enct{\type{T}}
                }{
                    \encc{z}{C} \| \encc{z}{\mbb{M}} \vdashAst \ol{\enct{\type{\Gamma}}}, \ol{\enct{\type{\Delta}}}, z: \enct{\type{U}}, x: \ol{\enct{\type{T}}}, a: \enct{\type{T}}
                }
            }{
                \nuf{xa}(\encc{z}{C} \| \encc{z}{\mbb{M}}) \vdashAst \ol{\enct{\type{\Gamma}}}, \ol{\enct{\type{\Delta}}}, z: \enct{\type{U}}
            }
        \end{mathpar}
    \end{mdframed}
    \caption{Translation of configuration typing rules with full derivations.}\label{f:transConf}
\end{figure}

\begin{figure}[t!]
    \begin{mdframed}\small
        \begin{mathpar}
            \benc{x}{P \vdashAst \Lambda, x: \ol{\enct{\type{S'}}}}{
                \inferrule[T-Buf]{ }{
                    \type{\emptyset} \vdashB \term{\bfr{\epsilon}}: \type{S'} > \type{S'}
                }
            }
            = P \vdashAst \Lambda, x: \ol{\enct{\type{S'}}}
        \end{mathpar}

        \begin{mathpar}
            \benc{x}{P \vdashAst \Lambda, x: \ol{\enct{\type{S'}}}}{
                \inferrule[T-BufSend]{
                    \type{\Gamma} \vdashM \term{M}: \type{T}
                    \\
                    \type{\Delta} \vdashB \term{\bfr{\vec{m}}}: \type{S'} > \type{S}
                }{
                    \type{\Gamma}, \type{\Delta} \vdashB \term{\bfr{\vec{m},M}}: \type{S'} > \type{{!}T \sdot S}
                }
            }
            \mprset{sep=1.4em}
            = \inferrule{
                \inferrule*{
                    \inferrule*{
                        \inferrule*{
                            \inferrule*{
                                \inferrule*{ }{
                                    {\begin{array}[t]{@{}l@{}}
                                            x \fwd g \vdashAst {}
                                            \\
                                            x: (\enct{\type{T}} \parr \bullet) \tensor \ol{\enct{\type{S}}},
                                            \\
                                            g: (\ol{\enct{\type{T}}} \tensor \bullet) \parr \enct{\type{S}},
                                    \end{array}}
                                }
                                \\
                                \inferrule*{ }{
                                    {\begin{array}[t]{@{}l@{}}
                                            h[a,c] \vdashAst {}
                                            \\
                                            h: (\enct{\type{T}} \parr \bullet) \tensor \ol{\enct{\type{S}}},
                                            \\
                                            a: \ol{\enct{\type{T}}} \tensor \bullet, c: \enct{\type{S}}
                                    \end{array}}
                                }
                            }{
                                x \fwd g \| h[a,c] \vdashAst {\begin{array}[t]{@{}l@{}}
                                        x: (\enct{\type{T}} \parr \bullet) \tensor \ol{\enct{\type{S}}},
                                        \\
                                        g: (\ol{\enct{\type{T}}} \tensor \bullet) \parr \enct{\type{S}},
                                        \\
                                        h: (\enct{\type{T}} \parr \bullet) \tensor \ol{\enct{\type{S}}},
                                        \\
                                        a: \ol{\enct{\type{T}}} \tensor \bullet, c: \enct{\type{S}}
                                \end{array}}
                            }
                        }{
                            \nu{gh}(x \fwd g \| h[a,c]) \vdashAst {\begin{array}[t]{@{}l@{}}
                                    x: (\enct{\type{T}} \parr \bullet) \tensor \ol{\enct{\type{S}}},
                                    \\
                                    a: \ol{\enct{\type{T}}} \tensor \bullet, c: \enct{\type{S}}
                            \end{array}}
                        }
                        \\
                        \inferrule*{
                            \inferrule*{
                                \encc{e}{M} \vdashAst \ol{\enct{\type{\Gamma}}}, e: \enct{\type{T}}
                            }{
                                \encc{e}{M} \vdashAst \ol{\enct{\type{\Gamma}}}, e: \enct{\type{T}}, f: \bullet
                            }
                        }{
                            b(e,f) \sdot \encc{e}{M} \vdashAst \ol{\enct{\type{\Gamma}}}, b: \enct{\type{T}} \parr \bullet
                        }
                        \\
                        \inferrule*{}{
                            {\begin{array}[t]{@{}l@{}}
                                    \bencb{d}{P\{d/x\}}{\bfr{\vec{m}}} \vdashAst {}
                                    \\
                                    \ol{\enct{\type{\Delta}}}, \Lambda, d: \ol{\enct{\type{S}}}
                            \end{array}}
                        }
                    }{
                        \nu{gh}(x \fwd g \| h[a,c]) \| b(e,f) \sdot \encc{e}{M} \| \bencb{d}{P\{d/x\}}{\bfr{\vec{m}}} \vdashAst {\begin{array}[t]{@{}l@{}}
                                \ol{\enct{\type{\Gamma}}}, \ol{\enct{\type{\Delta}}}, \Lambda, x: (\enct{\type{T}} \parr \bullet) \tensor \ol{\enct{\type{S}}}, a: \ol{\enct{\type{T}}} \tensor \bullet,
                                \\
                                b: \enct{\type{T}} \parr \bullet, c: \enct{\type{S}}, d: \ol{\enct{\type{S}}}
                        \end{array}}
                    }
                }{
                    \nu{cd}(\nu{gh}(x \fwd g \| h[a,c]) \| b(e,f) \sdot \encc{e}{M} \| \bencb{d}{P\{d/x\}}{\bfr{\vec{m}}}) \vdashAst {\begin{array}[t]{@{}l@{}}
                            \ol{\enct{\type{\Gamma}}}, \ol{\enct{\type{\Delta}}}, \Lambda, x: (\enct{\type{T}} \parr \bullet) \tensor \ol{\enct{\type{S}}},
                            \\
                            a: \ol{\enct{\type{T}}} \tensor \bullet, b: \enct{\type{T}} \parr \bullet
                    \end{array}}
                }
            }{
                \nu{ab}\nu{cd}(\nu{gh}(x \fwd g \| h[a,c]) \| b(e,f) \sdot \encc{e}{M} \| \bencb{d}{P\{d/x\}}{\bfr{\vec{m}}}) \vdashAst {\begin{array}[t]{@{}l@{}}
                        \ol{\enct{\type{\Gamma}}}, \ol{\enct{\type{\Delta}}}, \Lambda,
                        \\
                        x: (\enct{\type{T}} \parr \bullet) \tensor \ol{\enct{\type{S}}}
                \end{array}}
            }
        \end{mathpar}

        \begin{mathpar}
            \benc{x}{P \vdashAst \Lambda, x: \ol{\enct{\type{S'}}}}{
                \inferrule[T-BufSelect]{
                    \type{\Gamma} \vdashB \term{\bfr{\vec{m}}}: \type{S'} > \type{S_j}
                    \\
                    j \in I
                }{
                    \type{\Gamma} \vdashB \term{\bfr{\vec{m},j}}: \type{S'} > \type{\oplus\{i:S_i\}_{i \in I}}
                }
            }
            = \inferrule{
                \inferrule*{
                    \inferrule*{
                        \inferrule*{
                            \inferrule*{ }{
                                x \fwd c \vdashAst {\begin{array}[t]{@{}l@{}}
                                        x: \oplus\{i:\ol{\enct{\type{S_i}}}\}_{i \in I},
                                        \\
                                        c: {\&}\{i:\enct{\type{S_i}}\}_{i \in I},
                                \end{array}}
                            }
                            \\
                            \inferrule*{ }{
                                d[a] \puts j \vdashAst d: \oplus\{i:\ol{\enct{\type{S_i}}}\}_{i \in I}, a: \enct{\type{S_j}}
                            }
                        }{
                            x \fwd c \| d[a] \puts j \vdashAst {\begin{array}[t]{@{}l@{}}
                                    x: \oplus\{i:\ol{\enct{\type{S_i}}}\}_{i \in I}, c: {\&}\{i:\enct{\type{S_i}}\}_{i \in I},
                                    \\
                                    d: \oplus\{i:\ol{\enct{\type{S_i}}}\}_{i \in I}, a: \enct{\type{S_j}}
                            \end{array}}
                        }
                    }{
                        \nu{cd}(x \fwd c \| d[a] \puts j) \vdashAst x: \oplus\{i:\ol{\enct{\type{S_i}}}\}_{i \in I}, a: \enct{\type{S_j}}
                    }
                    \\
                    \bencb{b}{P\{b/x\}}{\bfr{\vec{m}}} \vdashAst \ol{\enct{\type{\Gamma}}}, \Lambda, b: \ol{\enct{\type{S_j}}}
                }{
                    \nu{cd}(x \fwd c \| d[a] \puts j) \| \bencb{b}{P\{b/x\}}{\bfr{\vec{m}}} \vdashAst \ol{\enct{\type{\Gamma}}}, \Lambda, x: \oplus\{i:\ol{\enct{\type{S_i}}}\}_{i \in I}, a: \enct{\type{S_j}}, b: \ol{\enct{\type{S_j}}}
                }
            }{
                \nu{ab}(\nu{cd}(x \fwd c \| d[a] \puts j) \| \bencb{b}{P\{b/x\}}{\bfr{\vec{m}}}) \vdashAst \ol{\enct{\type{\Gamma}}}, \Lambda, x: \oplus\{i:\ol{\enct{\type{S_i}}}\}_{i \in I}
            }
        \end{mathpar}
    \end{mdframed}
    \caption{Translation of buffer typing rules with full derivations.}\label{f:transBuf}
\end{figure}

\section{Correctness of the Translation of \ourGV into APCP: Full Proofs}
\label{a:transFullProofs}

\subsection{Completeness}
\label{as:completeness}

Our proofs rely on the following auxiliary properties of the translation.

\begin{lemma}\label{l:transProps}\leavevmode
    \begin{enumerate}
        \item\label{i:transFn}
            Given $\type{\Gamma} \vdashM \term{\mbb{M}}: \type{T}$, if $\term{x} \notin \fn(\term{\mbb{M}})$, then $x \notin \fn(\encc{z}{\mbb{M}})$.

        \item\label{i:transConfFn}
            Given $\type{\Gamma} \vdashC{\phi} \term{C}: \type{T}$, if $\term{x} \notin \fn(\term{C})$, then $x \notin \fn(\encc{z}{C})$.

        \item\label{i:transSubst}
            Given $\type{\Gamma}, \term{x}:\type{U} \vdashM \term{\mbb{M}}: \type{T}$, then $\encc{z}{\mbb{M}\{y/x\}} = \encc{z}{\mbb{M}}\{y/x\}$.

        \item\label{i:transBufCtx}
            Given $\type{\Gamma} \vdashB \term{\bfr{\vec{m}}}: \type{S'} > \type{S}$, if $x \notin \fn(\mcl{R})$, then $\bencb{x}{\mcl{R}[P]}{\bfr{\vec{m}}} \equiv \mcl{R}\big[\bencb{x}{P}{\bfr{\vec{m}}}\big]$.

        \item\label{i:transBufSwap}
            Given $\type{\Gamma_1} \vdashB \term{\bfr{\vec{m}}}: \type{S'_1} > \type{S_1}$ and $\type{\Gamma_2} \vdashB \term{\bfr{\vec{n}}}: \type{S'_2} > \type{S_2}$, then $\bencb{x}{\benc{v}{P}{\bfr{\vec{n}}}}{\term{\bfr{\vec{m}}}} \equiv \bencb{v}{\benc{x}{P}{\bfr{\vec{m}}}}{\term{\bfr{\vec{n}}}}$.

        \item\label{i:transBufSC}
            Given $\type{\Gamma} \vdashB \term{\bfr{\vec{m}}}: \type{S'} > \type{S}$, if $P \equiv Q$, then $\bencb{x}{P}{\bfr{\vec{m}}} \equiv \bencb{x}{Q}{\bfr{\vec{m}}}$.

        \item\label{i:transBufRed}
            Given $\type{\Gamma} \vdashB \term{\bfr{\vec{m}}}: \type{S'} > \type{S}$, if $P \redd Q$, then $\bencb{x}{P}{\bfr{\vec{m}}} \redd \bencb{x}{Q}{\bfr{\vec{m}}}$.

        \item\label{i:transBufCombine}
            Given $\type{\Gamma_1} \vdashB \term{\bfr{\vec{m}_1}}: \type{S'_1} > \type{S_1}$ and $\type{\Gamma_2} \vdashB \term{\bfr{\vec{m}_2}}: \type{S'_2} > \type{S_2}$, then $\bencb{x}{\bencb{x}{P}{\bfr{\vec{m}_2}}}{\bfr{\vec{m}_1}} = \bencb{x}{P}{\bfr{\vec{m}_2,\vec{m}_1}}$.
    \end{enumerate}
\end{lemma}

\begin{proof}
    All properties follow by induction on typing.
\end{proof}

\begin{restatable}[Preservation of Structural Congruence for Terms]
{theorem}{thmTransTermSC}
\label{t:transTermSC}
    Given $\type{\Gamma} \vdashM \term{\mbb{M}}: \type{T}$, if $\term{\mbb{M}} \equivM \term{\mbb{N}}$, then $\encc{z}{\mbb{M}} \equiv \encc{z}{\mbb{N}}$.
\end{restatable}

\begin{proof}
    By induction on the derivation of the structural congruence (\ih{1}).
    The inductive cases follow from \ih{1} straightforwardly.
    The only base case rule is \scc{SC-SubExt}: $\term{x} \notin \fn(\term{\mcl{R}}) \implies \term{(\mcl{R}[\mbb{M}])\xsub{ \mbb{N}/x }} \equivM \term{\mcl{R}[\mbb{M}\xsub{ \mbb{N}/x }]}$.
    We apply induction on the structure of $\term{\mcl{R}}$ (\ih{2}), assuming $\term{x} \notin \fn(\term{\mcl{R}})$.
    We only show the base case and two representative inductive cases:
    \begin{itemize}
        \item
            Case $\term{\mcl{R}} = \term{[]}$: $\term{\mbb{M}\xsub{ \mbb{N}/x }} \equivM \term{\mbb{M}\xsub{ \mbb{N}/x }}$.
            The thesis follows immediately, since the terms are equal.

        \item
            Case $\term{\mcl{R}} = \term{\mcl{R}'~\mbb{M}'}$: $\term{(\mcl{R}'[\mbb{M}]~\mbb{M}')\xsub{ \mbb{N}/x }} \equivM \term{(\mcl{R}'[\mbb{M}\xsub{ \mbb{N}/x }])~\mbb{M}'}$.
            \begin{align*}
                & \encc{z}{(\mcl{R}'[\mbb{M}]~\mbb{M}')\xsub{ \mbb{N}/x }}
                \\
                {}={} & \nuf{xa}(\nu{bc}(\encc{b}{\mcl{R}'[\mbb{M}]} \| \nu{de}(c[d,z] \| e(f,\_) \sdot \encc{f}{\mbb{M}'})) \| \encc{a}{\mbb{N}})
                \\
                {}\equiv{} & \nu{bc}(\nuf{xa}(\encc{b}{\mcl{R}'[\mbb{M}]} \| \encc{a}{\mbb{N}}) \| \nu{de}(c[d,z] \| e(f,\_) \sdot \encc{f}{\mbb{M}'}))
                &&\text{(\encpropref{i:transFn})}
                \\
                {}={} & \encc{z}{((\mcl{R}'[\mbb{M}])\xsub{ \mbb{N}/x })~\mbb{M}'} \equiv \encc{z}{(\mcl{R}'[\mbb{M}\xsub{ \mbb{N}/x }])~\mbb{M}'}
                &&\text{(\ih{2})}
            \end{align*}

        \item
            Case $\term{\mcl{R}} = \term{V~\mcl{R}'}$: $\term{(V~(\mcl{R}'[\mbb{M}]))\xsub{ \mbb{N}/x }} \equivM \term{V~(\mcl{R}'[\mbb{M}\xsub{ \mbb{N}/x }])}$.
            By cases on $\term{V}$.
            We only show the representative case where $\term{V} = \term{\sff{spawn}}$:
            \begin{align*}
                & \encc{z}{(\sff{spawn}~(\mcl{R}'[\mbb{M}]))\xsub{ \mbb{N}/x }}
                \\
                {}={} & \nuf{xa}(\nu{bc}(\encc{b}{\mcl{R}'[\mbb{M}]} \| c(d,e) \sdot (\nunil d[\_,\_] \| \nunil e[z,\_])) \| \encc{a}{\mbb{N}})
                \\
                {}\equiv{} & \nu{bc}(\nuf{xa}(\encc{b}{\mcl{R}'[\mbb{M}]} \| \encc{a}{\mbb{N}}) \| c(d,e) \sdot (\nunil d[\_,\_] \| \nunil e[z,\_]))
                &&\text{(\encpropref{i:transFn})}
                \\
                {}={} & \encc{z}{\sff{spawn}~((\mcl{R}'[\mbb{M}])\xsub{ \mbb{N}/x })} \equiv \encc{z}{\sff{spawn}~(\mcl{R}'[\mbb{M}\xsub{ \mbb{N}/x }])}
                &&\text{(\ih{2})}
                \tag*{\qedhere}
            \end{align*}
    \end{itemize}
\end{proof}

\begin{restatable}[Preservation of Structural Congruence for Configurations]
{theorem}{thmTransConfSC}\label{t:transConfSC}
    Given $\type{\Gamma} \vdashC{\phi} \term{C}: \type{T}$, if $\term{C} \equivC \term{D}$, then $\encc{z}{C} \equiv \encc{z}{D}$.
\end{restatable}

\begin{proof}
    By induction on the derivation of the structural congruence (\ih{1}).
    The inductive cases follow from \ih{1} straightforwardly.
    We consider each base case rule:
    \begin{itemize}
        \item
            Rule~\scc{SC-TermSC}: $\term{\mbb{M}} \equivM \term{\mbb{M}'} \implies \term{\phi\, \mbb{M}} \equivC \term{\phi\, \mbb{M}'}$.
            We assume $\term{\mbb{M}} \equivM \term{\mbb{M}'}$.
            \begin{align*}
                & \encc{z}{\phi\,\mbb{M}}
                \\
                {}={} & \encc{z}{\mbb{M}}
                \\
                {}\equiv{} & \encc{z}{\mbb{N}}
                &&\text{(\Cref{t:transTermSC})}
                \\
                {}={} & \encc{z}{\phi\,\mbb{N}}
            \end{align*}

        \item
            Rule~\scc{SC-ResSwap}: $\term{\nu{x\bfr{\epsilon}y}C} \equivC \term{\nu{y\bfr{\epsilon}x}C}$.
            \begin{align*}
                & \encc{z}{\nu{x\bfr{\epsilon}y}C}
                \\
                {}={} & \nu{xy}\encc{z}{C}
                \\
                {}\equiv{} & \nu{yx}\encc{z}{C}
                \\
                {}={} & \encc{z}{\nu{y\bfr{\epsilon}x}C}
            \end{align*}

        \item
            Rule~\scc{SC-ResComm}: $\term{\nu{x\bfr{\vec{m}}y}\nu{v\bfr{\vec{n}}w}C} \equivC \term{\nu{v\bfr{\vec{n}}w}\nu{x\bfr{\vec{m}}y}C}$.
            \begin{align*}
                & \encc{z}{\nu{x\bfr{\vec{m}}y}\nu{v\bfr{\vec{n}}w}C}
                \\
                {}={} & \nu{xy}\bencb{x}{\nu{vw}\bencb{v}{\encc{z}{C}}{\bfr{\vec{n}}}}{\term{\bfr{\vec{m}}}}
                \\
                {}\equiv{} & \nu{xy}\nu{vw}\bencb{x}{\bencb{v}{\encc{z}{C}}{\bfr{\vec{n}}}}{\term{\bfr{\vec{m}}}}
                &&\text{(\encpropref{i:transBufCtx})}
                \\
                {}\equiv{} & \nu{vw}\nu{xy}\bencb{x}{\bencb{v}{\encc{z}{C}}{\bfr{\vec{n}}}}{\term{\bfr{\vec{m}}}}
                \\
                {}\equiv{} & \nu{vw}\nu{xy}\bencb{v}{\bencb{x}{\encc{z}{C}}{\bfr{\vec{m}}}}{\term{\bfr{\vec{n}}}}
                &&\text{(\encpropref{i:transBufSwap})}
                \\
                {}\equiv{} & \nu{vw}\bencb{v}{\nu{xy}\bencb{x}{\encc{z}{C}}{\bfr{\vec{m}}}}{\term{\bfr{\vec{n}}}}
                &&\text{(\encpropref{i:transBufCtx})}
                \\
                {}={} & \encc{z}{\nu{v\bfr{\vec{n}}w}\nu{x\bfr{\vec{m}}y}C}
            \end{align*}

        \item
            Rule~\scc{SC-ResExt}: $\term{x},\term{y} \notin \fn(\term{C}) \implies \term{\nu{x\bfr{\vec{m}}y}(C \prl D)} \equivC \term{C \prl \nu{x\bfr{\vec{m}}y}D}$.
            W.l.o.g.\ assume $\term{C}$ is the main thread.
            We assume $\term{x},\term{y} \notin \fn(\term{C})$.
            Then, by \encpropref{i:transConfFn}, $x,y \notin \fn(\encc{z}{C})$.
            \begin{align*}
                & \encc{z}{\nu{x\bfr{\vec{m}}y}(C \prl D)}
                \\
                {}={} & \nu{xy}\bencb{x}{\encc{z}{C} \| \nunil \encc{\_}{D}}{\bfr{\vec{m}}}
                \\
                {}\equiv{} & \nu{xy}(\encc{z}{C} \| \nunil \bencb{x}{\encc{\_}{D}}{\bfr{\vec{m}}})
                &&\text{(\encpropref{i:transBufCtx})}
                \\
                {}\equiv{} & \encc{z}{C} \| \nunil \nu{xy}\bencb{x}{\encc{\_}{D}}{\bfr{\vec{m}}}
                \\
                {}={} & \encc{z}{C \prl \nu{x\bfr{\vec{m}}y}D}
            \end{align*}

        \item
            Rule~\scc{SC-ResNil}: $\term{x},\term{y} \notin \fn(\term{C}) \implies \term{\nu{x\bfr{\epsilon}y}C} \equivC \term{C}$.
            We assume $\term{x},\term{y} \notin \fn(C)$.
            \begin{align*}
                & \encc{z}{\nu{x\bfr{\epsilon}y}C}
                \\
                {}={} & \nu{xy}\encc{z}{C}
                \\
                {}\equiv{} & \nu{xy}(\encc{z}{C} \| \0)
                \\
                {}\equiv{} & \encc{z}{C} \| \nu{xy}\0
                &&\text{(\encpropref{i:transConfFn})}
                \\
                {}\equiv{} & \encc{z}{C}
            \end{align*}

        \item
            Rule~\scc{SC-Send'}: $\term{\nu{x\bfr{\vec{m}}y}(\hat{\mcl{F}}[\sff{send}'(M,x)] \prl C)} \equivC \term{\nu{x\bfr{M,\vec{m}}y}(\hat{\mcl{F}}[x] \prl C)}$.
            By induction on the structure of $\term{\hat{\mcl{F}}}$ (\ih{2}).
            Since the hole in $\term{\hat{\mcl{F}}}$ does not occur under an explicit substitution, $\term{\hat{\mcl{F}}}$ does not capture any free names of $\term{\sff{send}'(M,x)}$.
            Therefore, by \encpropref{i:transConfFn}, the encoding of $\term{\hat{\mcl{F}}}$ does not capture any free names of the encoding of $\term{\sff{send}'(M,x)}$.
            The inductive cases follow from \ih{2} straightforwardly, since the encoding of any case for $\term{\hat{\mcl{F}}}$ has the encoding of $\term{\sff{send}'(M,x)}$ only under restriction and parallel composition.
            We detail the base case ($\term{\hat{\mcl{F}}} = \term{\phi\,[]}$), w.l.o.g.\ assuming $\term{\phi} = \term{\main}$:
            \begin{align*}
                & \encc{z}{\nu{x\bfr{\vec{m}}y}(\main\,(\sff{send}'(M,x)) \prl C)}
                \\
                {}={} & \nu{xy}\bencb{x}{\encc{z}{\main\,(\sff{send}'(M,x))} \| \nunil \encc{\_}{C}}{\bfr{\vec{m}}}
                \\
                {}\equiv{} & \nu{xy}(\bencb{x}{\encc{z}{\main\,(\sff{send}'(M,x))}}{\bfr{\vec{m}}} \| \nunil \encc{\_}{C})
                &&\text{(\encpropref{i:transBufCtx})}
                \\
                {}={} & \nu{xy}(\bencb{x}{\nu{ab}(a(c,\_) \sdot \encc{c}{M} \| \nu{de}(x \fwd d \| \nu{fg}(e[b,f] \| g \fwd z)))}{\bfr{\vec{m}}} \| \nunil \encc{\_}{C})
                \\
                {}\equiv{} & \nu{xy}(\bencb{x}{\nu{ab}\nu{fg}(\nu{de}(x \fwd d \| e[b,f]) \| a(c,\_) \sdot \encc{c}{M} \| g \fwd z)}{\bfr{\vec{m}}} \| \nunil \encc{\_}{C})
                &&\text{(\encpropref{i:transBufSC})}
                \\
                {}={} & \nu{xy}(\bencb{x}{\nu{ab}\nu{fg}(\nu{de}(x \fwd d \| e[b,f]) \| a(c,\_) \sdot \encc{c}{M} \| (x \fwd z)\{g/x\})}{\bfr{\vec{m}}} \| \nunil \encc{\_}{C})
                \\
                {}={} & \nu{xy}(\bencb{x}{\nu{ab}\nu{fg}(\nu{de}(x \fwd d \| e[b,f]) \| a(c,\_) \sdot \encc{c}{M} \| \encc{z}{x}\{g/x\})}{\bfr{\vec{m}}} \| \nunil \encc{\_}{C})
                \\
                {}={} & \nu{xy}(\bencb{x}{\bencb{x}{\encc{z}{x}}{\bfr{M}}}{\term{\bfr{\vec{m}}}} \| \nunil \encc{\_}{C})
                \\
                {}={} & \nu{xy}(\bencb{x}{\encc{z}{x}}{\bfr{M,\vec{m}}} \| \nunil \encc{\_}{C})
                &&\text{(\encpropref{i:transBufCombine})}
                \\
                {}\equiv{} & \nu{xy}\bencb{x}{\encc{z}{x} \| \nunil \encc{\_}{C}}{\bfr{M,\vec{m}}}
                &&\text{(\encpropref{i:transBufCtx})}
                \\
                {}={} & \nu{xy}\bencb{x}{\encc{z}{\main\,x} \| \nunil \encc{\_}{C}}{\bfr{M,\vec{m}}}
                \\
                {}={} & \nu{xy}\bencb{x}{\encc{z}{\main\,x \prl C}}{\bfr{M,\vec{m}}}
                \\
                {}={} & \encc{z}{\nu{x\bfr{M,\vec{m}}y}(\main\,x \prl C)}
            \end{align*}

        \item
            Rule~\scc{SC-Select}: $\term{\nu{x\bfr{\vec{m}}y}(\mcl{F}[\sff{select}\, \ell\, x] \prl C)} \equivC \term{\nu{x\bfr{\ell,\vec{m}}y}(\mcl{F}[x] \prl C)}$.
            By induction on the structure of $\term{\mcl{F}}$.
            Similar to the case above, we only detail the base case ($\term{\mcl{F}} = \term{\phi\,[]}$), w.l.o.g.\ assuming $\term{\phi} = \term{\main}$:
            \begin{align*}
                & \encc{z}{\nu{x\bfr{\vec{m}}y}(\main\,(\sff{select}\,\ell\,x) \prl C)}
                \\
                {}={} & \nu{xy}\bencb{x}{\encc{z}{\main\,(\sff{select}\,\ell\,x)} \| \nunil \encc{\_}{C}}{\bfr{\vec{m}}}
                \\
                {}\equiv{} & \nu{xy}(\bencb{x}{\encc{z}{\main\,(\sff{select}\,\ell\,x)}}{\bfr{\vec{m}}} \| \nunil \encc{\_}{C})
                &&\text{(\encpropref{i:transBufCtx})}
                \\
                {}={} & \nu{xy}(\bencb{x}{\nu{ab}(x \fwd a \| \nu{cd}(b[c] \puts \ell \| d \fwd z))}{\bfr{\vec{m}}} \| \nunil \encc{\_}{C})
                \\
                {}\equiv{} & \nu{xy}(\bencb{x}{\nu{cd}(\nu{ab}(x \fwd a \| b[c] \puts \ell) \| d \fwd z)}{\bfr{\vec{m}}} \| \nunil \encc{\_}{C})
                &&\text{(\encpropref{i:transBufSC})}
                \\
                {}={} & \nu{xy}(\bencb{x}{\nu{cd}(\nu{ab}(x \fwd a \| b[c] \puts \ell) \| (x \fwd z)\{d/x\})}{\bfr{\vec{m}}} \| \nunil \encc{\_}{C})
                \\
                {}={} & \nu{xy}(\bencb{x}{\nu{cd}(\nu{ab}(x \fwd a \| b[c] \puts \ell) \| \encc{z}{x}\{d/x\})}{\bfr{\vec{m}}} \| \nunil \encc{\_}{C})
                \\
                {}={} & \nu{xy}(\bencb{x}{\bencb{x}{\encc{z}{x}}{\bfr{\ell}}}{\term{\bfr{\vec{m}}}} \| \nunil \encc{\_}{C})
                \\
                {}={} & \nu{xy}(\bencb{x}{\encc{z}{x}}{\bfr{\ell,\vec{m}}} \| \nunil \encc{\_}{C})
                &&\text{(\encpropref{i:transBufCombine})}
                \\
                {}\equiv{} & \nu{xy}\bencb{x}{\encc{z}{x} \| \nunil \encc{\_}{C}}{\bfr{\ell,\vec{m}}}
                &&\text{(\encpropref{i:transBufCtx})}
                \\
                {}={} & \nu{xy}\bencb{x}{\encc{z}{\main\,x} \| \nunil \encc{\_}{C}}{\bfr{\ell,\vec{m}}}
                \\
                {}={} & \nu{xy}\bencb{x}{\encc{z}{\main\,x \prl C}}{\bfr{\ell,\vec{m}}}
                \\
                {}={} & \encc{z}{\nu{x\bfr{\ell,\vec{m}}y}(\main\,x \prl C)}
            \end{align*}

        \item
            Rule~\scc{SC-ParNil} ($\term{C \prl \child\, ()} \equivC \term{C}$) is straightforward.

        \item
            Rule~\scc{SC-ParComm} ($\term{C \prl D} \equivC \term{D \prl C}$) is straightforward.

        \item
            Rule~\scc{SC-ParAssoc} ($\term{C \prl (D \prl E)} \equivC \term{(C \prl D) \prl E}$) is straightforward.

        \item
            Rule~\scc{SC-ConfSubst} ($\term{\phi\,(\mbb{M}\xsub{ \mbb{N}/x })} \equivC \term{(\phi\,\mbb{M})\xsub{ \mbb{N}/x }}$) is straightforward.

        \item
            Rule~\scc{SC-ConfSubstExt}: $\term{x} \notin \fn(\term{\mcl{G}}) \implies \term{(\mcl{G}[C])\xsub{ \mbb{M}/x }} \equivC \term{\mcl{G}[C\xsub{ \mbb{M}/x }]}$.
            We assume $\term{x} \notin \fn(\term{\mcl{G}})$ and apply induction on the structure of $\term{\mcl{G}}$ (\ih{2}).
            \begin{itemize}
                \item
                    Case $\term{\mcl{G}} = \term{[]}$: $\term{C\xsub{ \mbb{M}/x }} \equivC \term{C\xsub{ \mbb{M}/x }}$.
                    The thesis follows immediately, since the terms are equal.

                \item
                    Case $\term{\mcl{G}} = \term{\mcl{G}' \prl C'}$: $\term{(\mcl{G}'[C] \prl C')\xsub{ \mbb{M}/x }} \equivC \term{\mcl{G}'\big[C\xsub{ \mbb{M}/x }\big] \prl C'}$.
                    W.l.o.g.\ we assume $\term{\mcl{G}'[C]}$ is the main thread.
                    \begin{align*}
                        & \encc{z}{(\mcl{G}'[C] \prl C')\xsub{ \mbb{M}/x }}
                        \\
                        {}={} & \nuf{xa}(\encc{z}{\mcl{G}'[C]} \| \nunil \encc{\_}{C'} \| \encc{a}{\mbb{M}})
                        \\
                        {}\equiv{} & \nuf{xa}(\encc{z}{\mcl{G}'[C]} \| \encc{a}{\mbb{M}}) \| \nunil \encc{\_}{C'}
                        &&\text{(\encpropref{i:transConfFn})}
                        \\
                        {}={} & \encc{z}{(\mcl{G}'[C])\xsub{ \mbb{M}/x }} \| \nunil \encc{\_}{C'}
                        \\
                        {}\equiv{} & \encc{z}{\mcl{G}'\big[C\xsub{ \mbb{M}/x }\big]} \| \nunil \encc{\_}{C'}
                        &&\text{(\ih{2})}
                        \\
                        {}\equiv{} & \encc{z}{\mcl{G}'\big[C\xsub{ \mbb{M}/x }\big] \prl C'}
                    \end{align*}

                \item
                    Case $\term{\mcl{G}} = \term{\nu{v\bfr{\vec{m}}w}\mcl{G}'}$: $\term{(\nu{v\bfr{\vec{m}}w}(\mcl{G}'[C]))\xsub{ \mbb{M}/x }} \equivC \term{\nu{v\bfr{\vec{m}}w}(\mcl{G}'\big[C\xsub{ \mbb{M}/x }\big])}$.
                    \begin{align*}
                        & \encc{z}{(\nu{v\bfr{\vec{m}}w}(\mcl{G}'[C]))\xsub{ \mbb{M}/x }}
                        \\
                        {}={} & \nuf{xa}(\nu{vw}\bencb{z}{\encc{z}{\mcl{G'}[C]}}{\bfr{\vec{m}}} \| \encc{a}{\mbb{M}})
                        \\
                        {}\equiv{} & \nu{vw}\nuf{xa}(\bencb{z}{\encc{z}{\mcl{G'}[C]}}{\bfr{\vec{m}}} \| \encc{a}{\mbb{M}})
                        &&\text{(\encpropref{i:transConfFn})}
                        \\
                        {}\equiv{} & \nu{vw}\bencb{z}{\nuf{xa}(\encc{z}{\mcl{G'}[C]} \| \encc{a}{\mbb{M}})}{\bfr{\vec{m}}}
                        &&\text{(\encpropref{i:transBufCtx})}
                        \\
                        {}={} & \nu{vw}\bencb{z}{\encc{z}{(\mcl{G'}[C])\xsub{ \mbb{M}/x }}}{\bfr{\vec{m}}}
                        \\
                        {}\equiv{} & \nu{vw}\bencb{z}{\encc{z}{\mcl{G'}\big[C\xsub{ \mbb{M}/x }\big]}}{\bfr{\vec{m}}}
                        &&\text{(\ih{2})}
                        \\
                        {}={} & \encc{z}{\nu{v\bfr{\vec{m}}w}(\mcl{G}'\big[C\xsub{ \mbb{M}/x }\big])}
                    \end{align*}

                \item
                    Case $\term{\mcl{G}} = \term{\mcl{G}'\xsub{ \mbb{M}'/y }}$: $\term{((\mcl{G}'[C])\xsub{ \mbb{M}'/y })\xsub{ \mbb{M}/x }} \equivC \term{(\mcl{G}'\big[C\xsub{ \mbb{M}/x }\big])\xsub{ \mbb{M}'/y }}$.
                    \begin{align*}
                        & \encc{z}{((\mcl{G}'[C])\xsub{ \mbb{M}'/y })\xsub{ \mbb{M}/x }}
                        \\
                        {}={} & \nuf{xa}(\nuf{yb}(\encc{z}{\mcl{G}'[C]} \| \encc{b}{\mbb{M}'}) \| \encc{a}{\mbb{M}})
                        \\
                        {}\equiv{} & \nuf{yb}(\nuf{xa}(\encc{z}{\mcl{G}'[C]} \| \encc{a}{\mbb{M}}) \| \encc{b}{\mbb{M}'})
                        &&\text{(\encpropref{i:transConfFn})}
                        \\
                        {}={} & \nuf{yb}(\encc{z}{(\mcl{G}'[C])\xsub{ \mbb{M}/x }} \| \encc{b}{\mbb{M}'})
                        \\
                        {}\equiv{} & \nuf{yb}(\encc{z}{\mcl{G}'\big[C\xsub{ \mbb{M}/x }\big]} \| \encc{b}{\mbb{M}'})
                        &&\text{(\ih{2})}
                        \\
                        {}={} & \encc{z}{(\mcl{G}'\big[C\xsub{ \mbb{M}/x }\big])\xsub{ \mbb{M}'/y }}
                        \tag*{\qedhere}
                    \end{align*}
            \end{itemize}
    \end{itemize}
\end{proof}

\begin{restatable}[Completeness of Reduction for Terms]
{theorem}{thmTransTermRed}\label{t:transTermRed}
    Given $\type{\Gamma} \vdashM \term{\mbb{M}}: \type{T}$, if $\term{\mbb{M}} \reddM \term{\mbb{N}}$, then $\encc{z}{\mbb{M}} \redd^\ast \encc{z}{\mbb{N}}$.
\end{restatable}

\begin{proof}
    By induction on the derivation of the reduction (\ih{1}).
    We consider each rule:
    \begin{itemize}
        \item
            Rule~\scc{E-Lam}: $\term{(\lambda x \sdot M)\, \mbb{N}} \reddM \term{M\xsub{ \mbb{N}/x }}$.
            \begin{align*}
                & \encc{z}{(\lambda x \sdot M)\, \mbb{N}}
                \\
                {}={} & \nu{ab}(a(f,g) \sdot \nuf{hx}(\nunil f[h,\_] \| \encc{g}{M}) \| \nu{cd}(b[c,z] \| d(e,\_) \sdot \encc{e}{\mbb{N}}))
                \\
                {}\redd{} & \nu{cd}(\nuf{hx}(\nunil c[h,\_] \| \encc{z}{M}) \| d(e,\_) \sdot \encc{e}{\mbb{N}})
                \\
                {}\redd{} & \nuf{xh}(\encc{z}{M} \| \encc{h}{\mbb{N}})
                \\
                {}={} & \encc{z}{M\xsub{ \mbb{N}/x }}
            \end{align*}

        \item
            Rule~\scc{E-Pair}: $\term{\sff{let}\, (x,y) = (M_1,M_2)\, \sff{in}\, N} \reddM \term{N\xsub{ M_1/x,M_2/y }}$.
            \begin{align*}
                & \encc{z}{\sff{let}\, (x,y) = (M_1,M_2)\, \sff{in}\, N}
                \\
                &= \nu{ab}\begin{array}[t]{@{}l@{}}
                    (\nu{gh}\nu{kl}(a[g,k] \| h(m,\_) \sdot \encc{m}{M_1} \| l(n,\_) \sdot \encc{n}{M_2})
                    \\
                    {}\| b(c,d) \sdot \nuf{ex}\nuf{fy}(\nunil c[e,\_] \| \nunil d[f,\_] \| \encc{z}{N}))
                \end{array}
                \\
                &\redd \nu{hg}\nu{lk}(h(m,\_) \sdot \encc{m}{M_1} \| l(n,\_) \sdot \encc{n}{M_2} \| \nuf{ex}\nuf{fy}(\nunil g[e,\_] \| \nunil k[f,\_] \| \encc{z}{N}))
                \\
                &\redd \nuf{ex}\nu{lk}(\encc{e}{M_1} \| l(n,\_) \sdot \encc{n}{M_2} \| \nuf{fy}(\nunil k[f,\_] \| \encc{z}{N}))
                \\
                &\redd \nuf{ex}\nuf{fy}(\encc{e}{M_1} \| \encc{f}{M_2} \| \encc{z}{N})
                \\
                &= \encc{z}{N\xsub{ M_1/x,M_2/y }}
            \end{align*}

        \item
            Rule~\scc{E-SubstName}: $\term{\mbb{M}\xsub{ y/x }} \reddM \term{\mbb{M}\{y/x\}}$.
            \begin{align*}
                & \encc{z}{\mbb{M}\xsub{ y/x }}
                \\
                {}={} & \nuf{xa}(\encc{z}{\mbb{M}} \| y \fwd a)
                \\
                {}\redd{} & \encc{z}{\mbb{M}}\{y/x\}
                \\
                {}={} & \encc{z}{\mbb{M}\{y/x\}}
                &&\text{(\encpropref{i:transSubst})}
            \end{align*}

        \item
            Rule~\scc{E-NameSubst}: $\term{x\xsub{ \mbb{M}/x }} \reddM \term{\mbb{M}}$.
            \begin{align*}
                & \encc{z}{x\xsub{ \mbb{M}/x }}
                \\
                {}={} & \nuf{xa}(x \fwd z \| \encc{a}{\mbb{M}})
                \\
                {}\redd{} & \encc{z}{\mbb{M}}
            \end{align*}

        \item
            Rule~\scc{E-Send}: $\term{\sff{send}~(M,N)} \reddM \term{\sff{send}'(M,N)}$.
            \begin{align*}
                & \encc{z}{\sff{send}~(M,N)}
                \\
                {}={} & \nu{ab}\begin{array}[t]{@{}l@{}}
                    (\nu{kl}\nu{mn}(a[k,m] \| l(o,\_) \sdot \encc{o}{M} \| n(p,\_) \sdot \encc{p}{N})
                    \\
                    {}\| b(c,d) \sdot \nu{ef}(\nunil d[e,\_] \| \nu{gh}(f[c,g] \| h \fwd z)))
                \end{array}
                \\
                {}\redd{} & \nu{lk}\nu{nm}(l(o,\_) \sdot \encc{o}{M} \| n(p,\_) \sdot \encc{p}{N} \| \nu{ef}(\nunil m[e,\_] \| \nu{gh}(f[k,g] \| h \fwd z)))
                \\
                {}\redd{} & \nu{lk}\nu{ef}(l(o,\_) \sdot \encc{o}{M} \| \encc{e}{N} \| \nu{gh}(f[k,g] \| h \fwd z))
                \\
                {}\equiv{} & \nu{lk}(l(o,\_) \sdot \encc{o}{M} \| \nu{ef}(\encc{e}{N} \| \nu{gh}(f[k,g] \| h \fwd z)))
                \\
                {}={} & \encc{z}{\sff{send}'(M,N)}
            \end{align*}

        \item
            Rule~\scc{E-Lift}: $\term{\mbb{M}} \reddM \term{\mbb{N}} \implies \term{\mcl{R}[\mbb{M}]} \reddM \term{\mcl{R}[\mbb{N}]}$.
            By induction on the structure of $\term{\mcl{R}}$ (\ih{2}), assuming $\term{\mbb{M}} \reddM \term{\mbb{N}}$.
            The inductive cases follow from \ih{2} straightforwardly, since the encoding of any case for $\term{\mcl{R}}$ has the encoding of $\term{\mbb{M}}$ only under restriction and parallel composition.
            The base case ($\term{\mcl{R}} = []$) follows from \ih{1}: $\encc{z}{\mbb{M}} \redd^\ast \encc{z}{\mbb{N}}$.

        \item
            Rule~\scc{E-LiftSC}: $\term{\mbb{M}} \equivM \term{\mbb{M}'} \wedge \term{\mbb{M}'} \reddM \term{\mbb{N}'} \wedge \term{\mbb{N}'} \equivM \term{\mbb{N}} \implies \term{\mbb{M}} \reddM \term{\mbb{N}}$.
            We assume $\term{\mbb{M}} \equivM \term{\mbb{M}'}$, $\term{\mbb{M}'} \reddM \term{\mbb{N}'}$, and $\term{\mbb{N}'} \equivM \term{\mbb{N}}$.
            By \ih{1}, $\encc{z}{\mbb{M}'} \redd^\ast \encc{z}{\mbb{N}'}$.
            By \Cref{t:transTermSC}, $\encc{z}{\mbb{M}} \equiv \encc{z}{\mbb{M}'}$ and $\encc{z}{\mbb{N}} \equiv \encc{z}{\mbb{N}'}$.
            Hence, $\encc{z}{\mbb{M}} \redd^\ast \encc{z}{\mbb{N}}$.
            \qedhere
    \end{itemize}
\end{proof}

\thmTransConfRed*

\begin{proof}
    By induction on the derivation of the reduction (\ih{1}):
    \begin{itemize}
        \item
            Rule~\scc{E-New}: $\term{\mcl{F}[\sff{new}]} \reddC \term{\nu{x\bfr{\epsilon}y}(\mcl{F}[(x,y)])}$.
            By induction on the structure of $\term{\mcl{F}}$ (\ih{2}).
            The inductive cases follow from \ih{2} straightforwardly, since the encoding of any case for $\term{\mcl{F}}$ has the encoding of $\term{\sff{new}}$ only under restriction and parallel composition.
            We detail the base case ($\term{\mcl{F}} = \term{\phi\,[]}$), w.l.o.g.\ assuming $\term{\phi} = \term{\main}$:
            \begin{align*}
                & \encc{z}{\main\,\sff{new}}
                \\
                {}={} & \nu{ab}(\nunil a[\_,\_] \| b(\_,\_) \sdot \nu{xy}\encc{z}{(x,y)})
                \\
                {}\redd{} & \nu{xy}\encc{z}{(x,y)}
                \\
                {}={} & \nu{xy}\encc{z}{\main\,(x,y)}
                \\
                {}={} & \nu{xy}\bencb{x}{\encc{z}{\main\,(x,y)}}{\bfr{\epsilon}}
                \\
                {}={} & \encc{z}{\nu{x\bfr{\epsilon}y}(\main\,(x,y))}
            \end{align*}

        \item
            Rule~\scc{E-Spawn}: $\term{\hat{\mcl{F}}[\sff{spawn}~(M,N)]} \reddC \term{\hat{\mcl{F}}[N] \prl \child\,M}$.
            By induction on the structure of $\term{\hat{\mcl{F}}}$.
            Similar to the previous case, we only detail the base case ($\term{\hat{\mcl{F}}} = \term{\phi\,[]}$), w.l.o.g.\ assuming $\term{\phi} = \term{\main}$:
            \begin{align*}
                & \encc{z}{\main\,(\sff{spawn}~(M,N))}
                \\
                {}={} & \nu{ab}(\nu{ef}\nu{gh}(a[e,g] \| f(\_,\_) \sdot \encc{\_}{M} \| h(k,\_) \sdot \encc{k}{N}) \| b(c,d) \sdot (\nunil c[\_,\_] \| \nunil d[z,\_]))
                \\
                {}\redd{} & \nu{hg}(h(k,\_) \sdot \encc{k}{N} \| \nunil g[z,\_]) \| \nu{fe}(f(\_,\_) \sdot \encc{\_}{M} \| \nunil e[\_,\_])
                \\
                {}\redd{} & \encc{z}{N} \| \nu{fe}(f(\_,\_) \sdot \encc{\_}{M} \| \nunil e[\_,\_])
                \\
                {}\redd{} & \encc{z}{N} \| \nunil \encc{\_}{M}
                \\
                {}={} & \encc{z}{\main\,N} \| \nunil \encc{\_}{\child\,M}
                \\
                {}={} & \encc{z}{\main\,N \prl \child\,M}
            \end{align*}

        \item
            Rule~\scc{E-Recv}: $\term{\nu{x\bfr{\vec{m},M}y}(\hat{\mcl{F}}[\sff{recv}~y] \prl C)} \reddC \term{\nu{x\bfr{\vec{m}}y}(\hat{\mcl{F}}[(M,y)] \prl C)}$.
            By induction on the structure of $\term{\hat{\mcl{F}}}$.
            Similar to the previous cases, we only detail the base case ($\term{\hat{\mcl{F}}} = \term{\phi\,[]}$), w.l.o.g.\ assuming $\term{\phi} = \term{\main}$:
            \begin{align*}
                & \encc{z}{\nu{x\bfr{\vec{m},M}y}(\main\,(\sff{recv}~y) \prl C)}
                \\
                {}={} & \nu{xy}\bencb{x}{\encc{z}{\main\,(\sff{recv}~y)} \| \nunil \encc{\_}{C}}{\bfr{\vec{m},M}}
                \\
                {}\equiv{} & \nu{xy}(\bencb{x}{\nunil \encc{\_}{C}}{\bfr{\vec{m},M}} \| \encc{z}{\main\,(\sff{recv}~y)})
                &&\text{(\encpropref{i:transBufCtx})}
                \\
                {}\equiv{} & \nu{xy}(\bencb{x}{\nunil \encc{\_}{C}}{\bfr{\vec{m},M}} \| \nu{ab}(y \fwd a \| b(c,d) \sdot \nu{ef}(z[c,e] \| f(g,\_) \sdot d \fwd g)))
                \\
                {}={} & \nu{xy}(\nu{hk}\nu{lm}(\nu{op}(x \fwd o \| p[h,l]) \| k(n,\_) \sdot \encc{n}{M} \| \bencb{m}{\nunil \encc{\_}{C}\{m/x\}}{\bfr{\vec{m}}})\\
                &{}\| \nu{ab}(y \fwd a \| b(c,d) \sdot \nu{ef}(z[c,e] \| f(g,\_) \sdot d \fwd g)))
                \\
                {}\redd^2{} & \nu{xy}(\nu{hk}\nu{lm}(x[h,l] \| k(n,\_) \sdot \encc{n}{M} \| \bencb{m}{\nunil \encc{\_}{C}\{m/x\}}{\bfr{\vec{m}}})\\
                &{}\| y(c,d) \sdot \nu{ef}(z[c,e] \| f(g,\_) \sdot d \fwd g))
                \\
                {}\redd{} & \nu{ml}(\bencb{m}{\nunil \encc{\_}{C}\{m/x\}}{\bfr{\vec{m}}} \| \nu{hk}\nu{ef}(z[h,e] \| k(n,\_) \sdot \encc{n}{M} \| f(g,\_) \sdot l \fwd g))
                \\
                {}\equiv{} & \nu{xy}(\bencb{x}{\nunil \encc{\_}{C}}{\bfr{\vec{m}}} \| \nu{hk}\nu{ef}(z[h,e] \| k(n,\_) \sdot \encc{n}{M} \| f(g,\_) \sdot y \fwd g))
                \\
                {}={} & \nu{xy}(\bencb{x}{\nunil \encc{\_}{C}}{\bfr{\vec{m}}} \| \encc{z}{(M,y)})
                \\
                {}\equiv{} & \nu{xy}\bencb{x}{\encc{z}{(M,y)} \| \nunil \encc{\_}{C}}{\bfr{\vec{m}}}
                &&\text{(\encpropref{i:transBufCtx})}
                \\
                {}={} & \nu{xy}\bencb{x}{\encc{z}{\main\,(M,y)} \| \nunil \encc{\_}{C}}{\bfr{\vec{m}}}
                \\
                {}={} & \nu{xy}\bencb{x}{\encc{z}{\main\,(M,y) \prl C}}{\bfr{\vec{m}}}
                \\
                {}={} & \encc{z}{\nu{x\bfr{\vec{m}}y}(\main\,(M,y) \prl C)}
            \end{align*}

        \item
            Rule~\scc{E-Case}: $j \in I \implies \term{\nu{x\bfr{\vec{m},j}y}(\mcl{F}[\sff{case}\, y\, \sff{of}\, \{i:M_i\}_{i \in I}] \prl C)} \reddC \term{\nu{x\bfr{\vec{m}}y}(\mcl{F}[M_j~y] \prl C)}$.
            Assuming $j \in I$, we apply induction on the structure of $\term{\mcl{F}}$.
            Similar to the previous cases, we only detail the base case ($\term{\mcl{F}} = \term{\phi\,[]}$), w.l.o.g.\ assuming $\term{\phi} = \term{\main}$:
            \begin{align*}
                & \encc{z}{\nu{x\bfr{\vec{m},j}y}(\main\,(\sff{case}\,y\,\sff{of}\,\{i:M_i\}_{i \in I} \prl C)}
                \\
                {}={} & \nu{xy}\bencb{x}{\encc{z}{\main\,(\sff{case}\,y\,\sff{of}\,\{i:M_i\}_{i \in I}} \| \nunil \encc{\_}{C}}{\bfr{\vec{m},j}}
                \\
                {}\equiv{} & \nu{xy}(\bencb{x}{\nunil \encc{\_}{C}}{\bfr{\vec{m},j}} \| \encc{z}{\main\,(\sff{case}\,y\,\sff{of}\,\{i:M_i\}_{i \in I}})
                &&\text{(\encpropref{i:transBufCtx})}
                \\
                {}={} & \nu{xy}(\bencb{x}{\nunil \encc{\_}{C}}{\bfr{\vec{m},j}} \| \nu{ab}(y \fwd a \| b(c) \gets \{i: \encc{z}{M_i~c}\}_{i \in I}))
                \\
                {}={} & \nu{xy}(\nu{de}(\nu{fg}(x \fwd f \| g[d] \puts j) \| \bencb{e}{\nunil \encc{\_}{C}\{e/x\}}{\bfr{\vec{m}}})\\
                &{}\| \nu{ab}(y \fwd a \| b(c) \gets \{i: \encc{z}{M_i~c}\}_{i \in I}))
                \\
                {}\redd^2{} & \nu{xy}(\nu{de}(x[d] \puts j \| \bencb{e}{\nunil \encc{\_}{C}\{e/x\}}{\bfr{\vec{m}}})\\
                &{}\| y(c) \gets \{i: \encc{z}{M_i~c}\}_{i \in I})
                \\
                {}\redd{} & \nu{ed}(\bencb{e}{\nunil \encc{\_}{C}\{e/x\}}{\bfr{\vec{m}}} \| \encc{z}{M_j~c}\{d/c\})
                \\
                {}\equiv{} & \nu{xy}(\bencb{x}{\nunil \encc{\_}{C}}{\bfr{\vec{m}}} \| \encc{z}{M_j~c}\{y/c\})
                \\
                {}={} & \nu{xy}(\bencb{x}{\nunil \encc{\_}{C}}{\bfr{\vec{m}}} \| \encc{z}{M_j~y})
                &&\text{(\encpropref{i:transSubst})}
                \\
                {}\equiv{} & \nu{xy}\bencb{x}{\encc{z}{M_j~y} \| \nunil \encc{\_}{C}}{\bfr{\vec{m}}}
                &&\text{(\encpropref{i:transBufCtx})}
                \\
                {}={} & \nu{xy}\bencb{x}{\encc{z}{\main\,(M_j~y)} \| \nunil \encc{\_}{C}}{\bfr{\vec{m}}}
                \\
                {}={} & \nu{xy}\bencb{x}{\encc{z}{\main\,(M_j~y) \prl C}}{\bfr{\vec{m}}}
                \\
                {}={} & \encc{z}{\nu{x\bfr{\vec{m}}y}(\main\,(M_j~y) \prl C)}
            \end{align*}

        \item
            Rule~\scc{E-LiftC}: $\term{C} \reddC \term{C'} \implies \term{\mcl{G}[C]} \reddC \term{\mcl{G}[C']}$.
            By induction on the structure of $\term{\mcl{G}}$ (\ih{2}), assuming $\term{C} \reddC \term{C'}$.
            The base case ($\term{\mcl{G}} = \term{[]}$) follows from \ih{1}: $\encc{z}{C} \redd^\ast \encc{z}{C'}$.
            The case of buffered restriction ($\term{\mcl{G}} = \term{\nu{x\bfr{\vec{m}}y}\mcl{G}'}$) follows from \encpropref{i:transBufRed} and \ih{2}.
            The rest of the inductive cases follow from \ih{2} straightforwardly, since their encoding only places the encoding of $\term{C}$ under restriction and parallel composition.

        \item
            Rule~\scc{E-LiftM}: $\term{\mcl{M}} \reddM \term{\mcl{M}'} \implies \term{\mcl{F}[\mcl{M}]} \reddC \term{\mcl{F}[\mcl{M'}]}$.
            By induction on the structure of $\term{\mcl{F}}$ (\ih{2}), assuming $\term{\mcl{M}} \reddM \term{\mcl{M}'}$.
            The base case ($\term{\mcl{F}} = \term{\phi\,[]}$) follows from \Cref{t:transTermRed}: $\encc{z}{\phi\,\mcl{M}} = \encc{z}{\mcl{M}} \redd^\ast \encc{z}{\mcl{M}'} = \encc{z}{\phi\,\mcl{M}'}$.
            The inductive case ($\term{\mcl{F}} = \term{C\xsub{ (\mcl{F}'[\mcl{M}])/x }}$) follows from \ih{2} straightforwardly, since the encoding of $\term{\mcl{F}}$ places the encoding of $\term{\mcl{F}'[\mcl{M}]}$ only under restriction and parallel composition.

        \item
            Rule~\mbox{\scc{E-ConfLiftSC}}: $\term{C} \equivC \term{C'} \wedge \term{C'} \reddC \term{D'} \wedge \term{D'} \equivC \term{D} \implies \term{C} \reddC \term{D}$.
            We assume $\term{C} \equivC \term{C'}$, $\term{C'} \reddC \term{D'}$, and $\term{D'} \equivC \term{D}$.
            By \ih{1}, $\encc{z}{C'} \redd^\ast \encc{z}{D'}$.
            By \Cref{t:transConfSC}, $\encc{z}{C} \equiv \encc{z}{C'}$ and $\encc{z}{D} \equiv \encc{z}{D'}$.
            Hence, $\encc{z}{C} \redd^\ast \encc{z}{D}$.
            \qedhere
    \end{itemize}
\end{proof}

\subsection{Soundness}
\label{as:soundness}

\begin{definition}[Evaluation Context]
\label{d:redCtx}
    Evaluation contexts ($\mcl{E}$) are defined by the following grammar:
    \begin{align*}
        \mcl{E} &::= [] \sepr \mcl{E} \| P \sepr \nu{xy}\mcl{E}
    \end{align*}
    We write $\mcl{E}[P]$ to denote the process obtained by replacing the hole $[]$ in $\mcl{E}$ by $P$.
\end{definition}

\begin{definition}[Bound Forwarded Continuation]
\label{d:bcont}
    \sloppy
    The predicate $\bcont{x,y}(P)$ holds if and only if $P \equiv {\mcl{E}[\nu{xa}(x \fwd y \| \nu{cd}(c \fwd e \| \alpha))]}$ for some evaluation context $\mcl{E}$ where $\alpha \in \{d[f,a],d[a] \puts \ell\}$ implies $P \equiv \nu{eg}Q$ for some $Q$.
\end{definition}

\begin{definition}[Lazy Forwarded Semantics for APCP]
\label{d:apcpLazy}
    The \emph{lazy forwarder semantics} for APCP ($\reddL$) is defined by the rules in \Cref{f:apcpLazy}.

    \begin{figure}[t]
        \begin{mdframed}
            \begin{align*}
                & \rred{\overset{\leftrightarrow}{\scc{Id}}}
                &
                & &
                \nuf{yz}(x \fwd y \| P)
                &\reddL^{(x,y)} P \{x/z\}
                \\
                & \rred{\tensor \parr}
                &
                & &
                \nu{xy}(x[a,b] \| y(c,d) \sdot P)
                &\reddL^\cdot P\subst{a/c,b/d}
                \\
                & \rred{\overset{\leftrightarrow}{\tensor \parr}}
                &
                & &
                \nu{xy}(\nu{uv}(x \fwd u \| v[a,b]) \| \nu{wz}(y \fwd w \| z(c,d) \sdot P))
                &\reddL^\cdot P\subst{a/c,b/d}
                \\
                & \rred{\oplus \&}
                & j \in I \implies &
                &
                \nu{xy}(x[b] \puts j \| y(d) \gets \{i:P_i\}_{i \in I}
                &\reddL^\cdot P_j\subst{b/d}
                \\
                & \rred{\overset{\leftrightarrow}{\oplus \&}}
                &
                j \in I \implies
                & &
                \nu{xy}(\nu{uv}(x \fwd u \| v[b] \puts j) \| \nu{wz}(y \fwd w \| z(d) \gets \{i:P_i\}_{i \in I}))
                &\reddL^\cdot P_j\subst{b/d}
            \end{align*}
            For $S = \cdot$ or $S = (x,y)$ for some $x,y$:
            \begin{mathpar}
                \inferrule*[right=$\equiv$]
                { P \equiv P' \\ P' \reddL^S Q' \\ Q' \equiv Q }
                { P \reddL^S Q }
                \and
                \inferrule*[right=$\onu$]
                { P \reddL^S Q }
                { \nu{xy}P \reddL^S \nu{xy}Q }
                \and
                \inferrule*[right=$\|$]
                { P \reddL^S Q }
                { P \| R \reddL^S Q \| R }
                \and
                \inferrule*[right=$\cdot$]
                { P \reddL^\cdot Q }
                { P \reddL Q }
                \and
                \inferrule*[right={$(x,y)$}]
                { P \reddL^{(x,y)} Q \\ \bcont{x,y}(P) }
                { P \reddL Q }
            \end{mathpar}
        \end{mdframed}
        \caption{Lazy forwarder semantics for APCP (cf.\ \Cref{d:apcpLazy})}
        \label{f:apcpLazy}
    \end{figure}
\end{definition}

\noindent
This semantics is designed to have a very controlled reduction of forwarders.
Rule~$\rred{\overset{\leftrightarrow}{\scc{Id}}}$ implements the forwarder as a substitution, as usual.
However, it only does so if it is bound by a forwarder-enabled restriction.
Moreover, it annotates the reduction's arrow with the forwarded endpoints.
For example, $\nuf{yz}(x \fwd y \| P) \reddL^{(x,y)} P \subst{x/z}$.
The Closure rules~$\equiv$, $\onu$, and $\|$ preserve the arrow's annotation.
To derive the final unannotated reduction, Rule~$(x,y)$ can be applied.
This rule checks that, when one of the endpoints $x$ or $y$ is bound to the continuation endpoint of a forwarded output or selection, that output or selection is forwarded on a bound endpoint (cf.\ Def.\ \labelcref{d:bcont}).

Rules~$\rred{\tensor \parr}$ and $\rred{\oplus \&}$ are standard, and annotate the reduction's arrow with `$\cdot$'.
The closure rules preserve this annotation as well.
The final unannotated reduction can be derived using Rule~$\cdot$, which has no additional conditions.
Finally, Rules~$\rred{\overset{\leftrightarrow}{\tensor \parr}}$ and $\rred{\overset{\leftrightarrow}{\oplus \&}}$ define a short-circuit semantics for forwarded output (resp.\ selection) connected to forwarded input (resp.\ branching): instead of first reducing the two involved forwarders (which may not be possible under $\reddL$), communication occurs \emph{through} the forwarders whereafter the forwarders will have been consumed.
These rules also annotate the arrow with~`$\cdot$', such that there are no additional conditions when deriving the final unannotated reduction.

Before we prove soundness, let us compare APCP's standard semantics and the newly introduced lazy forwarder semantics, and their impact on soundness:

\begin{example}
    Consider the following configuration:
    \begin{align*}
        \term{\main\,((\lambda u \sdot \lambda w \sdot u)~(\sff{send}~(M,x)))}
        \reddC \term{\main\,((\lambda w \sdot u)\xsub{ \sff{send}~(M,x)/u })}
        \reddC \term{\main\,((\lambda w \sdot u)\xsub{ \sff{send}'(M,x)/u })}
    \end{align*}
    Now consider the encoding of the reduced configuration:
    \begin{align*}
        & \encc{z}{\main\,((\lambda w \sdot u)\xsub{ \sff{send}'(M,x)/u })}
        \\
        {}={} & \nuf{ua}\begin{array}[t]{@{}l@{}}
            (z(b,c) \sdot \nuf{dw}(\nunil b[d,\_] \| u \fwd c)
            \\
            {} \| \nu{ef}\begin{array}[t]{@{}l@{}}
                (e(g,\_) \sdot \encc{g}{M}
                \\
                {} \| \nu{hk}\begin{array}[t]{@{}l@{}}
                    (x \fwd h
                    \\
                    {} \| \nu{lm}(k[f,l] \| m \fwd a))))
                \end{array}
            \end{array}
        \end{array}
    \end{align*}
    Due to the forwarders, this process can reduce in two ways:
    \begin{align*}
        &&
        & \encc{z}{\main\,((\lambda w \sdot u)\xsub{ \sff{send}'(M,x)/u })}
        \\[10pt]
        & \text{(i)}
        &
        {}\redd{} & \nuf{ua}\begin{array}[t]{@{}l@{}}
            (z(b,c) \sdot \nuf{dw}(\nunil b[d,\_] \| u \fwd c)
            \\
            {} \| \nu{ef}\begin{array}[t]{@{}l@{}}
                (e(g,\_) \sdot \encc{g}{M}
                \\
                {} \| \nu{lm}(x[f,l] \| m \fwd a))))
            \end{array}
        \end{array}
        \\[10pt]
        &\text{(ii)}
        &
        {}\redd{} & \nu{ml}\begin{array}[t]{@{}l@{}}
            (z(b,c) \sdot \nuf{dw}(\nunil b[d,\_] \| m \fwd c)
            \\
            {} \| \nu{ef}\begin{array}[t]{@{}l@{}}
                (e(g,\_) \sdot \encc{g}{M}
                \\
                {} \| \nu{hk}(x \fwd h \| k[f,l])))
            \end{array}
        \end{array}
    \end{align*}
    These processes are not encodings of any configuration, thus disproving the encoding's soundness under APCP's standard semantics.

    On the other hand, the encoding of the reduced configuration does not reduce under the lazy forwarder semantics: the forwarder $x \fwd h$ is not bound by forwarder-enabled restrictions, and the forwarder $m \fwd a$ forwards the continuation endpoint of an output on an endpoint that is forwarded to a free endpoint.
    \begin{align*}
        \encc{z}{\main\,((\lambda w \sdot u)\xsub{ \sff{send}'(M,x)/u })}~ \nreddL
    \end{align*}
    Now consider the same configuration, but within a buffered restriction:
    \begin{align*}
        & \encc{z}{\nu{x\bfr{\epsilon}y}(\main\,((\lambda w \sdot u)\xsub{ \sff{send}'(M,x)/u }) \prl C)}
        \\
        {}={} & \nu{xy}\begin{array}[t]{@{}l@{}}
            (\nuf{ua}\begin{array}[t]{@{}l@{}}
                (z(b,c) \sdot \nuf{dw}(\nunil b[d,\_] \| u \fwd c)
                \\
                {} \| \nu{ef}\begin{array}[t]{@{}l@{}}
                    (e(g,\_) \sdot \encc{g}{M}
                    \\
                    {} \| \nu{hk}\begin{array}[t]{@{}l@{}}
                        (x \fwd h
                        \\
                        {} \| \nu{lm}(k[f,l] \| m \fwd a))))
                    \end{array}
                \end{array}
            \end{array}
            \\
            {} \| \nunil \encc{\_}{C})
        \end{array}
    \end{align*}
    Now the forwarder $m \fwd a$ may reduce, because $x$ is bound:
    \begin{align*}
        & \encc{z}{\nu{x\bfr{\epsilon}y}(\main\,((\lambda w \sdot u)\xsub{ \sff{send}'(M,x)/u }) \prl C)}
        \\
        {}\reddL{} & \nu{xy}\begin{array}[t]{@{}l@{}}
            (\nu{ef}\nu{lm}\begin{array}[t]{@{}l@{}}
                (\nu{hk}(x \fwd h \| k[f,l])
                \\
                {} \| e(g,\_) \sdot \encc{g}{M}
                \\
                {} \|
                z(b,c) \sdot \nuf{dw}(\nunil b[d,\_] \| m \fwd c))
            \end{array}
            \\
            {} \| \nunil \encc{\_}{C})
        \end{array}
        \\
        {}={}
        & \encc{z}{\nu{x\bfr{M}y}(\main\,(\lambda w \sdot x) \prl C)}
    \end{align*}
    In this case, soundness is not disproven, because
    \begin{align*}
        \term{\nu{x\bfr{\epsilon}y}(\main\,((\lambda w \sdot u)\xsub{ \sff{send}'(M,x)/u }) \prl C)} \reddC \term{\nu{x\bfr{M}y}(\main\,(\lambda w \sdot x) \prl C)}.
    \end{align*}
\end{example}

\thmTransSndConf*

\begin{proof}
    By induction on the structure of $\term{C}$ (\ih{1}).
    In each case for $\term{C}$, we use induction on the number $k$ of steps $\encc{z}{C} \reddL^k Q$ (\ih{2}).
    We observe the initial reduction $\encc{z}{C} \reddL Q_0$ and discuss all possible following reductions.
    Then, we isolate $k'$ steps  such that $\encc{z}{C} \reddL^{k'} \encc{z}{D'}$ for some $\term{D'}$ such that $\term{C} \reddC \term{D'}$ ($k'$ may be different in each case).
    Since then $\encc{z}{D'} \reddL^{k-k'} Q$, it follows from \ih{2} that there exists $\term{D}$ such that $\term{D'} \reddC^+ \term{D}$ and $\encc{z}{D'} \reddL^\ast \encc{z}{D}$.
    \begin{itemize}
        \item
            Case $\term{C} = \term{\phi\,\mbb{M}}$.
            By construction, we can identify a maximal context $\term{\mcl{R}}$ and a term $\term{\mbb{M}'}$ such that the observed reduction $\encc{z}{\phi\,(\mcl{R}[\mbb{M}'])} \reddL Q_0$ originates from the encoding of $\term{\mbb{M}'}$ directly, i.e.\ not solely from the encoding of a subterm of $\term{\mbb{M}'}$.
            We discuss all syntactical cases for $\term{\mbb{M}'}$, though not all cases may show a reduction.
            \begin{itemize}
                \item
                    Case $\term{\mbb{M}'} = \term{x}$.
                    By induction on the structure of $\term{\mcl{R}}$.
                    In the base case ($\term{\mcl{R}} = \term{[]}$), there are two encodings of this term, depending on its typing: $\encc{z}{x} = x \fwd z$ and $\encc{z}{x} = \0$.
                    Either way, the encoding does not reduce, so this case does not cause the reduction $\encc{z}{\phi\,(\mcl{R}[x])}$.
                    The inductive cases are straightforward.

                \item
                    Case $\term{\mbb{M}'} = \term{\sff{new}}$.
                    \begin{align*}
                        & \encc{z}{\sff{new}}
                        \\
                        {}={} & \nu{ab}(\nunil a[\_,\_] \| b(\_,\_) \sdot \nu{xy}\encc{z}{(x,y)})
                    \end{align*}
                    The reduction $\encc{z}{\phi\,(\mcl{R}[\sff{new}])} \reddL Q_0$ is due to a synchronization between the output on $a$ and the input on $b$.
                    Let $\term{D'} := \term{\nu{x\bfr{\epsilon}y}(\phi\,(\mcl{R}[(x,y)]))}$.
                    By induction on the structure of $\term{\mcl{R}}$ (\ih{3}), we show that $\encc{z}{C} \reddL \encc{z}{D'}$.
                    We discuss the base case ($\term{\mcl{R}} = \term{[]}$) and one representative, inductive case ($\term{\mcl{R}} = \term{\mcl{R}'~N}$).
                    \begin{itemize}
                        \item
                            Base case: $\term{\mcl{R}} = \term{[]}$.
                            There is only one reduction possible:
                            \begin{align*}
                                & \encc{z}{\phi\,\sff{new}}
                                \\
                                {}\reddL{} & \nu{xy}\encc{z}{(x,y)}
                                \\
                                {}={} & \encc{z}{\nu{x\bfr{\epsilon}y}(\phi\,(x,y))}
                                \\
                                {}={} & \encc{z}{D'}
                            \end{align*}

                        \item
                            Representative, inductive case: $\term{\mcl{R}} = \term{\mcl{R}'~N}$.
                            \begin{align*}
                                & \encc{z}{\phi\,(\mcl{R}'[\sff{new}]~N)}
                                \\
                                {}={} & \nu{ab}(\encc{a}{\mcl{R}'[\sff{new}]} \| \nu{cd}(b[c,z] \| d(e,\_) \sdot \encc{e}{N}))
                                \\
                                {}={} & \nu{ab}(\encc{a}{\phi\,(\mcl{R}'[\sff{new}])} \| \nu{cd}(b[c,z] \| d(e,\_) \sdot \encc{e}{N}))
                                \\
                                {}\reddL{} & \nu{ab}(\encc{a}{\nu{x\bfr{\epsilon}y}(\phi\,(\mcl{R'}[(x,y)]))} \| \nu{cd}(b[c,z] \| d(e,\_) \sdot \encc{e}{N}))
                                & \text{(\ih{3})}
                                \\
                                {}={} & \nu{ab}(\nu{xy}\encc{a}{\mcl{R'}[(x,y)]} \| \nu{cd}(b[c,z] \| d(e,\_) \sdot \encc{e}{N}))
                                \\
                                {}\equiv{} & \nu{xy}\nu{ab}(\encc{a}{\mcl{R'}[(x,y)]} \| \nu{cd}(b[c,z] \| d(e,\_) \sdot \encc{e}{N}))
                                \\
                                {}={} & \encc{z}{\nu{x\bfr{\epsilon}y}(\phi\,(\mcl{R'}[(x,y)]~N))}
                                \\
                                {}={} & \encc{z}{D'}
                            \end{align*}
                    \end{itemize}
                    Indeed, $\term{C} \reddC \term{D'}$.
                    Since $\encc{z}{D'} \reddL^{k-1} Q$, the thesis follows by \ih{2}.

                \item
                    Case $\term{\mbb{M}'} = \term{\sff{spawn}~\mbb{M}''}$.
                    \begin{align*}
                        & \encc{z}{\sff{spawn}~\mbb{M}''}
                        \\
                        {}={} & \nu{ab}(\encc{a}{\mbb{M}''} \| b(c,d) \sdot (\nunil c[\_,\_] \| \nunil d[z,\_]))
                    \end{align*}
                    The reduction $\encc{z}{\phi\,(\mcl{R}[\sff{spawn}~\mbb{M}''])} \reddL Q_0$ is due to a synchronization between the input on $b$ and an output on $a$ in $\encc{a}{\mbb{M}''}$.
                    By typability, this only occurs if $\term{\mbb{M}''} = \term{(M_1,M_2)\xsub{ \mbb{L}_1/x_1,\ldots,\mbb{L}_n/x_n }}$.
                    Since we can always lift explicit substitutions all the way to the top of a configuration, we omit them here.
                    Let $\term{D'} := \term{\phi\,(\mcl{R}[M_2]) \prl \child\,M_1}$.
                    By induction on the structure of $\term{\mcl{R}}$ (\ih{3}), we show that $\encc{z}{C} \reddL^3 \encc{z}{D'}$.
                    We discuss the base case ($\term{\mcl{R}} = \term{[]}$) and one representative, inductive case ($\term{\mcl{R}} = \term{\mcl{R}'~N}$).
                    \begin{itemize}
                        \item
                            Base case: $\term{\mcl{R}} = \term{[]}$.
                            \begin{align*}
                                & \encc{z}{\phi\,(\sff{spawn}~(M_1,M_2))}
                                \\
                                {}={} & \nu{ab}(\nu{fg}\nu{hl}(a[f,h] \| g(m,\_) \sdot \encc{m}{M_1} \| l(o,\_) \sdot \encc{o}{M_2})\\
                                & {} \| b(c,d) \sdot (\nunil c[\_,\_] \| \nunil d[z,\_]))
                                \\
                                {}\reddL{} & \nu{fg}\nu{hl}(g(m,\_) \sdot \encc{m}{M_1} \| l(o,\_) \sdot \encc{o}{M_2} \| \nunil f[\_,\_] \| \nunil h[z,\_])
                                \\
                                {}\reddL{} & \nu{hl}(\nunil \encc{\_}{M_1} \| l(o,\_) \sdot \encc{o}{M_2} \| \nunil h[z,\_])
                                \\
                                {}\reddL{} & \nunil \encc{\_}{M_1} \| \encc{z}{M_2}
                                \\
                                {}\equiv{} & \encc{z}{M_2} \| \nunil \encc{\_}{M_1}
                                \\
                                {}={} & \encc{z}{\phi\,M_2 \prl \child\,M_1}
                                \\
                                {}={} & \encc{z}{D'}
                            \end{align*}

                        \item
                            Representative inductive case: $\term{\mcl{R}} = \term{\mcl{E'}~N}$.
                            \begin{align*}
                                & \encc{z}{\phi\,(\mcl{R}'[\sff{spawn}~(M_1,M_2)]~N)}
                                \\
                                {}={} & \nu{ab}(\encc{a}{\mcl{R}'[\sff{spawn}~(M_1,M_2)]} \| \nu{cd}(b[c,z] \| d(e,\_) \sdot \encc{e}{N}))
                                \\
                                {}={} & \nu{ab}(\encc{a}{\phi\,(\mcl{R}'[\sff{spawn}~(M_1,M_2)])} \| \nu{cd}(b[c,z] \| d(e,\_) \sdot \encc{e}{N}))
                                \\
                                {}\reddL^3{} & \nu{ab}(\encc{a}{\phi\,(\mcl{R}'[M_2]) \prl \child\,M_1} \| \nu{cd}(b[c,z] \| d(e,\_) \sdot \encc{e}{N}))
                                & \text{(\ih{3})}
                                \\
                                {}={} & \nu{ab}(\encc{a}{\mcl{R}'[M_2]} \| \nunil \encc{\_}{\child\,M_1} \| \nu{cd}(b[c,z] \| d(e,\_) \sdot \encc{e}{N}))
                                \\
                                {}\equiv{} & \nu{ab}(\encc{a}{\mcl{R}'[M_2]} \| \nu{cd}(b[c,z] \| d(e,\_) \sdot \encc{e}{N})) \| \nunil \encc{\_}{\child\,M_1}
                                \\
                                {}={} & \encc{z}{\phi\,(\mcl{R}'[M_2]~N) \prl \child\,M_1}
                                \\
                                {}={} & \encc{z}{D'}
                            \end{align*}
                    \end{itemize}
                    Indeed, $\term{C} \reddC \term{D'}$.
                    Since $\encc{z}{D'} \reddL^{k-3} Q$, the thesis follows by \ih{2}.

                \item
                    Case $\term{\mbb{M}'} = \term{\sff{send}~\mbb{M}''}$.
                    \begin{align*}
                        & \encc{z}{\sff{send}~\mbb{M}''}
                        \\
                        {}={} & \nu{ab}(\encc{a}{\mbb{M}''} \| b(c,d) \sdot \nu{ef}(\nunil d[e,\_] \| \nu{gh}(f[c,g] \| h \fwd z)))
                    \end{align*}
                    The reduction $\encc{z}{\phi\,(\mcl{R}[\sff{send}~\mbb{M}''])} \reddL Q_0$ is due to a synchronization between the input on $b$ and an output on $a$ in $\encc{a}{\mbb{M}''}$.
                    By typability, this only occurs if $\term{\mbb{M}''} = \term{(M_1,M_2)\xsub{ \mbb{L}_1/x_1,\ldots,\mbb{L}_n/x_n }}$.
                    As in the case above, we omit the explicit substitutions.
                    Let $\term{D'} = \term{\phi\,(\mcl{R}[\sff{send'}(M_1,M_2)])}$.
                    By induction on the structure of $\term{\mcl{R}}$, we show that $\encc{z}{C} \reddL^2 \encc{z}{D'}$.
                    In the base case ($\term{\mcl{R}} = \term{[]}$), we have the following:
                    \begin{align*}
                        & \encc{z}{\phi\,(\sff{send}~(M_1,M_2))}
                        \\
                        {}={} & \nu{ab}(\nu{mo}\nu{pq}(a[m,p] \| o(r,\_) \sdot \encc{r}{M_1} \| q(s,\_) \sdot \encc{s}{M_2})\\
                        & {} \| b(c,d) \sdot \nu{ef}(\nunil d[e,\_] \| \nu{gh}(f[c,g] \| h \fwd z)))
                        \\
                        {}\reddL{} & \nu{mo}\nu{pq}(o(r,\_) \sdot \encc{r}{M_1} \| q(s,\_) \sdot \encc{s}{M_2})\\
                        & {} \| \nu{ef}(\nunil p[e,\_] \| \nu{gh}(f[m,g] \| h \fwd z)))
                        \\
                        {}\reddL{} & \nu{ef}\nu{mo}(o(r,\_) \sdot \encc{r}{M_1} \| \encc{e}{M_2}) \| \nu{gh}(f[m,g] \| h \fwd z)))
                        \\
                        {}\equiv{} & \nu{om}(o(r,\_) \sdot \encc{r}{M_1} \| \nu{ef}(\encc{e}{M_2}) \| \nu{gh}(f[m,g] \| h \fwd z))
                        \\
                        {}={} & \encc{z}{\phi\,(\sff{send}'(M_1,M_2))}
                        \\
                        {}={} & \encc{z}{D'}
                    \end{align*}
                    The inductive cases are straightforward.
                    Indeed, $\term{C} \reddC \term{D'}$.
                    Since $\encc{z}{D'} \reddL^{k-2} Q$, the thesis follows by \ih{2}.

                \item
                    Case $\term{\mbb{M}'} = \term{\sff{recv}~\mbb{M}''}$.
                    \begin{align*}
                        & \encc{z}{\sff{recv}~\mbb{M}''}
                        \\
                        {}={} & \nu{ab}(\encc{a}{\mbb{M}''} \| b(c,d) \sdot \nu{ef}(z[c,e] \| f(g,\_) \sdot d \fwd g))
                    \end{align*}
                    The reduction $\encc{z}{\phi\,(\mcl{R}[\sff{recv}~\mbb{M}''])} \reddL Q_0$ would be due to a synchronization between the input on $b$ and an output on $a$ in $\encc{a}{\mbb{M}''}$.
                    However, by typability, there is no $\term{\mbb{M}''}$ for which the encoding has an unguarded output on $a$, so this case does not cause the reduction.

                \item
                    Case $\term{\mbb{M}'} = \term{\lambda x \sdot M''}$.
                    The encoding $\encc{z}{\lambda x \sdot M''} = z(a,b) \sdot \nuf{cx}(\nunil a[c,\_] \| \encc{b}{M''})$ does not reduce, so this case does not cause the reduction $\encc{z}{\phi\,(\mcl{R}[\lambda x \sdot M''])} \reddL Q_0$.

                \item
                    Case $\term{\mbb{M}'} = \term{\mbb{M}_1~M_2}$.
                    \begin{align*}
                        & \encc{z}{\mbb{M}_1~M_2}
                        \\
                        {}={} & \nu{ab}(\encc{a}{\mbb{M}_1} \| \nu{cd}(b[c,z] \| d(e,\_) \sdot \encc{e}{M_2}))
                    \end{align*}
                    The reduction $\encc{z}{\phi\,(\mcl{R}[\mbb{M}_1~M_2])} \reddL Q_0$ is due to a synchronization between the input on $b$ and an output on $a$ in $\encc{a}{\mbb{M}_1}$.
                    By typability, this only occurs if $\term{\mbb{M}_1} = \term{(\lambda x \sdot M'_1)\xsub{ \mbb{L}_1/x_1,\ldots,\mbb{L}_n/x_n }}$.
                    Again, we omit the explicit substitutions.
                    Let $\term{D'} = \term{M'_1\xsub{ \mbb{M}_2/x }}$.
                    By induction on the structure of $\term{\mcl{R}}$, we show that $\encc{z}{C} \reddL^2 \encc{z}{D'}$.
                    In the base case ($\term{\mcl{R}} = \term{[]}$), we have the following:
                    \begin{align*}
                        & \encc{z}{\phi\,((\lambda x \sdot M'_1)~M_2)}
                        \\
                        {}={} & \nu{ab}(a(g,h) \sdot \nuf{lx}(\nunil g[l,\_] \| \encc{h}{M'_1}) \| \nu{cd}(b[c,z] \| d(e,\_) \sdot \encc{e}{M_2}))
                        \\
                        {}\reddL{} & \nu{cd}(\nuf{lx}(\nunil c[l,\_] \| \encc{z}{M'_1}) \| d(e,\_) \sdot \encc{e}{M_2})
                    \end{align*}
                    At this point, there are three ways for this process to reduce: $\encc{z}{M'_1}$ can reduce internally, a forwarder on $x$ in $\encc{z}{M'_1}$ can reduce, or there is a synchronization between the output on $c$ and the input on $d$.
                    These reductions are independent from each other, we so we postpone the former two cases and focus on the latter.
                    \begin{align*}
                        {}\reddL{} & \nuf{lx}(\encc{z}{M'_1} \| \encc{l}{M_2})
                        \\
                        {}\equiv{} & \nuf{lx}(\encc{l}{M_2} \| \encc{z}{M'_1})
                        \\
                        {}={} & \encc{z}{\phi\,(M'_1\xsub{ M_2/x })}
                        \\
                        {}={} & \encc{z}{D'}
                    \end{align*}
                    The inductive cases are straightforward.
                    Indeed, $\term{C} \reddC \term{D'}$.
                    Since $\encc{z}{D'} \reddL^{k-2} Q$, the thesis follows by \ih{2}.

                \item
                    Case $\term{\mbb{M}'} = \term{()}$.
                    The encoding $\encc{z}{()} = \0$ does not reduce, so this case does not cause the reduction $\encc{z}{\phi\,(\mcl{R}[()])} \reddL Q_0$.

                \item
                    Case $\term{\mbb{M}'} = \term{(M_1,M_2)}$.
                    The encoding $\encc{z}{(M_1,M_2)} = \nu{ab}\nu{cd}(z[a,c] \| b(e,\_) \sdot \encc{e}{M_1} \| d(f,\_) \sdot \encc{f}{M_2})$ does not reduce, so this case does not cause the reduction $\encc{z}{\phi\,(\mcl{R}[(M_1,M_2)])} \reddL Q_0$.

                \item
                    Case $\term{\mbb{M}'} = \term{\sff{let}\,(x,y)=\mbb{M}_1\,\sff{in}\,M_2}$.
                    \begin{align*}
                        & \encc{z}{\sff{let}\,(x,y)=\mbb{M}_1\,\sff{in}\,M_2}
                        \\
                        {}={} & \nu{ab}(\encc{a}{\mbb{M}_1} \| b(c,d) \sdot \nuf{ex}\nuf{fy}(\nunil c[e,\_] \| \nunil d[f,\_] \| \encc{z}{M_2}))
                    \end{align*}
                    The reduction $\encc{z}{\phi\,(\mcl{R}[\sff{let}\,(x,y)=\mbb{M}_1\,\sff{in}\,M_2])} \reddL Q_0$ is due to a synchronization between the input on $b$ and an output on $a$ in $\encc{a}{\mbb{M}_1}$.
                    By typability, this only occurs if $\term{\mbb{M}_1} = \term{(M'_1,M'_2)\xsub{ \mbb{L}_1/x_1,\ldots,\mbb{L}_n/x_n }}$.
                    Again, we omit the explicit substitutions.
                    Let $\term{D'} = \term{M_2\xsub{ M'_1/x,M'_2/y }}$.
                    By induction on the structure of $\term{\mcl{R}}$, we show that $\encc{z}{C} \reddL^3 \encc{z}{D'}$.
                    In the base case ($\term{\mcl{R}} = \term{[]}$), we have the following:
                    \begin{align*}
                        & \encc{z}{\phi\,(\sff{let}\,(x,y)=(M'_1,M'_2)\,\sff{in}\,M_2)}
                        \\
                        {}={} & \nu{ab}(\nu{hl}\nu{mo}(a[h,m] \| l(p,\_) \sdot \encc{p}{M'_1} \| o(r,\_) \sdot \encc{r}{M'_2})\\
                        &{}\| b(c,d) \sdot \nuf{ex}\nuf{fy}(\nunil c[e,\_] \| \nunil d[f,\_] \| \encc{z}{M_2}))
                        \\
                        {}\reddL{} & \nu{hl}\nu{mo}(l(p,\_) \sdot \encc{p}{M'_1} \| o(r,\_) \sdot \encc{r}{M'_2})\\
                        &{}\| \nuf{ex}\nuf{fy}(\nunil h[e,\_] \| \nunil m[f,\_] \| \encc{z}{M_2}))
                    \end{align*}
                    At this point, there are three ways for this process to reduce: $\encc{z}{M_2}$ can reduce internally, or there can be a synchronization between the output on $h$ and the input on $l$ or between the output on $m$ and the input on $o$.
                    These reductions are independent from each other, se we postpone the former case and focus on the latter two, which can occur in any order; w.l.o.g.\ we assume the synchronization between $h$ and $l$ occurs first.
                    \begin{align*}
                        {}\reddL{} & \nuf{ex}\nu{mo}(\encc{e}{M'_1} \| o(r,\_) \sdot \encc{r}{M'_2} \| \nuf{fy}(\nunil m[f,\_] \| \encc{z}{M_2}))
                        \\
                        {}\reddL{} & \nuf{ex}\nuf{fy}(\encc{e}{M'_1} \| \encc{f}{M'_2} \| \encc{z}{M_2})
                        \\
                        {}\equiv{} & \nuf{yf}(\nuf{xe}(\encc{z}{M_2} \| \encc{e}{M'_1}) \| \encc{f}{M'_2})
                        \\
                        {}={} & \encc{z}{\phi\,(M_2\xsub{ M'_1/x,M'_2/y })}
                        \\
                        {}={} & \encc{z}{D'}
                    \end{align*}
                    The inductive cases are straightforward.
                    Indeed, $\term{C} \reddC \term{D'}$.
                    Since $Q' \reddL^{k-3} Q$, the thesis follows by \ih{2}.

                \item
                    Case $\term{\mbb{M}'} = \term{\sff{select}\,\ell\,\mbb{M}''}$.
                    \begin{align*}
                        & \encc{z}{\sff{select}\,\ell\,\mbb{M}''}
                        \\
                        {}={} & \nu{ab}(\encc{a}{\mbb{M}''} \| \nu{cd}(b[c] \puts \ell \| d \fwd z))
                    \end{align*}
                    The reduction $\encc{z}{\phi\,(\mcl{R}[\sff{select}\,\ell\,\mbb{M}''])} \reddL Q_0$ would be due to a synchronization between the selection on $b$ and a branch on $a$ in $\encc{a}{\mbb{M}''}$.
                    However, by typability, there is no $\term{\mbb{M}''}$ for which the encoding has an unguarded case on $a$, so this case does not cause the reduction.

                \item
                    Case $\term{\mbb{M}'} = \term{\sff{case}\,\mbb{M}''\,\sff{of}\,\{i:M_i\}_{i \in I}}$.
                    \begin{align*}
                        & \encc{z}{\sff{case}\,\mbb{M}''\,\sff{of}\,\{i:M_i\}_{i \in I}}
                        \\
                        {}={} & \nu{ab}(\encc{a}{\mbb{M}''} \| b(c) \gets \{i:\encc{z}{M_i~c}\}_{i \in I})
                    \end{align*}
                    The reduction $\encc{z}{\phi\,(\mcl{R}[\sff{case}\,\mbb{M}''\,\sff{of}\,\{i:M_i\}_{i \in I}])} \reddL Q_0$ would be due to a synchronization between the branch on $b$ and a selection on $a$ in $\encc{a}{\mbb{M}''}$.
                    However, by typability, there is no $\term{\mbb{M}''}$ for which the encoding has an unguarded selection on $a$, so this case does not cause the reduction.

                \item
                    Case $\term{\mbb{M}'} = \term{\mbb{M}_1\xsub{ \mbb{M}_2/x }}$.
                    \begin{align*}
                        & \encc{z}{\mbb{M}_1\xsub{ \mbb{M}_2/x }}
                        \\
                        {}={} & \nuf{xa}(\encc{z}{\mbb{M}_1} \| \encc{a}{\mbb{M}_2})
                    \end{align*}
                    The reduction $\encc{z}{\phi\,(\mcl{R}\big[\mbb{M}_1\xsub{ \mbb{M}_2/x }\big])} \reddL Q_0$ can be due to two cases: a forwarder on $x$ in $\encc{z}{\mbb{M}_1}$, or a forwarder on $a$ in $\encc{a}{\mbb{M}_2}$.
                    \begin{itemize}
                        \item
                            A forwarder on $x$ in $\encc{z}{\mbb{M}_1}$ reduces.
                            By typability, this only occurs if $\term{\mbb{M}_1} = \term{\mcl{R}'[x]}$.
                            Let $\term{D'} = \term{\phi\,(\mcl{R}\big[\mcl{R}'[\mbb{M}_2]\big])}$.
                            By induction on the structures of $\term{\mcl{R}}$ (\ih{3}) and $\term{\mcl{R}'}$ (\ih{4}), we show that $\encc{z}{C} \reddL \encc{z}{D'}$.
                            We consider the base case ($\term{\mcl{R}} = \term{\mcl{R}'} = \term{[]}$) and a representative, inductive case ($\term{\mcl{R}} = \term{[]}$ and $\term{\mcl{R}'} = \term{\mcl{R}''~N}$):
                            \begin{itemize}
                                \item
                                    Case $\term{\mcl{R}} = \term{\mcl{R}'} = \term{[]}$.
                                    \begin{align*}
                                        & \encc{z}{\phi\,(x\xsub{ \mbb{M}_2/x })}
                                        \\
                                        {}={} & \nuf{xa}(x \fwd z \| \encc{a}{\mbb{M}_2})
                                        \\
                                        {}\reddL{} & \encc{z}{\mbb{M}_2}
                                        \\
                                        {}={} & \encc{z}{\phi\,\mbb{M}_2}
                                        \\
                                        {}={} & \encc{z}{D'}
                                    \end{align*}

                                \item
                                    Case $\term{\mcl{R}} = \term{[]}$ and $\term{\mcl{R}'} = \term{\mcl{R}''~N}$.
                                    \begin{align*}
                                        & \encc{z}{\phi\,((\mcl{R}''[x]~N)\xsub{ \mbb{M}_2/x })}
                                        \\
                                        {}\equiv{} & \encc{z}{\phi\,((\mcl{R}''[x])\xsub{ \mbb{M}_2/x }~N)}
                                        & \text{(\Cref{t:transConfSC})}
                                        \\
                                        {}={} & \nu{ab}(\encc{a}{(\mcl{R}''[x])\xsub{ \mbb{M}_2/x }} \| \nu{cd}(b[c,z] \| d(e,\_) \sdot \encc{e}{N}))
                                        \\
                                        {}={} & \nu{ab}(\encc{a}{\phi\,((\mcl{R}''[x])\xsub{ \mbb{M}_2/x })} \| \nu{cd}(b[c,z] \| d(e,\_) \sdot \encc{e}{N}))
                                        \\
                                        {}\reddL{} & \nu{ab}(\encc{a}{\phi\,(\mcl{R}''[\mbb{M}_2])} \| \nu{cd}(b[c,z] \| d(e,\_) \sdot \encc{e}{N}))
                                        & \text{(\ih{4})}
                                        \\
                                        {}={} & \nu{ab}(\encc{a}{\mcl{R}''[\mbb{M}_2]} \| \nu{cd}(b[c,z] \| d(e,\_) \sdot \encc{e}{N}))
                                        \\
                                        {}={} & \encc{z}{\phi\,((\mcl{R}''[\mbb{M}_2])~N)}
                                        \\
                                        {}={} & \encc{z}{D'}
                                    \end{align*}
                            \end{itemize}
                            For $\term{\mcl{R}}$, the inductive cases are straightforward.
                            Indeed, $\term{C} \reddC \term{D'}$.
                            Since $\encc{z}{D'} \reddL^{k-1} Q$, the thesis follows by \ih{2}.

                        \item
                            A forwarder on $a$ in $\encc{a}{\mbb{M}_2}$ reduces.
                            There are three cases in which a forwarder on $a$ occurs unguarded in $\encc{a}{\mbb{M}_2}$ (eliding possible further explicit substitutions): $\term{\mbb{M}_2}$ is a name, or $\term{\mbb{M}_2}$ contains a $\term{\sff{send}'}$ or a $\term{\sff{select}}$ on a name inside a reduction context.
                            \begin{itemize}
                                \item
                                    $\term{\mbb{M}_2}$ is a name, i.e.\ $\term{\mbb{M}_2} = w$.
                                    Let $\term{D'} = \term{\phi\,(\mbb{M}_1\{w/x\})}$.
                                    \begin{align*}
                                        & \encc{z}{\phi\,(\mbb{M}_1\xsub{ w/x })}
                                        \\
                                        {}={} & \nuf{xa}(\encc{z}{\mbb{M}_1} \| w \fwd a)
                                        \\
                                        {}\reddL{} & \encc{z}{\mbb{M}_1}\{w/x\}
                                        \\
                                        {}\equiv{} & \encc{z}{\mbb{M}_1\{w/x\}}
                                        &&
                                        & \text{(\encpropref{i:transSubst})}
                                        \\
                                        {}={} & \encc{z}{\phi\,(\mbb{M}_1\{w/x\})}
                                        \\
                                        {}={} & \encc{z}{D'}
                                    \end{align*}
                                    Indeed, $\term{C} \reddC \term{D'}$.
                                    Since $\encc{z}{D'} \reddL^{k-1} Q$, the thesis follows by \ih{2}.

                                \item
                                    $\term{\mbb{M}_2}$ contains a $\term{\sff{send}'}$ on a name inside a reduction context, i.e.\ $\term{\mbb{M}_2} = \term{\mcl{R}[\sff{send}'(M,w)]}$.
                                    We apply induction on the structure of $\term{\mcl{R}}$.
                                    In the base case ($\term{\mcl{R}} = \term{[]}$), we have the following:
                                    \begin{align*}
                                        & \encc{z}{\phi\,(\mbb{M}_1\xsub{ \sff{send}'(M,w)/x })}
                                        \\
                                        {}={} & \nuf{xa}(\encc{z}{\mbb{M}_1} \| \nu{bc}(b(d,\_) \sdot \encc{d}{M} \| \nu{ef}(w \fwd e \| \nu{gh}(f[c,g] \| h \fwd a))))
                                    \end{align*}
                                    Notice that the forwarder $h \fwd a$ forwards the continuation endpoint of an output forwarded on $w$.
                                    However, $w$ is not bound, so this forwarder reduction is not enabled: this case does not cause the reduction.
                                    The inductive cases for $\term{\mcl{R}}$ are straightforward.

                                \item
                                    $\term{\mbb{M}_2}$ contains a $\term{\sff{select}}$ on a name inside a reduction context, i.e.\ $\term{\mbb{M}_2} = \term{\mcl{R}[\sff{select}\,\ell\,w)]}$.
                                    This case is largely analogous to the case above.
                            \end{itemize}
                    \end{itemize}

                \item
                    Case $\term{\mbb{M}'} = \term{\sff{send}'(M_1,\mbb{M}_2)}$.
                    \begin{align*}
                        & \encc{z}{\sff{send}'(M_1,\mbb{M}_2)}
                        \\
                        {}={} & \nu{ab}(a(c,\_) \sdot \encc{c}{M_1} \| \nu{de}(\encc{d}{\mbb{M}_2} \| \nu{fg}(e[b,f] \| g \fwd z)))
                    \end{align*}
                    The reduction $\encc{z}{\phi\,(\mcl{R}[\sff{send"}(M_1,\mbb{M}_2)])} \reddL Q_0$ would be due to a synchronization between the output on $e$ and an input on $d$ in $\encc{d}{\mbb{M}_2}$.
                    However, by typability, there is no $\term{\mbb{M}_2}$ for which the encoding has an unguarded input on $d$, so this case does not cause the reduction.
            \end{itemize}

        \item
            Case $\term{C} = \term{C_1 \prl C_2}$.
            The encoding of this configuration depends on which of $\term{C_1}$ and $\term{C_2}$ is a child thread; w.l.o.g., we assume $\term{C_2}$ is a child thread.
            \begin{align*}
                & \encc{z}{C_1 \prl C_2}
                \\
                {}={} & \encc{z}{C_1} \| \nunil \encc{\_}{C_2}
            \end{align*}
            There are two possible causes for the reduction $\encc{z}{C_1 \prl C_2} \reddL Q_0$: $\encc{z}{C_1}$ or $\encc{\_}{C_2}$ reduces internally.
            Either reduction remains available when the other reduction takes place, so we can postpone one reduction and focus on the other.
            As a representative case, we focus on the reduction of $\encc{z}{C_1}$.
            If $\encc{z}{C_1} \reddL Q'_1$, by \ih{1}, there exists $\term{D_1}$ such that $\term{C_1} \reddC^+ \term{D_1}$ and $Q'_1 \reddL^\ast \encc{z}{D_1}$.
            Let $\term{D} = \term{\nu{x\bfr{\epsilon}y}D'}$.
            \begin{align*}
                & \encc{z}{C_1 \prl C_2}
                \\
                {}\reddL{} & Q'_1 \| \nunil \encc{\_}{C_2}
                \\
                {}\reddL^\ast{} & \encc{z}{D_1} \| \nunil \encc{\_}{C_2}
                \\
                {}={} & \encc{z}{D_1 \prl C_2}
                \\
                {}={} & \encc{z}{D}
            \end{align*}
            Indeed, $\term{C} \reddC^+ \term{D}$.

        \item
            Case $\term{C} = \term{\nu{x\bfr{\vec{m}}y}C'}$.
            This case depends on the contents of $\term{\vec{m}}$, so we apply induction on the size of $\term{\vec{m}}$ (\ih{3}).
            \begin{itemize}
                \item
                    Case $\term{\vec{m}} = \term{\epsilon}$.
                    \begin{align*}
                        & \encc{z}{\nu{x\bfr{\epsilon}y}C'}
                        \\
                        {}={} & \nu{xy}\encc{z}{C'}
                    \end{align*}
                    There are nine possible causes for the reduction $\encc{z}{\nu{x\bfr{\epsilon}y}C'} \reddL Q_0$: $\encc{z}{C'}$ reduces internally, an output forwarded on $x$ in $\encc{z}{C'}$ synchronizes with an input forwarded on $y$ in $\encc{z}{C'}$ or vice versa, a selection forwarded on $x$ in $\encc{z}{C'}$ synchronizes with a case forwarded on $y$ in $\encc{z}{C'}$ or vice versa, or a forwarder of the continuation endpoint of an output or selection forwarded on $x$ or $y$ in $\encc{z}{C'}$ reduces.
                    \begin{itemize}
                        \item
                            $\encc{z}{C'}$ reduces internally, i.e.\ $\encc{z}{C'} \reddL Q'$.
                            By \ih{1}, there exists $\term{D'}$ such that $\term{C'} \reddC^+ \term{D'}$ and $Q' \reddL^\ast \encc{z}{D'}$.
                            Let $\term{D} = \term{\nu{x\bfr{\epsilon}y}D'}$.
                            \begin{align*}
                                & \encc{z}{\nu{x\bfr{\epsilon}y}C'}
                                \\
                                {}\reddL{} & \nu{xy}Q'
                                \\
                                {}\reddL^\ast{} & \nu{xy}\encc{z}{D'}
                                \\
                                {}={} & \encc{z}{\nu{x\bfr{\epsilon}y}D'}
                                \\
                                {}={} & \encc{z}{D}
                            \end{align*}
                            Indeed, $\term{C} \reddC^+ \term{D}$.

                        \item
                            An output forwarded on $x$ in $\encc{z}{C'}$ synchronizes with an input forwarded on $y$ in $\encc{z}{C'}$.
                            By typability, the only possibility for this reduction to occur is if
                            \begin{align*}
                                \term{C'} = \term{\mcl{G}\Big[\mcl{G}_1\big[\mcl{F}_1[\sff{send}'(M,x)]\big] \prl \mcl{G}_2\big[\mcl{F}_2[\sff{recv}~y]\big]\Big]}.
                            \end{align*}
                            Let $\term{D'} = \term{\nu{x\bfr{\epsilon}y}\mcl{G}\Big[\mcl{G}_1\big[\mcl{F}_1[x]\big] \prl \mcl{G}_2\big[\mcl{F}_2[(M,y)]\big]\Big]}$.
                            By induction on the structures of $\term{\mcl{G}}$, $\term{\mcl{G}_1}$, $\term{\mcl{G}_2}$, $\term{\mcl{F}_1}$, and $\term{\mcl{F}_2}$, we show that $\encc{z}{C} \reddL \encc{z}{D'}$.
                            For the base case, w.l.o.g.\ we assume that the send is in the main thread and that the receive is in a child thread.
                            Note that, by typability, the base case $\term{\mcl{F}_2} = \term{\child\,[]}$ is invalid: instead, we use $\term{\mcl{F}_2} = \term{\child\,(\sff{let}\,(v,w) = []\,\sff{in\,}())}$.
                            For the rest of the contexts, the base case is $\term{\mcl{G}} = \term{\mcl{G}_1} = \term{\mcl{G}_2} = \term{[]}$, and $\term{\mcl{F}_1} = \term{\main\,[]}$.
                            We have the following:
                            \begin{align*}
                                & \encc{z}{\nu{x\bfr{\epsilon}y}(\main\,(\sff{send}'(M,x)) \prl \child\,(\sff{let}\,(v,w) = (\sff{recv}~y)\,\sff{in}\,()))}
                                \\
                                {}={} & \nu{xy}(\nu{ab}(a(c,\_) \sdot \encc{c}{M} \| \nu{de}(x \fwd d \| \nu{fg}(e[b,f] \| g \fwd z)))\\
                                &{}\| \nunil \nu{a'b'}(\nu{hl}(y \fwd h \| l(n,o) \sdot \nu{pq}(a'[n,p] \| q(r,\_) \sdot o \fwd r))\\
                                &{}\| b'(c',d') \sdot \nu{e'v}\nu{f'y}(\nunil c'[e',\_] \| \nunil d'[f',\_] \| \0)))
                                \\
                                {}\reddL{} & \nu{ab}\nu{fg}(a(c,\_) \sdot \encc{c}{M} \| g \fwd z \| \nunil \nu{a'b'}(\nu{pq}(a'[b,p] \| q(r,\_) \sdot f \fwd r)\\
                                &{}\| b'(c',d') \sdot \nu{e'v}\nu{f'y}(\nunil c'[e',\_] \| \nunil d'[f',\_] \| \0)))
                                \\
                                {}\equiv{} & \nu{gf}(g \fwd z \| \nunil \nu{a'b'}(\nu{ba}\nu{pq}(a'[b,p] \| a(c,\_) \sdot \encc{c}{M} \| q(r,\_) \sdot f \fwd r)\\
                                &{}\| b'(c',d') \sdot \nu{e'v}\nu{f'y}(\nunil c'[e',\_] \| \nunil d'[f',\_] \| \0)))
                                \\
                                {}\equiv{} & \nu{xy}(x \fwd z \| \nunil \nu{a'b'}(\nu{ba}\nu{pq}(a'[b,p] \| a(c,\_) \sdot \encc{c}{M} \| q(r,\_) \sdot y \fwd r)\\
                                &{}\| b'(c',d') \sdot \nu{e'v}\nu{f'y}(\nunil c'[e',\_] \| \nunil d'[f',\_] \| \0)))
                                \\
                                {}={} & \encc{z}{\nu{x\bfr{\epsilon}y}(\main\,x \prl \child\,(\sff{let}\,(v,w) = (M,y)\,\sff{in}\,()))}
                                \\
                                {}={} & \encc{z}{D'}
                            \end{align*}
                            Indeed,
                            \begin{align*}
                                & \term{\nu{x\bfr{\epsilon}y}(\main\,(\sff{send}'(M,x)) \prl \child\,(\sff{let}\,(v,w) = (\sff{recv}~y)\,\sff{in}\,()))}
                                \\
                                {}\equivC{} & \term{\nu{x\bfr{M}y}(\main\,x \prl \child\,(\sff{let}\,(v,w) = (\sff{recv}~y)\,\sff{in}\,()))}
                                \\
                                {}\reddC{} & \term{D'}.
                            \end{align*}
                            All inductive cases are trivial by the IH, except explicit substitution in $\term{\mcl{F}_1}$ and $\term{\mcl{F}_2}$: in this case, we can extrude the substitution's scope using structural congruence in both $\term{C'}$ and $\encc{z}{C'}$ (cf.\ \Cref{t:transTermSC,t:transConfSC}).
                            Since $\encc{z}{D'} \reddL^{k-1} Q$, the thesis follows by \ih{2}.

                        \item
                            An output forwarded on $y$ in $\encc{z}{C'}$ synchronizes with an input forwarded on $x$ in $\encc{z}{C'}$.
                            Since $\term{\nu{x\bfr{\epsilon}y}C'} \equivC \term{\nu{y\bfr{\epsilon}x}C'}$, this case is analogous to the case above.

                        \item
                            A selection forwarded on $x$ in $\encc{z}{C'}$ synchronizes with a case forwarded on $y$ in $\encc{z}{C'}$.
                            By typability, the only possibility for this reduction to occur is if
                            \begin{align*}
                                \term{C'} = \term{\mcl{G}\Big[\mcl{G}_1\big[\mcl{F}_1[\sff{select}\,j\,x]\big] \prl \mcl{G}_2\big[\mcl{F}_2[\sff{case}\,y\,\sff{of}\,\{i:M_i\}_{i \in I}]\big]\Big]}
                            \end{align*}
                            where $j \in I$.
                            Let $\term{D'} = \term{\nu{x\bfr{\epsilon}y}\mcl{G}\Big[\mcl{G}_1\big[\mcl{F}_1[x]\big] \prl \mcl{G}_2\big[\mcl{F}_2[M_j~y]\big]\Big]}$.
                            This case is very similar to the case above with forwarded output and input, so we show, by similar induction, that $\encc{z}{C} \reddL \encc{z}{D'}$.
                            In the base case, typability requires that $\term{\mcl{F}_2} = \term{\child\,[]}$ and for every $i \in I$, $\term{M_i} = \term{\lambda v \sdot ()}$.
                            We have the following:
                            \begin{align*}
                                & \encc{z}{\nu{x\bfr{\epsilon}y}(\main\,(\sff{select}\,j\,x) \prl \child\,(\sff{case}\,y\,\sff{of}\,\{i:\lambda v \sdot ()\}_{i \in I}))}
                                \\
                                {}={} & \nu{xy}(\nu{ab}(x \fwd a \| \nu{cd}(b[c] \puts j \| d \fwd z))\\
                                &{}\| \nunil \nu{a'b'}(y \fwd a' \| b'(c') \gets \{i:\encc{\_}{(\lambda v \sdot ())~c'}\}_{i \in I}))
                                \\
                                {}\reddL{} & \nu{cd}(d \fwd z \| \nunil \encc{\_}{(\lambda v \sdot ())~c})
                                \\
                                {}\equiv{} & \nu{xy}(x \fwd z \| \nunil \encc{\_}{(\lambda v \sdot ())~y})
                                \\
                                {}={} & \encc{z}{\nu{x\bfr{\epsilon}y}(\main\,x \prl \child\,((\lambda v \sdot ())~y))}
                                \\
                                {}={} & \encc{z}{D'}
                            \end{align*}
                            Indeed,
                            \begin{align*}
                                & \term{\nu{x\bfr{\epsilon}y}(\main\,(\sff{select}\,j\,x) \prl \child\,(\sff{case}\,y\,\sff{of}\,\{i:\lambda v \sdot ()\}_{i \in I}))}
                                \\
                                {}\equivC{} & \term{\nu{x\bfr{j}y}(\main\,x \prl \child\,(\sff{case}\,y\,\sff{of}\,\{i:\lambda v \sdot ()\}_{i \in I}))}
                                \\
                                {}\reddC{} & \term{D'}.
                            \end{align*}
                            As in the prior case, the inductive cases are straightforward.
                            Since $\encc{z}{D'} \reddL^{k-1} Q$, the thesis follows from \ih{2}.

                        \item
                            A selection forwarded on $y$ in $\encc{z}{C'}$ synchronizes with a case forwarded on $x$ in $\encc{z}{C'}$.
                            Since $\term{\nu{x\bfr{\epsilon}y}C'} \equivC \term{\nu{y\bfr{\epsilon}x}C'}$, this case is analogous to the case above.

                        \item
                            A forwarder of the continuation endpoint of an output forwarded on $x$ in $\encc{z}{C'}$ reduces.
                            This case occurs when there is a $\term{\sff{send}'}$ on $\term{x}$ as part of an explicit substitution (again, omitting explicit substitutions).
                            This explicit substitution may occur on the level of terms
                            \begin{align*}
                                \term{C'} = \term{\mcl{G}\Big[\mcl{F}\big[\mbb{M}\xsub{ \sff{send}'(N,x)/w }\big]\Big]},
                            \end{align*}
                            or on the level of configurations
                            \begin{align*}
                                \term{C'} = \term{\mcl{G}\big[C''\xsub{ \sff{send}'(N,x)/w }\big]}.
                            \end{align*}
                            Both cases are structurally congruent, and so are their encodings (cf.\ \Cref{t:transConfSC}), so we focus on the former.
                            Let $\term{D'} = \term{\nu{x\bfr{N}y}\mcl{G}\big[\mcl{F}[\mbb{M}\{x/w\}]\big]}$.
                            By induction on the structures of $\term{\mcl{G}}$, $\term{\mcl{F}}$, and $\term{\mcl{R}}$, we show that $\encc{z}{C} \reddL \encc{z}{D'}$.
                            In the base case ($\term{\mcl{G}} = \term{\mcl{R}} = \term{[]}$ and $\term{\mcl{F}} = \term{\phi\,[]}$), w.l.o.g.\ assuming that $\term{\phi} = \term{\main}$, we have the following:
                            \begin{align*}
                                & \encc{z}{\nu{x\bfr{\epsilon}y}(\main\,(\mbb{M}\xsub{ \sff{send}'(N,x)/w }))}
                                \\
                                {}={} & \nu{xy}\nuf{wa}(\encc{z}{\mbb{M}} \| \nu{bc}(b(d,\_) \sdot \encc{d}{N} \| \nu{ef}(x \fwd e \| \nu{gh}(f[c,g] \| h \fwd a))))
                                \\
                                {}\reddL{} & \nu{xy}\nu{hg}(\encc{z}{\mbb{M}}\{h/w\} \| \nu{bc}(b(d,\_) \sdot \encc{d}{N} \| \nu{ef}(x \fwd e \| f[c,g])))
                                \\
                                {}\equiv{} & \nu{xy}\nu{cb}\nu{gh}(\nu{ef}(x \fwd e \| f[c,g]) \| b(d,\_) \sdot \encc{d}{N} \| \encc{z}{\mbb{M}}\{h/w\})
                                \\
                                {}={} & \nu{xy}\nu{cb}\nu{gh}(\nu{ef}(x \fwd e \| f[c,g]) \| b(d,\_) \sdot \encc{d}{N} \| (\encc{z}{\mbb{M}}\{x/w\})\{h/x\})
                                \\
                                {}={} & \nu{xy}\nu{cb}\nu{gh}(\nu{ef}(x \fwd e \| f[c,g]) \| b(d,\_) \sdot \encc{d}{N} \| \encc{z}{\mbb{M}\{x/w\}}\{h/x\})
                                &\text{(\encpropref{i:transSubst})}
                                \\
                                {}={} & \encc{z}{\nu{x\bfr{N}y}(\main\,(\mbb{M}\{x/w\}))}
                                \\
                                {}={} & \encc{z}{D'}
                            \end{align*}
                            Indeed,
                            \begin{align*}
                                & \term{\nu{x\bfr{\epsilon}y}(\main\,(\mbb{M}\xsub{ \sff{send}'(N,x)/w }))}
                                \\
                                {}\equivC{} & \term{\nu{x\bfr{N}y}(\main\,(\mbb{M}\xsub{ x/w }))}
                                \\
                                {}\reddC{} & \term{D'}.
                            \end{align*}
                            The inductive cases are straightforward.
                            Since $\encc{z}{D'} \reddL^{k-1} Q$, the thesis follows by \ih{2}.

                        \item
                            A forwarder of the continuation endpoint of an output forwarded on $y$ in $\encc{z}{C'}$ reduces.
                            Since $\term{\nu{x\bfr{\epsilon}y}C'} \equivC \term{\nu{y\bfr{\epsilon}x}C'}$, this case is analogous to the case above.

                        \item
                            A forwarder of the continuation endpoint of a selection forwarded on $x$ or $y$ in $\encc{z}{C'}$ reduces.
                            These cases are largely analogous to the two cases above.
                    \end{itemize}

                \item
                    Case $\term{\vec{m}} = \term{\vec{m}',M}$.
                    \begin{align*}
                        & \encc{z}{\nu{x\bfr{\vec{m}',M}y}C'}
                        \\
                        {}={} & \nu{xy}\nu{ab}(a \fwd x \| \nu{cd}\nu{ef}(b[c,e] \| d(g,\_) \sdot \encc{g}{M} \| \bencb{f}{\encc{z}{C'}\{f/x\}}{\bfr{\vec{m}'}}))
                    \end{align*}
                    There are six possible causes for the reduction $\encc{z}{\nu{x\bfr{\vec{m}',M}y}C'} \reddL Q_0$: there is an internal reduction in $\encc{z}{C'}\{f/x\}$, there is a synchronization between the output forwarded on $x$ and an input forwarded on $y$ in $\encc{z}{C'}\{f/x\}$, or a forwarder of the continuation endpoint of an output or selection forwarded on $f$ or $y$ in $\encc{z}{C'}\{f/x\}$ reduces.
                    \begin{itemize}
                        \item
                            There is an internal reduction in $\encc{z}{C'}\{f/x\}$, i.e.\ $\encc{z}{C'}\{f/x\} \reddL Q'$.
                            By \ih{1}, there exists $\term{D'}$ such that $\term{C'} \reddC^+ \term{D'}$ and $Q' \reddL^\ast \encc{z}{D'}$.
                            Indeed, $\term{\nu{x\bfr{\vec{m}',M}y}C'} \reddC^+ \term{\nu{x\bfr{\vec{m}',M}y}D'}$.

                        \item
                            There is a synchronization between the output forwarded on $x$ and an input forwarded on $y$ in $\encc{z}{C'}\{f/x\}$.
                            By typability, this case only occur if $\term{C'} = \term{\mcl{G}\big[\mcl{F}[\sff{recv}~y]\big]}$.
                            Let $\term{D'} = \term{\nu{x\bfr{\vec{m}'}y}(\main\,(M,y))}$.
                            By induction on the structures of $\term{\mcl{G}}$ and $\term{\mcl{F}}$, we show that $\encc{z}{C} \reddL \encc{z}{D'}$.
                            We consider the base case: $\term{\mcl{G}} = \term{[]}$ and $\term{\mcl{F}} = \term{\phi\,[]}$, w.l.o.g.\ assuming that $\term{\phi} = \term{\main}$.
                            Note that we are simplifying the situation: the actual possible structure of $\term{\mcl{F}}$ depends on the contents of $\vec{m}$.
                            However, all possible reductions originating from the encoding of $\term{\mcl{F}}$ are internal to the encoding of $\term{C'}$, which we handle separately.
                            \begin{align*}
                                & \encc{z}{\nu{x\bfr{\vec{m}',M}y}(\main\,(\sff{recv}~y))'}
                                \\
                                {}={} & \nu{xy}\nu{ab}(a \fwd x \| \nu{cd}\nu{ef}(b[c,e] \| d(g,\_) \sdot \encc{g}{M}\\
                                &{}\| \bencb{f}{\nu{a'b'}(y \fwd a' \| b'(c',d') \sdot \nu{e'f'}(z[c',e'] \| f'(g',\_) \sdot d' \fwd g'))}{\bfr{\vec{m}'}}))
                                \\
                                {}\equiv{} & \nu{xy}(\nu{ab}(a \fwd x \| \nu{cd}\nu{ef}(b[c,e] \| d(g,\_) \sdot \encc{g}{M} \| \bencb{f}{\0}{\bfr{\vec{m}'}}))\\
                                &{}\| \nu{a'b'}(y \fwd a' \| b'(c',d') \sdot \nu{e'f'}(z[c',e'] \| f'(g',\_) \sdot d' \fwd g')))
                                &\text{(\encpropref{i:transBufCtx})}
                                \\
                                {}\reddL{} & \nu{cd}\nu{ef}(d(g,\_) \sdot \encc{g}{M} \| \bencb{f}{\0}{\bfr{\vec{m}'}} \| \nu{e'f'}(z[c,e'] \| f'(g',\_) \sdot e \fwd g'))
                                \\
                                {}\equiv{} & \nu{fe}(\bencb{f}{\0}{\bfr{\vec{m}'}} \| \nu{cd}\nu{e'f'}(z[c,e'] \| d(g,\_) \sdot \encc{g}{M} \| f'(g',\_) \sdot e \fwd g'))
                                \\
                                {}\equiv{} & \nu{xy}(\bencb{x}{\0}{\bfr{\vec{m}'}} \| \nu{cd}\nu{e'f'}(z[c,e'] \| d(g,\_) \sdot \encc{g}{M} \| f'(g',\_) \sdot y \fwd g'))
                                \\
                                {}\equiv{} & \nu{xy}\bencb{x}{\nu{cd}\nu{e'f'}(z[c,e'] \| d(g,\_) \sdot \encc{g}{M} \| f'(g',\_) \sdot y \fwd g')}{\bfr{\vec{m}'}}
                                &\text{(\encpropref{i:transBufCtx})}
                                \\
                                {}={} & \encc{z}{\nu{x\bfr{\vec{m}'}y}(\main\,(M,y))}
                                \\
                                {}={} & \encc{z}{D'}
                            \end{align*}
                            Indeed,
                            \begin{align*}
                                & \term{\nu{x\bfr{\vec{m}',M}y}(\main\,(\sff{recv}~y))}
                                \\
                                {}\equivC{} & \term{\nu{x\bfr{\vec{m}',M}y}(\main\,(\sff{recv}~y) \prl \child\,())}
                                \\
                                {}\reddC{} & \term{\nu{x\bfr{\vec{m}'}y}(\main\,(M,y) \prl \child\,())}
                                \\
                                {}\equivC{} & \term{D'}.
                            \end{align*}
                            The inductive cases are straightforward.
                            Since $\encc{z}{D'} \reddL^{k-1} Q$, the thesis follows by \ih{2}.

                        \item
                            A forwarder of the continuation endpoint of an output forwarded on $f$ in $\encc{z}{C'}\{f/x\}$ reduces.
                            This case occurs when there is a $\term{\sff{send}'}$ on $\term{f}$ as part of an explicit substitution (again, omitting explicit substitutions).
                            Note that, by \encpropref{i:transSubst}, $\encc{z}{C'}\{f/x\} = \encc{z}{C'\{f/x\}}$.
                            As in the similar case for when $\term{\vec{m}} = \term{\epsilon}$, we focus on the case where
                            \begin{align*}
                                \term{C'\{f/x\}} = \term{\mcl{G}\Big[\mcl{F}\big[\mbb{L}\xsub{ \sff{send}'(N,f)/w }\big]\Big]}.
                            \end{align*}
                            Let $\term{D'} = \term{\nu{x\bfr{N,\vec{m}',M}y}\mcl{G}\big[\mcl{F}[\mbb{L}\{x/w\}]\big]}$.
                            By induction on the structures of $\term{\mcl{G}}$, $\term{\mcl{F}}$, and $\term{\mcl{R}}$, we show that $\encc{z}{C} \reddL \encc{z}{D'}$.
                            We consider the base case ($\term{\mcl{G}} = \term{\mcl{R}} = \term{[]}$ and $\term{\mcl{F}} = \term{\phi\,[]}$), w.l.o.g.\ assuming that $\term{\phi} = \term{\main}$.
                            \begin{align*}
                                & \encc{z}{\nu{x\bfr{\vec{m}',M}y}(\main\,(\mbb{L}\xsub{ \sff{send}'(N,x)/w }))}
                                \\
                                {}={} & \nu{xy}\nu{ab}(a \fwd x \| \nu{cd}\nu{ef}(b[c,e] \| d(g,\_) \sdot \encc{g}{M}\\
                                &{}\| \bencb{f}{\nuf{wh}(\encc{z}{\mbb{L}} \| \nu{a'b'}(a'(c',\_) \sdot \encc{c'}{N} \| \nu{d'e'}(f \fwd d' \| \nu{f'g'}(e'[b',f'] \| g' \fwd h))))}{\bfr{\vec{m}'}}))
                                \\
                                {}\reddL{} & \nu{xy}\nu{ab}(a \fwd x \| \nu{cd}\nu{ef}(b[c,e] \| d(g,\_) \sdot \encc{g}{M}\\
                                &{}\| \bencb{f}{\nu{g'f'}(\encc{z}{\mbb{L}}\{g'/w\} \| \nu{a'b'}(a'(c',\_) \sdot \encc{c'}{N} \| \nu{d'e'}(f \fwd d' \| e'[b',f'])))}{\bfr{\vec{m}'}}))
                                &\text{(\encpropref{i:transBufRed})}
                                \\
                                {}\equiv{} & \nu{xy}\nu{ab}(a \fwd x \| \nu{cd}\nu{ef}(b[c,e] \| d(g,\_) \sdot \encc{g}{M}\\
                                &{}\| \bencb{f}{\nu{a'b'}\nu{f'g'}(\nu{d'e'}(f \fwd d' \| e'[b',f']) \| a'(c',\_) \sdot \encc{c'}{N} \| \encc{z}{\mbb{L}}\{g'/w\})}{\bfr{\vec{m}'}}))
                                \\
                                {}={} & \nu{xy}\nu{ab}(a \fwd x \| \nu{cd}\nu{ef}(b[c,e] \| d(g,\_) \sdot \encc{g}{M}\\
                                &{}\| \bencb{f}{\nu{a'b'}\nu{f'g'}(\nu{d'e'}(f \fwd d' \| e'[b',f']) \| a'(c',\_) \sdot \encc{c'}{N} \| (\encc{z}{\mbb{L}}\{f/w\})\{g'/f\})}{\bfr{\vec{m}'}}))
                                \\
                                {}={} & \nu{xy}\nu{ab}(a \fwd x \| \nu{cd}\nu{ef}(b[c,e] \| d(g,\_) \sdot \encc{g}{M}\\
                                &{}\| \bencb{f}{\bencb{f}{\encc{z}{\mbb{L}}\{f/w\}}{\bfr{N}}}{\bfr{\vec{m}'}}))
                                \\
                                {}={} & \nu{xy}\nu{ab}(a \fwd x \| \nu{cd}\nu{ef}(b[c,e] \| d(g,\_) \sdot \encc{g}{M}\\
                                &{}\| \bencb{f}{\encc{z}{\mbb{L}}\{f/w\}}{\bfr{N,\vec{m}'}}))
                                \\
                                {}={} & \nu{xy}\bencb{x}{\encc{z}{\mbb{L}}\{x/w\}}{\bfr{N,\vec{m}',M}}
                                \\
                                {}={} & \nu{xy}\bencb{x}{\encc{z}{\mbb{L}\{x/w\}}}{\bfr{N,\vec{m}',M}}
                                &\text{(\encpropref{i:transSubst})}
                                \\
                                {}={} & \nu{xy}\bencb{x}{\encc{z}{\main\,(\mbb{L}\{x/w\})}}{\bfr{N,\vec{m}',M}}
                                \\
                                {}={} & \encc{z}{\nu{x\bfr{N,\vec{m}',M}y}(\main\,(\mbb{L}\{x/w\}))}
                                \\
                                {}={} & \encc{z}{D'}
                            \end{align*}
                            Indeed,
                            \begin{align*}
                                & \term{\nu{x\bfr{\vec{m}',M}y}(\main\,(\mbb{L}\xsub{ \sff{send}'(N,x)/w }))}
                                \\
                                {}\equivC{} & \term{\nu{x\bfr{N,\vec{m}',M}y}(\main\,(\mbb{L}\xsub{ x/w }))}
                                \\
                                {}\reddC{} & \term{D'}
                            \end{align*}
                            The inductive cases are straightforward.
                            Since $\encc{z}{D'} \reddL^{k-1} Q$, the thesis follow by \ih{2}.

                        \item
                            A forwarder of the continuation endpoint of a selection forwarded on $f$ in $\encc{z}{C'}\{f/x\}$ reduces.
                            This case is largely analogous to the case above.

                        \item
                            A forwarder of the continuation endpoint of an output or selection forwarded on $y$ in $\encc{z}{C'}\{f/x\}$ reduces.
                            By typability, $\term{C'}$ can only specify an input on $y$, so these cases do not occur.
                    \end{itemize}

                \item
                    Case $\term{\vec{m}} = \term{\vec{m}',\ell}$.
                    This case is largely analogous to the case above.
            \end{itemize}

        \item
            Case $\term{C} = \term{C'\xsub{ \mbb{M}/x }}$.
            By typability, there are two cases for $\term{C'}$: $\term{C'} = \term{\mcl{G}[\phi\,\mbb{M}']}$ where $\term{x} \notin \fn(\term{\mcl{G}})$, or $\term{C'} = \term{\mcl{G}\big[C''\xsub{ \mbb{M}'/y }\big]}$ where $\term{x} \notin \fn(\term{\mcl{G}}) \cup \fn(\term{C''})$.
            In both cases, $\term{C} \equivC \term{\mcl{G}'\Big[\phi\,(\mcl{R}\big[\mbb{M}''\xsub{ \mbb{M}/x }\big])\Big]}$, and, by \Cref{t:transConfSC}, the encodings of these terms are structurally congruent and thus show the same reductions.
            \begin{align*}
                \encc{z}{C'\xsub{ \mbb{M}/x }} \equiv \encc{z}{\mcl{G}'\Big[\phi\,(\mcl{R}\big[\mbb{M}''\xsub{ \mbb{M}/x }\big])\Big]}
            \end{align*}
            We apply induction on the structures of $\term{\mcl{G}'}$ and $\term{\mcl{R}}$.
            We consider the base case ($\term{\mcl{G}'} = \term{[]}$ and $\term{\mcl{R}} = \term{[]}$): $\encc{z}{C} \equiv \encc{z}{\phi\,(\mbb{M}''\xsub{ \mbb{M}/x })}$.
            Clearly, this case is analogous to the corresponding case above for $\term{C} = \term{\phi\,\mbb{M}}$.
            Since the encodings of contexts only place their contents under restriction and parallel composition, the inductive cases follow from the IH straightforwardly.
    \end{itemize}
    This concludes the proof.
\end{proof}

\section{Deadlock-free \ourGV: Full Proofs}
\label{a:ourGVdlFree}

\paragraph*{Liveness}

\begin{definition}[Active Names]\label{d:an}
    The \emph{set of active names} of $P$, denoted $\an(P)$, contains the (free) names that are used for unguarded actions (output, input, selection, branching):
    \begin{align*}
        \an(x[y,z])
        & := \{x\}
        &
        \an(x(y,z) \sdot P)
        & := \{x\}
        &
        \an(\0)
        & := \emptyset
        \\
        \an(x[z] \triangleleft j)
        & := \{x\}
        &
        \an(x(z) \triangleright \{i: P_i\}_{i \in I})
        & := \{x\}
        &
        \an(x \fwd y)
        & := \emptyset
        \\
        \an(P \| Q)
        & := \an(P) \cup \an(Q)
        &
        \an(\mu X(\tilde{x}) \sdot P)
        & := \an(P)
        \\
        \an(\nu{x y}P)
        & := \an(P) \setminus \{x,y\}
        &
        \an(X\call{\tilde{x}})
        & := \emptyset
    \end{align*}
\end{definition}

\begin{definition}[Live Process]\label{d:live}
    A process $P$ is \emph{live}, denoted $\live(P)$, if
    \begin{enumerate}
        \item there are names $x,y$ and process $P'$ such that $P \equiv \nu{x y}P'$ with $x,y \in \an(P')$, or
        \item there are names $x,y,z$ and process $P'$ such that $P \equiv \mcl{E}[\nu{yz}(x \fwd y \| P')]$.
    \end{enumerate}
\end{definition}

\begin{lemma}
\label{l:reddSoLive}
    If $P \redd Q$, then $\live(P)$.
\end{lemma}

\begin{proof}
Immediate from \Cref{d:live}.
\end{proof}

\begin{lemma}
\label{l:transWT}
    If $\type{\Gamma} \vdashC{\phi} \term{C}:\type{T}$ and $\encc{z}{C} \vdash \Gamma'$, then $\Gamma' = \ol{\enct{\Gamma}}, z:\enct{T}$.
\end{lemma}

\begin{proof}
    By \Cref{t:transTypePres} (type preservations for the encoding), $\encc{z}{C} \vdashAst \ol{\enct{\Gamma}}, z:\enct{T}$.
    Hence, the only free endpoints of $\encc{z}{C}$ are $z$ and the endpoints in $\ol{\enct{\Gamma}}$ (except endpoints of type $\bullet$).
    By \Cref{d:vdashAst}, the typing rules under `$\vdash$' and `$\vdashAst$' are nearly identical, so it is impossible to assign different types to endpoints in derivations of the same process under `$\vdash$' and `$\vdashAst$'.
    The same holds for recursion variables.
    Hence, it must be that $\Gamma' = \ol{\enct{\Gamma}}, z:\enct{T}$.
\end{proof}

The following result ensures that the translation of \ourGV terms and configurations to APCP processes is compositional with respect to reduction contexts.

\begin{lemma}
\label{l:transCtxts}
    Consider reduction contexts $\term{\mcl{R}},\term{\mcl{F}},\term{\mcl{G}}$ for \ourGV and $\mcl{E}$ for APCP, respectively (cf.\ \Cref{d:redCtx}). We have:
    \begin{itemize}
        \item If $\type{\Gamma} \vdashM \term{\mcl{R}[\mbb{M}]}: \type{T}$, then $\encc{z}{\mcl{R}[\mbb{M}]} = \mcl{E}[\encc{y}{\mbb{M}}]$, for some $\mcl{E}$ and $y$.

        \item If $\type{\Gamma} \vdashC{\phi} \term{\mcl{F}[\mbb{M}]}: \type{T}$, then $\encc{z}{\mcl{F}[\mbb{M}]} = \mcl{E}[\encc{y}{\mbb{M}}]$, for some $\mcl{E}$ and $y$.

        \item If $\type{\Gamma} \vdashC{\phi} \term{\mcl{G}[C]}: \type{T}$, then $\encc{z}{\mcl{G}[C]} = \mcl{E}[\encc{y}{C}]$, for some $\mcl{E}$ and $y$.
    \end{itemize}
\end{lemma}

\thmConfTransReddL*

\begin{proof}
    By \Cref{l:transWT}, $\encc{z}{C} \vdash z:\bullet$.
    Moreover, by \Cref{l:reddSoLive}, the assumption $\encc{z}{C} \redd Q$ ensures that $\encc{z}{C}$ is live (\Cref{d:live}).
    Using liveness of $\encc{z}{C}$ we will find a $Q'$ such that $\encc{z}{C} \reddL Q'$. Note that it need not be the case that $Q' = Q$, i.e., we may find a different reduction than $\encc{z}{C} \redd Q$.

    By definition, liveness can occur in two forms, which may hold multiple times and simultaneously, namely: $\encc{z}{C} \equiv \nu{xy}P'$ with $x,y \in \an(P')$, or $\encc{z}{C} \equiv \mcl{E}[\nu{yz}(x \fwd y \| P')]$ for some $\mcl{E}$.
    We consider two cases depending on whether the former case holds or not.
    \begin{itemize}
        \item Suppose $\encc{z}{C} \equiv \nu{xy}P'$ with $x,y \in \an(P')$.
            Then $x$ and $y$ are the subjects of output / input / selection / branching in $P'$.
            By the well-typedness of $\encc{z}{C}$, the types of $x$ and $y$ in $P'$ are dual---they are the subjects of complementary actions in $P'$.
            As a representative example, assume $x$ and $y$ form an output / input pair.
            Since $x,y \in \an(P')$, the output on $x$ and the input on $y$ must occur in parallel in $P'$.
            Then $\encc{z}{C} \equiv \nu{xy}(\mcl{E}_x[x[a,b]] \| \mcl{E}_y[y(c,d) \sdot P'']) \equiv \mcl{E}'[\nu{xy}(x[a,b] \| y(c,d) \sdot P'')] \reddL \mcl{E}'[P''\subst{a/c,b/d}] = Q'$, proving the thesis.

        \item Suppose there are no $x',y'$ such that $\encc{z}{C} \equiv \nu{x'y'}P''$ with $x',y' \in \an(P'')$.
            Since $\encc{z}{C}$ is live, it must then be that $\encc{z}{C} \equiv \mcl{E}[\nu{yz}(x \fwd y \| P')]$, for some $\mcl{E}$. We distinguish two possibilities.

            The first is when the restriction on $y,z$ is annotated, i.e., $\encc{z}{C} \equiv \mcl{E}[\nuf{yz}(x \fwd y \| P')]$, and $\bcont{x,y}(\encc{z}{C})$ (\Cref{d:bcont}) holds.
            In such a case we have  $\encc{z}{C} \reddL \mcl{E}[P'\subst{x/z}] = Q'$ directly, proving the thesis.

            Otherwise, we cannot directly mimick the reduction from $\encc{z}{C}$ under $\reddL$.
            The rest of the analysis depends on how the forwarder $x \fwd y$ can originate from the translation from $\term{C}$ into $\encc{z}{C}$.
            In each case, this involves the occurrence of a variable $\term{x}$ in $\term{C}$, for which the well-typedness of $\term{C}$ permits three cases:
            \begin{itemize}
            	\item $\term{x}$ occurs under a reduction context subject to an explicit substitution;
                \item $\term{x}$ occurs as the replacement term in an explicit substitution;
                \item $\term{x}$ occurs as the channel of a $\term{\sff{send'}}$ / $\term{\sff{select}}$ / $\term{\sff{recv}}$ / $\term{\sff{case}}$ / buffered message (further referred to as \emph{communication primitive}) under a thread context.
            \end{itemize}

            Note that there may be multiple live forwarders in $\encc{z}{C}$, each of which falls under one these three cases.
            Based on this observation, we consider three possibilities for determining the desired lazy forwarder reduction to $Q'$: the first case holds, the second case holds, or neither of the first two cases holds.

            \begin{itemize}
                \item The variable $\term{x}$ occurs in $\term{C}$ under a reduction context subject to an explicit substitution, i.e.,
                    \[ \term{C} \equivC \term{\mcl{G}\Big[\mcl{F}\big[(\mcl{R}[x])\xsub{ N / x }\big]\Big]}. \]
                    Given the translation of explicit substitutions in \Cref{f:transTermShort}, by \Cref{l:transCtxts}, we then have the following:
                    \[ \encc{z}{C} \equiv \mcl{E}'\big[\nuf{xb}(\mcl{E}''[x \fwd y] \| \encc{b}{N})\big] \equiv {\mcl{E}'''[\nuf{xb}(x \fwd y \| \encc{b}{N})]} \]

                    We argue that $\bcont{x,y}(\encc{z}{C})$ holds, enabling a forwarder reduction under $\reddL$.
                    Towards a contradiction, suppose it does not hold.
                    Then $y$ would be bound to the continuation endpoint of a forwarded output or selection.
                    However, the forwarder $x \fwd y$ is translated from the occurrence of the variable $\term{x}$ under a reduction context $\term{\mcl{R}}$, where the name $y$ is the observable endpoint of the translated term.
                    None of the constructors of reduction contexts translate in such a way that this would bind the subterm's observable endpoint to the continuation endpoint of a forwarded output or selection (cf.\ \Cref{f:transTermShort}).
                    This leads to a contradiction, and so $\bcont{x,y}(\encc{z}{C})$ must hold.

                    Let $Q' = \mcl{E}'''[\encc{y}{N}]$.
                    By Rule~$\rred{\overset{\leftrightarrow}{\scc{Id}}}$ in \Cref{f:apcpLazy}, $\nuf{xb}(x \fwd y \| \encc{b}{N}) \reddL^{(x,y)} \encc{y}{N}$.
                    Then, since $\encc{z}{C} \equiv \mcl{E}'''[\nuf{xb}(x \fwd y \| \encc{b}{N})]$, by Rules~$\rred{\equiv},\rred{\onu},\rred{\|}$, $\encc{z}{C} \reddL^{(x,y)} Q'$.
                    Since $\bcont{x,y}(\encc{z}{C})$ holds, by Rule~$\rred{(x,y)}$, $\encc{z}{C} \reddL Q'$, proving the thesis.

                \item The variable $\term{x}$ occurs in $\term{C}$ as the replacement term in an explicit substitution, i.e.,
                    \[ \term{C} \equivC \term{\mcl{G}\Big[\mcl{F}\big[\mbb{N}\xsub{ x/z }\big]\Big]}. \]
                    Given the translation of explicit substitutions in \Cref{f:transTermShort}, by \Cref{l:transCtxts}, we then have the following:
                    \[ \encc{z}{C} \equiv \mcl{E}'[\nuf{zy}(\encc{w}{\mbb{N}} \| x \fwd y)] \]

                    We argue that $\bcont{x,y}(\encc{z}{C})$ holds, enabling a lazy forwarder reduction.
                    Towards a contradiction, suppose it does not hold.
                    Then $x$ would be bound to the continuation endpoint of a forwarded output or selection.
                    The forwarder $x \fwd y$ is translated from the occurrence of the variable $\term{x}$ as the replacement term in an explicit substitution.
                    In this case, since $\type{\emptyset} \vdashC{\main} \term{C} : \type{\1}$, the variable $\term{x}$ is bound somehow in $\term{C}$.
                    However, none of the well-typed ways of binding $\term{x}$ translates in such a way that $x$ would be bound to the continuation endpoint of a forwarded output or selection.
                    Here again we reach a contradiction, and conclude that $\bcont{x,y}(\encc{z}{C})$ must hold.

                    Let $Q' = \mcl{E}'[\encc{w}{\mbb{N}}\subst{x/z}]$.
                    By Rule~$\rred{\overset{\leftrightarrow}{\scc{Id}}}$ in \Cref{f:apcpLazy}, $\nuf{zy}(\encc{w}{\mbb{N}} \| x \fwd y) \reddL^{(x,y)} \encc{w}{\mbb{N}}\subst{x/z}$.
                    Then, since $\encc{z}{C} \equiv \mcl{E}'[\nuf{zy}(\encc{w}{\mbb{N}} \| x \fwd y)]$, by Rules~$\rred{\equiv},\rred{\onu},\rred{\|}$, $\encc{z}{C} \reddL^{(x,y)} Q'$.
                    Since $\bcont{x,y}(\encc{z}{C})$ holds, by Rule~$\rred{(x,y)}$, $\encc{z}{C} \reddL Q'$, proving the thesis.

                \item No variable $\term{x'}$ occurs in $\term{C}$ under a reduction context subject to an explicit substitution, or as the replacement term in an explicit substitution.
                    Since there is a live forwarder in $\encc{z}{C}$, it must then be that the variable $\term{x}$ (which translates to the live forwarder) occurs as the channel of a communication primitive under a thread context.

                    Given at least one communication primitive with a variable as channel in a thread context, take the communication primitive whose translated type has the least priority.
                    As a representative example, we consider the case where the communication primitive is a $\term{\sff{send'}}$ on variable $\term{x}$; the analysis is analogous for the other communication primitives.
                    By typability, there is a corresponding $\term{\sff{recv}}$ on some variable $\term{y}$ connected by a buffer $\term{\nu{x\bfr{\epsilon}y}}$.
                    We argue that the translation of this $\term{\sff{recv}}$ on $\term{y}$ is not blocked, i.e.\ occurs under a thread context.
                    That is, we show the following:
                    \[ \term{C} \equivC \term{\mcl{G}\big[\nu{x\bfr{\epsilon}y}(\mcl{F}[\sff{send'}(M,x)] \prl \mcl{F'}[\sff{recv}~y])\big]} \]
                    In this step, we crucially rely on the assumption that $\encc{z}{C}$ is well-typed with satisfiable priority requirements.

                    Towards a contradiction, assume that the translation of the $\term{\sff{recv}}$ on $\term{y}$ is blocked by a (forwarded) input or branch.
                    The analysis is analogous for inputs and branches, so w.l.o.g.\ assume that we are only dealing with (forwarded) inputs.
                    Since $z$ is the only possible free endpoint of $\encc{z}{C}$ and has type $\bullet$, this blocking (forwarded) input must be connected to a corresponding (forwarded) output.
                    The (forwarded) output may in turn be blocked as well.
                    Following the priorities of the involved inputs and outputs, we follow a chain of such blocking (forwarded) inputs, until we eventually find a non-blocked pair of (forwarded) input and (forwarded) output.
                    Since we have assumed that there are is no pair of active names in $\encc{z}{C}$, the pair must consist of a forwarded input and a forwarded output.
                    Then, since we have established that any live forwarders can only be translations of communication primitives, it must be that the pair is the translation of a $\term{\sff{send'}}$ and $\term{\sff{recv}}$ pair.
                    The priority of the types associated to this pair must be lower than the priority associated to the translation of the $\term{\sff{recv}}$ on $\term{y}$, for otherwise it would not be possible for the just discoverd pair to block it.
                    However, we initially took the $\term{\sff{send'}}$ with the least priority, and by duality the $\term{\sff{recv}}$ on $\term{y}$ has the same priority.
                    This contradicts that the just discovered pair has lower priority.
                    Hence, the translation of the $\term{\sff{send'}}$ on $\term{x}$ and the $\term{\sff{recv}}$ on $\term{y}$ are not blocked.

                    Thus, we have established the following:
                    \[ \term{C} \equivC \term{\mcl{G}\big[\nu{x\bfr{\epsilon}y}(\mcl{F}[\sff{send'}(M,x)] \prl \mcl{F'}[\sff{recv}~y])\big]} \]
                    Given the translations in \Cref{f:transTermShort,f:transConfBufShort}, by \Cref{l:transCtxts}, we then have the following:
                    \begin{align*}
                        \encc{z}{C} &\equiv \mcl{E}\big[\nu{xy} \begin{array}[t]{@{}l@{}}
                            (\mcl{E'}[\nu{ab}(a(c,d) \sdot \encc{c}{M} \| \nu{ef}(x \fwd e \| \nu{gh}(f[b,g] \| h \fwd q)))]
                            \\
                            {} \| \mcl{E''}[\nu{a'b'}(y \fwd a' \| b'(c',d') \sdot \nu{e'f'}(q'[c',e'] \| f'(g',h') \sdot d' \fwd g'))]) \big]
                        \end{array}
                        \\
                        &\equiv \mcl{E'''}[\nu{xy}(\nu{ef}(x \fwd e \| f[b,g]) \| \nu{a'b'}(y \fwd a' \| b'(c',d') \sdot P))]
                    \end{align*}
                    Let $Q' = \mcl{E'''}[P\subst{b/c',g/d'}]$.
                    By Rule~$\rred{\overset{\leftrightarrow}{\tensor \parr}}$ in \Cref{f:apcpLazy}, we have the following:
                    \[ \nu{xy}(\nu{ef}(x \fwd e \| f[b,g]) \| \nu{a'b'}(y \fwd a' \| b'(c',d') \sdot P)) \reddL^\cdot P\subst{b/c',g/d'} \]
                    Then, by Rules~$\rred{\equiv},\rred{\onu},\rred{\|}$, $\encc{z}{C} \reddL^\cdot Q'$.
                    Finally, by Rule~$\rred{\cdot}$, $\encc{z}{C} \reddL Q'$, proving the thesis.
                    \qedhere
            \end{itemize}
    \end{itemize}
\end{proof}

\thmOurGVdfFree*

\begin{proof}
    By \Cref{l:transWT}, $\encc{z}{C} \vdash z:\bullet$.
    Then $\nu{z\_}\encc{z}{C} \vdash \emptyset$.
    Hence, by \Cref{t:APCPdlFree} (deadlock-freedom), either $\nu{z\_}\encc{z}{C} \equiv \0$ or $\nu{z\_}\encc{z}{C} \redd Q$ for some $Q$.
    The rest of the analysis depends on which possibility holds:
    \begin{itemize}
        \item We have $\nu{z\_}\encc{z}{C} \equiv \0$.
            We straightforwardly deduce from the well-typedness and translation of $\term{C}$ that $\term{C} \equivC \term{\main\,()}$, proving the thesis.

        \item We have $\nu{z\_}\encc{z}{C} \redd Q$ for some $Q$.
            We argue that this implies that $\encc{z}{C} \redd Q_0$, for some $Q_0$.
            The analysis depends on whether this reduction involves the endpoint $z$ or not:
            \begin{itemize}
                \item The reduction involves $z$.
                    Since $z$ is of type $\bullet$ in $\encc{z}{C}$, then $z$ must occur in a forwarder, and the reduction involves a forwarder $z \fwd x$ for some $x$.
                    Since $x$ is not free in $\encc{z}{C}$, there must be a restriction on $x$.
                    Hence, the forwarder $z \fwd x$ can also reduce with this restriction, instead of with the restriction $\nu{z\_}$.
                    This means that $\encc{z}{C} \redd Q_0$ for some $Q_0$.

                \item The reduction does not involve  $z$, in which case the reduction must be internal to $\encc{z}{C}$.
                    That is, $\encc{z}{C} \redd Q_0$ for some $Q_0$ and $Q \equiv \nu{z\_}Q_0$.
            \end{itemize}

            We have just established that $\encc{z}{C} \redd Q_0$ for some $Q_0$.
            By \Cref{t:confTransReddL}, this implies that $\encc{z}{C} \reddL Q'$ for some $Q'$.
            Then, by \Cref{t:soundness} (soundness of reduction for configurations), there exists $\term{D}$ such that $\term{C} \reddC^+ \term{D}$.
            Hence, there exists $\term{D_0}$ such that $\term{C} \reddC \term{D_0}$, proving the thesis.
            \qedhere
    \end{itemize}
\end{proof}

\section{Example: Transference of Deadlock-freedom from APCP to \texorpdfstring{\Cref{x:gvRed}}{Example~\ref{x:gvRed}}}
\label{a:example}

Consider again the configuration $\term{C_1}$ from \Cref{x:gvRed}:
\[
    \term{C_1} = \term{\main\, (\sff{let}\, (f,g) = \sff{new}\, \sff{in}\,
    \sff{let}\, (h,k) = \sff{new}\, \sff{in}\,
    \sff{spawn}~\left(\begin{tarr}[c]{l}
            \sff{let}\, f' = (\sff{send}~(u,f))\, \sff{in} \\
            \quad \sff{let}\, (v',h') = (\sff{recv}~h)\, \sff{in}\, (),
            \\[4pt]
            \sff{let}\, k' = (\sff{send}~(v,k))\, \sff{in} \\
            \quad \sff{let}\, (u',g') = (\sff{recv}~g)\, \sff{in}\, ()
    \end{tarr}\right))}
\]
We verify the deadlock-freedom of $\term{C_1}$ by translating its typing derivation into APCP, and checking whether priority requirements can be satisfied.
First, the typing derivation of $\term{C_1}$ (only showing types and rules applied), where $\type{S} = \type{{!}\sff{end} \sdot \sff{end}}$:
\[
    \lambda_1 =
    \inferrule*[Right=$\sccsm{T-App}$,vcenter]
    {
        \inferrule*[Right=$\sccsm{T-Abs}$]
        {
            \inferrule*[Right=$\sccsm{T-EndL}$]
            {
                \inferrule*[Right=$\sccsm{T-Split}$]
                {
                    \inferrule*[Right=$\sccsm{T-Recv}$]
                    {
                        \inferrule*[Right=$\sccsm{T-Var}$]
                        { }
                        { \term{h}:\type{\ol{S}} \vdashM \type{\ol{S}} }
                    }
                    { \term{h}:\type{\ol{S}} \vdashM \type{\sff{end} \times \sff{end}} }
                    \\
                    \inferrule*[Right=$\sccsm{T-EndL}$]
                    {
                        \inferrule*[Right=$\sccsm{T-EndL}$]
                        {
                            \inferrule*[Right=$\sccsm{T-Unit}$]
                            { }
                            { \type{\emptyset} \vdashM \type{\1} }
                        }
                        { \term{v'}:\type{\sff{end}} \vdashM \type{\1} }
                    }
                    { \term{v'}:\type{\sff{end}}, \term{h'}:\type{\sff{end}} \vdashM \type{\1} }
                }
                { \term{h}:\type{\ol{S}} \vdashM \type{\1} }
            }
            { \term{h}:\type{\ol{S}}, \term{f'}:\type{\sff{end}} \vdashM \type{\1} }
        }
        { \term{h}:\type{\ol{S}} \vdashM \type{\sff{end} \lolli \1} }
        \\
        \inferrule*[Right=$\sccsm{T-Send}$]
        {
            \inferrule*[Right=$\sccsm{T-Pair}$]
            {
                \inferrule*[Right=$\sccsm{T-EndR}$]
                { }
                { \type{\emptyset} \vdashM \type{\sff{end}} }
                \\
                \inferrule*[Right=$\sccsm{T-Var}$]
                { }
                { \term{f}:\type{S} \vdashM \type{S} }
            }
            { \term{f}:\type{S} \vdashM \type{\sff{end} \times S} }
        }
        { \term{f}:\type{S} \vdashM \type{\sff{end}} }
    }
    { \term{f}:\type{S}, \term{h}:\type{\ol{S}} \vdashM \type{\1} }
\]
\[
    \lambda_2 =
    \inferrule*[Right=$\sccsm{T-App}$,vcenter]
    {
        \inferrule*[Right=$\sccsm{T-Abs}$]
        {
            \inferrule*[Right=$\sccsm{T-EndL}$]
            {
                \inferrule*[Right=$\sccsm{T-Split}$]
                {
                    \inferrule*[Right=$\sccsm{T-Recv}$]
                    {
                        \inferrule*[Right=$\sccsm{T-Var}$]
                        { }
                        { \term{g}:\type{\ol{S}} \vdashM \type{\ol{S}} }
                    }
                    { \term{g}:\type{\ol{S}} \vdashM \type{\sff{end} \times \sff{end}} }
                    \\
                    \inferrule*[Right=$\sccsm{T-EndL}$]
                    {
                        \inferrule*[Right=$\sccsm{T-EndL}$]
                        {
                            \inferrule*[Right=$\sccsm{T-Unit}$]
                            { }
                            { \type{\emptyset} \vdashM \type{\1} }
                        }
                        { \term{g'}:\type{\sff{end}} \vdashM \type{\1} }
                    }
                    { \term{u'}:\type{\sff{end}}, \term{g'}:\type{\sff{end}} \vdashM \type{\1} }
                }
                { \term{g}:\type{\ol{S}} \vdashM \type{\1} }
            }
            { \term{g}:\type{\ol{S}}, \term{k'}:\type{\sff{end}} \vdashM \type{\1} }
        }
        { \term{g}:\type{\ol{S}} \vdashM \type{\sff{end} \lolli \1} }
        \\
        \inferrule*[Right=$\sccsm{T-Send}$]
        {
            \inferrule*[Right=$\sccsm{T-Pair}$]
            {
                \inferrule*[Right=$\sccsm{T-EndR}$]
                { }
                { \type{\emptyset} \vdashM \type{\sff{end}} }
                \\
                \inferrule*[Right=$\sccsm{T-Var}$]
                { }
                { \term{k}:\type{S} \vdashM \type{S} }
            }
            { \term{k}:\type{S} \vdashM \type{\sff{end} \times S} }
        }
        { \term{k}:\type{S} \vdashM \type{\sff{end}} }
    }
    { \term{g}:\type{\ol{S}}, \term{k}:\type{S} \vdashM \type{\1} }
\]
\[
    \inferrule*[right=$\sccsm{T-Main}$]
    {
        \inferrule*[Right=$\sccsm{T-Split}$]
        {
            \inferrule*[Right=$\sccsm{T-New}$]
            { }
            { \type{\emptyset} \vdashM \type{S \times \ol{S}} }
            \\
            \inferrule*[Right=$\sccsm{T-Split}$]
            {
                \inferrule*[Right=$\sccsm{T-New}$]
                { }
                { \type{\emptyset} \vdashM \type{\ol{S} \times S} }
                \\
                \inferrule*[Right=$\sccsm{T-Spawn}$]
                {
                    \inferrule*[Right=$\sccsm{T-Pair}$]
                    {
                        \lambda_1
                        \\
                        \lambda_2
                    }
                    { \term{f}:\type{S}, \term{g}:\type{\ol{S}}, \term{h}:\type{\ol{S}}, \term{k}:\type{S} \vdashM \type{\1 \times \1} }
                }
                { \term{f}:\type{S}, \term{g}:\type{\ol{S}}, \term{h}:\type{\ol{S}}, \term{k}:\type{S} \vdashM \type{\1} }
            }
            { \term{f}:\type{S}, \term{g}:\type{\ol{S}} \vdashM \type{\1} }
        }
        { \type{\emptyset} \vdashM \type{\1} }
    }
    { \type{\emptyset} \vdashC{\main} \term{C_1} : \type{\1} }
\]

To verify the deadlock-freedom of $\term{C_1}$, we have to translate this typing derivation into APCP, and check that its priority requirements are satisfiable.
We avoid writing out whole typing derivations as per the translation.
Instead, we first inspect the derivations in \Cref{a:transFull}, annotate all types with priorities, and write out the priority requirements that would need to be satisfied.
For example, annotating the derivation of the translation of Rule~\scc{T-Abs}:
\begin{mathpar}
    \enc{z}{
        \inferrule[T-Abs]{
            \type{\Gamma}, \term{x}: \type{T} \vdashM \term{M}: \type{U}
        }{
            \type{\Gamma} \vdashM \term{\lambda x \sdot M}: \type{T \lolli U}
        }
    }
    = \inferrule{
        \inferrule*{
            \inferrule{
                \inferrule*{
                    \inferrule{
                        \inferrule{ }{
                            a[c,e] \vdash a: \ol{\enct{\type{T}}} \tensor^{\pri_1} \bullet, c: \enct{\type{T}}, e: \bullet
                        }
                    }{
                        a[c,e] \vdash a: \ol{\enct{\type{T}}} \tensor^{\pri_1} \bullet, c: \enct{\type{T}}, e: \bullet, f: \bullet
                    }
                }{
                    \nu{ef}a[c,e] \vdash a: \ol{\enct{\type{T}}} \tensor^{\pri_1} \bullet, c: \enct{\type{T}}
                }
                \\
                \encc{b}{M} \vdash \ol{\enct{\type{\Gamma}}}, x: \ol{\enct{\type{T}}}, b: \enct{\type{U}}
            }{
                \nu{ef}a[c,e] \| \encc{b}{M} \vdash \ol{\enct{\type{\Gamma}}}, a: \ol{\enct{\type{T}}} \tensor^{\pri_1} \bullet, b: \enct{\type{U}}, c: \enct{\type{T}}, x: \ol{\enct{\type{T}}}
            }
        }{
            \nuf{cx}(\nu{ef}a[c,e] \| \encc{b}{M}) \vdash \ol{\enct{\type{\Gamma}}}, a: \ol{\enct{\type{T}}} \tensor^{\pri_1} \bullet, b: \enct{\type{U}}
        }
        \pri_2 < \pr(\ol{\enct{\Gamma}})
    }{
        z(a,b) \sdot \nuf{cx}(\nu{ef}a[c,e] \| \encc{b}{M}) \vdash \ol{\enct{\type{\Gamma}}}, z: (\ol{\enct{\type{T}}} \tensor^{\pri_1} \bullet) \parr^{\pri_2} \enct{\type{U}}
    }
\end{mathpar}
Now we can annotate the short version of the translation of Rule~\scc{T-Abs} from \Cref{f:transTermShort} with the priority requirement:
\[
    \enc{z}{
        \inferrule[T-Abs]{
            \type{\Gamma}, \term{x}: \type{T} \vdashM \term{M}: \type{U}
        }{
            \type{\Gamma} \vdashM \term{\lambda x \sdot M}: \type{T \lolli U}
        }
    }
    =
    \inferrule*[fraction={===},vcenter]
    {
        \encc{b}{M} \vdash \ol{\enct{\type{\Gamma}}}, x: \ol{\enct{\type{T}}}, b: \enct{\type{U}}
        \\
        \pri_2 < \pr(\ol{\enct{\Gamma}})
    }
    { z(a,b) \sdot \nuf{cx}(\nu{ef}a[c,e] \| \encc{b}{M}) \vdash \ol{\enct{\type{\Gamma}}}, z: (\ol{\enct{\type{T}}} \tensor^{\pri_1} \bullet) \parr^{\pri_2} \enct{\type{U}} }
\]
We can do this for all typing rules.
It then suffices to only translate the types of a \ourGV configuration, assigning new priorities whenever new connectives are introduced, and annotating the priority requirements.
The configuration is then deadlock-free if all priority requirements are satisfiable.

Let us apply this to $\term{C_1}$.
We repeat its typing derivation, but this time translate the types to APCP, and add priorities and priority requirements.
Here, we need two versions of $\enct{S}$, as they need different priorities: $A_1 = \enct{S} = (\bullet \tensor^{\pri_1} \bullet) \parr^{\pri_2} \bullet$ and $A_2 = \enct{S} = (\bullet \tensor^{\pri_3} \bullet) \parr^{\pri_4} \bullet$.
\[
    \pi_1 =
    \inferrule*[right=$\sccsm{T-New}$]
    { \scriptstyle \pri_5,\pri_7 < \pri_2 }
    { a_1:(A_1 \parr^{\pri_5} \bullet) \tensor^{\pri_6} (\ol{A_1} \parr^{\pri_7} \bullet) }
\]
\[
    \pi_2 =
    \inferrule*[right=$\sccsm{T-New}$]
    { \scriptstyle \pri_8,\pri_{10} < \pri_4 }
    { a_2:(\ol{A_2} \parr^{\pri_8} \bullet) \tensor^{\pri_9} (A_2 \parr^{\pri_{10}} \bullet) }
\]
\[
    \pi_3 =
    \inferrule*[right=$\sccsm{T-Abs}$]
    {
        \inferrule*[Right=$\sccsm{T-EndL}$]
        {
            \inferrule*[Right=$\sccsm{T-Split}$]
            {
                \inferrule*[Right=$\sccsm{T-Recv}$]
                {
                    \inferrule*[Right=$\sccsm{T-Var}$]
                    { }
                    { h:A_2, a_6:\ol{A_2} }
                    \\
                    {\scriptstyle \pri_4 < \pri_{17} }
                }
                { h:A_2, a_5:(\bullet \parr^{\pri_{16}} \bullet) \tensor^{\pri_{17}} (\bullet \parr^{\pri_{18}} \bullet) }
                \\
                \inferrule*[Right=$\sccsm{T-EndL}$]
                {
                    \inferrule*[Right=$\sccsm{T-EndL}$]
                    {
                        \inferrule*[Right=$\sccsm{T-Unit}$]
                        { }
                        { b_1:\bullet }
                    }
                    { v':\bullet, b_1:\bullet }
                }
                { v':\bullet, h':\bullet, b_1:\bullet }
            }
            { h:A_2, b_1:\bullet }
        }
        { h:A_2, f':\bullet, b_1:\bullet }
        \\
        {\scriptstyle \pri_{15} < \pri_4 }
    }
    { h:A_2, a_4:(\bullet \tensor^{\pri_{14}} \bullet) \parr^{\pri_{15}} \bullet }
\]
\[
    \pi_4 =
    \inferrule*[right=$\sccsm{T-App}$]
    {
        \pi_3
        \\
        \inferrule*[Right=$\sccsm{T-Send}$]
        {
            \inferrule*[Right=$\sccsm{T-Pair}$]
            {
                \inferrule*[Right=$\sccsm{T-EndR}$]
                { }
                { e_3:\bullet }
                \\
                \inferrule*[Right=$\sccsm{T-Var}$]
                { }
                { f:\ol{A_1}, g_2:A_1 }
                \\
                {\scriptstyle \pri_{21} < \pri_2 }
            }
            { f:\ol{A_1}, a_7:(\bullet \parr^{\pri_{19}} \bullet) \tensor^{\pri_{20}} (A_1 \parr^{\pri_{21}} \bullet) }
        }
        { f:\ol{A_1}, e_2:\bullet }
        \\
        {\scriptstyle \pri_{14} < \pri_2 }
    }
    { f:\ol{A_1}, h:A_2, e_1:\bullet }
\]
\[
    \pi_5 =
    \inferrule*[Right=$\sccsm{T-Abs}$]
    {
        \inferrule*[Right=$\sccsm{T-EndL}$]
        {
            \inferrule*[Right=$\sccsm{T-Split}$]
            {
                \inferrule*[Right=$\sccsm{T-Recv}$]
                {
                    \inferrule*[Right=$\sccsm{T-Var}$]
                    { }
                    { g:A_1, a_{10}:\ol{A_1} }
                    \\
                    {\scriptstyle \pri_2 < \pri_{25}}
                }
                { g:A_1, a_9:(\bullet \parr^{\pri_{24}} \bullet) \tensor^{\pri_{25}} (\bullet \parr^{\pri_{26}} \bullet) }
                \\
                \inferrule*[Right=$\sccsm{T-EndL}$]
                {
                    \inferrule*[Right=$\sccsm{T-EndL}$]
                    {
                        \inferrule*[Right=$\sccsm{T-Unit}$]
                        { }
                        { b_2:\bullet }
                    }
                    { g':\bullet, b_2:\bullet }
                }
                { u':\bullet, g':\bullet, b_2:\bullet }
            }
            { g:A_1, b_2:\bullet }
        }
        { g:A_1, k':\bullet, b_2:\bullet }
        \\
        {\scriptstyle \pri_{23} < \pri_2 }
    }
    { g:A_1, a_8:(\bullet \tensor^{\pri_{22}} \bullet) \parr^{\pri_{23}} \bullet }
\]
\[
    \pi_6 =
    \inferrule*[right=$\sccsm{T-App}$]
    {
        \pi_5
        \\
        \inferrule*[Right=$\sccsm{T-Send}$]
        {
            \inferrule*[Right=$\sccsm{T-Pair}$]
            {
                \inferrule*[Right=$\sccsm{T-EndR}$]
                { }
                { e_5:\bullet }
                \\
                \inferrule*[Right=$\sccsm{T-Var}$]
                { }
                { k:\ol{A_2}, g_3:A_2 }
                \\
                {\scriptstyle \pri_{29} < \pri_4 }
            }
            { k:\ol{A_2}, a_{11}:(\bullet \parr^{\pri_{27}} \bullet) \tensor^{\pri_{28}} (\bullet \parr^{\pri_{29}} \bullet) }
        }
        { k:\ol{A_2}, e_4:\bullet }
        \\
        {\scriptstyle \pri_{22} < \pri_4 }
    }
    { g:A_1, k:\ol{A_2}, g_1:\bullet }
\]
\[
    \inferrule*[right=$\sccsm{T-Main}$]
    {
        \inferrule*[Right=$\sccsm{T-Split}$]
        {
            \pi_1
            \\
            \inferrule*[Right=$\sccsm{T-Split}$]
            {
                \pi_2
                \\
                \inferrule*[Right=$\sccsm{T-Spawn}$]
                {
                    \inferrule*[Right=$\sccsm{T-Pair}$]
                    {
                        \pi_4
                        \\
                        \pi_6
                        \\
                        { \scriptstyle
                            \pri_{11},\pri_{13} < \pri_2,\pri_4
                        }
                    }
                    { f:\ol{A_1}, g:A_1, h:A_2, k:\ol{A_2}, a_3:(\bullet \parr^{\pri_{11}} \bullet) \tensor^{\pri_{12}} (\bullet \parr^{\pri_{13}} \bullet) }
                }
                { f:\ol{A_1}, g:A_1, h:A_2, k:\ol{A_2}, z:\bullet }
                \\
                { \scriptstyle
                    \pri_9 < \pri_2
                }
            }
            { f:\ol{A_1}, g:A_1, z:\bullet }
        }
        { z:\bullet }
    }
    { z:\bullet }
\]
In the end, we have the following priority requirements:
\begin{align*}
    \pri_5,\pri_7,\pri_9,\pri_{11},\pri_{13},\pri_{14},\pri_{21},\pri_{23} < \pri_2 < \pri_{25}
    \\
    \pri_8,\pri_{10},\pri_{11},\pri_{13},\pri_{15},\pri_{22},\pri_{29} < \pri_4 < \pri_{17}
\end{align*}
It is easy to see that these requirements are satisfiable, so $\term{C_1}$ is indeed deadlock-free.

\fi

\end{document}